\documentclass[a4paper,USenglish,thm-restate]{lipics-v2021}

\pdfoutput=1 
\hideLIPIcs  

\title{Perpetual exploration in anonymous synchronous networks with a Byzantine black hole} 

\author{Adri Bhattacharya}{Indian Institute of Technology Guwahati, Assam-781039, India}{a.bhattacharya@iitg.ac.in}{https://orcid.org/0000-0003-1517-8779}{upported by CSIR, Govt. of India, Grant Number:
09/731(0178)/2020- EMR-I}

\author{Pritam Goswami}{Sister Nivedita University, Kolkata, India}{pgoswami.cs@gmail.com}{https://orcid.org/0000-0002-0546-3894}{}

\author{Evangelos Bampas}{Universit\'e Paris-Saclay, CNRS, Laboratoire Interdisciplinaire des Sciences du Num\'erique, 91400 Orsay, France}{bampas@lisn.fr}{https://orcid.org/0000-0002-1496-9299}{}

\author{Partha Sarathi Mandal}{Indian Institute of Technology Guwahati, Assam-781039, India}{psm@iitg.ac.in}{https://orcid.org/0000-0002-8632-5767}{}

\authorrunning{A. Bhattacharya, P. Goswami, E. Bampas, and P.S. Mandal} 

\Copyright{Adri Bhattacharya, Pritam Goswami, Evangelos Bampas, and Partha Sarathi Mandal} 



\acknowledgements{We thank the DISC 2025 reviewers for their careful reading and useful feedback.}

\nolinenumbers 

\EventEditors{Hagit Attiya}
\EventNoEds{2}
\EventLongTitle{39th International Symposium on Distributed Computing}
\EventShortTitle{DISC 2025}
\EventAcronym{DISC}
\EventYear{2025}
\EventDate{December 28--30, 2025}
\EventLocation{Berlin, Germany}
\EventLogo{}
\SeriesVolume{39}
\ArticleNo{23}
\usepackage{amsfonts, amssymb, amsmath,pifont,float,caption,multicol}
\usepackage[ruled,vlined,linesnumbered]{algorithm2e}
\usepackage{wrapfig}
\usepackage[dvipsnames]{xcolor}
\usepackage{subcaption}

\newcommand{\pbmPerpExpl}{\textsc{PerpExploration-BBH}}
\newcommand{\pbmPerpExplHome}{\textsc{PerpExploration-BBH-Home}}

\newcommand{\calA}{\mathcal{A}}
\newcommand{\calB}{\mathcal{B}}
\newcommand{\calE}{\mathcal{E}}
\newcommand{\calG}{\mathcal{G}}
\newcommand{\calI}{\mathcal{I}}
\newcommand{\calO}{\mathcal{O}}
\newcommand{\calP}{\mathcal{P}}
\newcommand{\frakb}{\mathfrak{b}}

\newcommand{\cons}{\mathrm{Cons}}

\begin{document}

\maketitle

\begin{abstract}
In this paper, we investigate the following question:

``How can a group of initially co-located mobile agents perpetually explore an unknown graph, when one stationary node occasionally behaves maliciously, under the control of an adversary?''

This malicious node is termed as ``Byzantine black hole (BBH)'' and at any given round it may choose to destroy all visiting agents, or none of them.
While investigating this question, we found out that
this subtle power turns out to drastically undermine even basic exploration strategies which have been proposed in the context of a classical, always active, black hole.

We study this perpetual exploration problem in the presence of at most one BBH, without initial knowledge of the network size. Since the underlying graph may be 1-connected, perpetual exploration of the entire graph may be infeasible. Accordingly, we define two variants of the problem, termed as \pbmPerpExpl\ and \pbmPerpExplHome. In the former, the agents are tasked to perform perpetual exploration of at least one component, obtained after the exclusion of the BBH. In the latter, the agents are tasked to perform perpetual exploration of the component which contains the \emph{home} node, where agents are initially co-located. Naturally, \pbmPerpExplHome\ is a special case of \pbmPerpExpl.
The mobile agents are controlled by a synchronous scheduler, and they communicate via \textit{face-to-face} model of communication.

The main objective in this paper is to determine the minimum number of agents necessary and sufficient to solve these problems.
We first consider the problems in acyclic networks, and we obtain optimal algorithms that solve \pbmPerpExpl\ with $4$ agents, and \pbmPerpExplHome\ with $6$ agents in trees. The lower bounds hold even in path graphs. In general graphs, we give a non-trivial lower bound of $2\Delta-1$ agents for \pbmPerpExpl, and an upper bound of $3\Delta+3$ agents for \pbmPerpExplHome. To the best of our knowledge, this is the first paper that studies a variant of a black hole in  arbitrary networks, without initial topological knowledge about the network.
\end{abstract}

\section{Introduction}
In distributed mobile agent algorithms, a fundamental task is the collaborative exploration of a network by a collection of mobile agents. It was introduced and formulated by Shannon \cite{Shannon-First-Exploration} in 1951. Later, the real life applicability of this problem in fields like unmanned search and rescue, monitoring, network search, etc. has garnered a lot of interest from researchers across the world which leads them to study this problem in many different settings.
Ensuring the security of agents against threats or breaches is one of the major concerns in the design of exploration algorithms. Among the various security threats, two types have received the most attention in the literature of mobile agent algorithms so far. These are threats from \textit{malicious agents} \cite{DPP14,IntrusionDetection} and \textit{malicious hosts} \cite{CzyzowiczBHSSyncTree,DBLP:journals/dc/DobrevFPS06,DobrevAnonymousRingBHS}. In this work, we are interested in the latter. A malicious host in a network is a stationary node that can destroy any incoming agent without leaving any trace. Dobrev et al.\ introduced this type of malicious node in their 2006 paper \cite{DobrevAnonymousRingBHS}, referring to such malicious hosts as \textit{black holes}. In this paper, we will use interchangeably the terms ``classical black hole'' and ``black hole''.

There is extensive literature on \textit{Black Hole Search} (BHS) problem, that requires locating the black hole by multiple mobile agents in a network. The BHS problem has been studied under many different scenarios and under many different communication models (\cite{CzyzowiczComplexityBHS,CzyzowiczBHSSyncTree,DiLunaBHSDynamicRing, DobrevDangerousGraphTokens,DobrevAnonymousRingBHS,DobrevBHSUnorientedRingScattered,FlocchiniPingPongBHSPebble}). In all of these above mentioned works, the black hole is considered to be the classical black hole which always destroys any incoming agent without fail. 
In this work we are interested in a more general version of the classical black hole, called \textit{Byzantine Black Hole}, or BBH \cite{DBLP:conf/opodis/GoswamiBD024}. A BBH has the choice to act, at any given moment, as a black hole, destroying any data present on that node, or as a non-malicious node. Moreover, it is assumed that the initial position of the agents is safe. Note that the BHS problem does not have an exact equivalent under the Byzantine black hole assumption. Indeed, if the BBH always acts as an non-malicious node, then it can never be detected by any algorithm. Thus, in this work, our goal is not to locate the Byzantine black hole (BBH), but rather to perpetually explore the network despite its presence. Note that, if all agents are initially co-located and the BBH is a cut-vertex of the network, then it becomes impossible to visit every vertex, as the BBH can block access by behaving as a black hole. Consequently, in this work, we modify the exploration problem to focus on perpetually exploring \emph{one connected component} of the graph, among the connected components which would be obtained if the BBH and all its incident edges was removed from the network. In \cite{DBLP:conf/opodis/GoswamiBD024}, the problem was first introduced as \pbmPerpExpl\ for a ring network, where it was studied under various communication models and also under any initial deployment. In this work, in addition to investigating \pbmPerpExpl, we focus on a variant in which agents are required to perpetually explore specifically the connected component that includes their initial position, referred to as the \emph{home}. We call this variant \pbmPerpExplHome. \pbmPerpExplHome\ is particularly relevant in practical scenarios where the starting vertex serves as a central base, for example to aggregate the information collected by individual agents or as a charging station for energy-constrained agents. Therefore, during perpetual exploration, it is essential to ensure that the component being explored by the agents includes their \emph{home}. Note that the two variants are equivalent if the BBH is not a cut vertex.

\subsection{Related work}

Exploration of underlying topology by mobile agents in presence of a malicious host (also termed as black hole) has been well studied in literature. This problem, also known as black hole search or BHS problem was first introduced by Dobrev et al. \cite{DobrevAnonymousRingBHS}. Detailed surveys of the problems related to black holes can be found in \cite{DBLP:journals/eatcs/Markou12,DBLP:series/lncs/MarkouS19,WeiShiSurveyBHS}. Královič and Mikl\'{\i}k \cite{DBLP:conf/sirocco/KralovicM10,:Miklik2010} were the first to introduce some variants of the classical black hole, among which the \textit{gray hole}, a Byzantine version of the black hole (controlled by the adversary) with the ability, in each round, to behave either as a regular node or as a black hole. They considered the model, in which each agents are controlled by an asynchronous scheduler, where the underlying network is a ring with each node containing a whiteboard (i.e., each node can store some data, with which the agents can communicate among themselves). In this model, one of the problems they solved is Periodic Data Retrieval in presence of a gray hole with 9 agents. Later, Bampas et al.~\cite{DBLP:journals/tcs/BampasLMPP15} improved the results significantly under the same model. For the gray hole case, they obtained an optimal periodic data retrieval strategy with 4 agents. It may be noted that BBH behavior coincides with the ``gray hole'' considered in these papers. Moreover, periodic data retrieval is similar to perpetual exploration, as the main source of difficulty lies in periodically visiting all non-malicious nodes in order to collect the generated data. Goswami et al.~\cite{DBLP:conf/opodis/GoswamiBD024} first studied the perpetual exploration problem in presence of BBH (i.e., \pbmPerpExpl) in ring networks, where the agents are controlled by a synchronous scheduler. They studied this problem under different communication models (i.e., \textit{face-to-face}, \textit{pebble} and \textit{whiteboard}, where \textit{face-to-face} is the weakest and \textit{whiteboard} is the strongest model of communication), for different initial positions of the agents, i.e., for the case when the agents are initially co-located and for the case when they are initially scattered. Specifically for the case where the agents communicate \textit{face-to-face} and they are initially co-located at a node, they obtained an upper bound of 5 agents and a lower bound of 3 agents for \pbmPerpExpl.


\subsection{Our contributions}
We study the \pbmPerpExpl\ and \pbmPerpExplHome\ problems in specific topologies and in arbitrary synchronous connected networks with at most one BBH under the \textit{face-to-face} communication model (i.e., two agents can only communicate if both of them are on the same node). To the best of our knowledge, all previous papers in the literature which study variants of the classical black hole, do so assuming that the underlying network is a ring. Our primary optimization objective is the number of agents required to solve these problems.

For tree networks, we provide a tight bound on the number of agents for both problems. Specifically, we show that 4 agents are necessary and sufficient for \pbmPerpExpl\ in trees, and that 6 agents are necessary and sufficient for \pbmPerpExplHome\ in trees. Our algorithms work without initial knowledge of the size~$n$ of the network. However, knowledge of~$n$ would not reduce the number of agents, as our lower bounds do not assume that $n$ is unknown. Note that our 4-agent \pbmPerpExpl\ algorithm can be used to directly improve the 5-agent algorithm for rings, with knowledge of the network size, given in~\cite{DBLP:conf/opodis/GoswamiBD024} in the same model (\emph{face-to-face} communication, initially co-located agents).

To simplify the presentation, we present these results for path networks, and then we explain how to adapt the algorithms to work in trees with the same number of agents.


In the case of general network topologies, we propose an algorithm that solves the problem \pbmPerpExplHome, hence also \pbmPerpExpl, with $3\Delta+3$ agents, where $\Delta$ is the maximum degree of the graph, without knowledge of the size of the graph.

In terms of lower bounds, we first show that, if the BBH behaves as a classical black hole (i.e., if it is activated in every round), then at least $\Delta+1$ agents are necessary for perpetual exploration, even with knowledge of~$n$. In the underlying graph used in this proof, the BBH is not a cut vertex, hence the lower bound holds even for \pbmPerpExpl. This can be seen as an analogue, in our model, of the well-known $\Delta+1$ lower bound for black hole search in asynchronous networks~\cite{DBLP:journals/dc/DobrevFPS06}. In passing, under the same assumption of the BBH behaving as a classical black hole, we discuss an algorithm that performs perpetual exploration with only $\Delta+2$ agents, without knowledge of~$n$. 

We then use the full power of the Byzantine black hole to prove a stronger, and more technical, lower bound of $2\Delta-1$ agents. In this last lower bound, the structure and, indeed, the size of the graph is decided dynamically based on the actions of the algorithm. Hence, the fact that agents do not have initial knowledge of~$n$ is crucial in this proof. In this case, the BBH may be a cut vertex, but the adversarial strategy that we define in the proof never allows any agents to visit any other component than the one containing the home node. Therefore, this lower bound also carries over to \pbmPerpExpl.

The above results are the first bounds to be obtained for a more powerful variant of a black hole, for general graphs, under the assumption that the agents have no initial topological knowledge about the network. 

Several technical details and proofs are omitted from the main part of the paper, and they can be found in the Appendix.


\section{Model and basic definitions}

The agents operate in a simple, undirected, connected port-labeled graph $G=(V,E,\lambda)$, where $\lambda=\bigl(\lambda_v\bigr)_{v\in V}$ is a collection of port-labeling functions $\lambda_v:E_v\rightarrow \{1,\dots,\delta_v\}$, where $E_v$ is the set of edges incident to node~$v$ and $\delta_v$ is the degree of~$v$. We denote by~$n$ the number of nodes and by~$\Delta$ the maximum degree of~$G$.

An algorithm is modeled as a deterministic Turing machine. Agents are modeled as instances of an algorithm (i.e., copies of the corresponding deterministic Turing machine) which move in~$G$. Each agent is initially provided with a unique identifier.

The execution of the system proceeds in synchronous rounds. In each round, each agent receives as input the degree of its current node, the local port number through which it arrived at its current node, i.e.\ $\lambda_v(\{u,v\})$ if it just arrived at node~$v$ by traversing edge~$\{u,v\}$, or $0$ if it did not move in the last round, and the configurations of all agents present at its current node. It then  computes the local port label of the edge that it wishes to traverse next (or $0$ if it does not wish to move). All agents are activated, compute their next move, and perform their moves in simultaneous steps within a round. We assume that all local computations take the same amount of time and that edge traversals are instantaneous.

Note that we will only consider initial configurations in which all agents are co-located on a node called ``the \emph{home}''. In this setting, the set of unique agent identifiers becomes common knowledge in the very first round.

At most one of the nodes $\frakb \in V$ is a Byzantine black hole. In each round, the adversary may choose to activate the black hole. If the black hole is activated, then it destroys all agents that started the round at~$\frakb$, as well as all agents that choose to move to~$\frakb$ in that round.  The agents have no information on the position of the Byzantine black hole, except that it is not located at the \emph{home} node. Furthermore, the agents do not have initial knowledge of the size of the graph.

\subsection{Problem definition}
We define the \textsc{Perpetual Exploration} problem with initially co-located agents in the presence of a Byzantine black hole, hereafter denoted \pbmPerpExpl, as the problem of perpetually exploring at least one of the connected components resulting from the removal of the Byzantine black hole from the graph. If the graph does not contain a Byzantine black hole, then the entire graph must be perpetually explored.

If, in particular, the perpetually explored component \emph{must} be the component containing the \emph{home}, then the corresponding problem is denoted as \pbmPerpExplHome.

\begin{definition}\label{def:instance}
    An \emph{instance} of the \pbmPerpExpl\ problem is a tuple~$\left\langle G,k,h,\frakb \right\rangle$, where $G=(V,E,\lambda)$ is a connected port-labeled graph, $k\geq 1$ is the number of agents starting on the home $h\in V$, and $\frakb\in (V\setminus\{h\})\cup\{\bot\}$ is the node that contains the Byzantine black hole. If $\frakb=\bot$, then $G$ does not contain a Byzantine black hole.
\end{definition}

For the following definitions, fix a \pbmPerpExpl\ instance $I=\left\langle G,k,h,\frakb \right\rangle$, where $G=(V,E,\lambda)$, and let $\calA$ be an algorithm.

\begin{definition}
    We say that an execution of $\calA$ on $I$ \emph{perpetually explores a subgraph~$H$ of~$G$} if every node of~$H$ is visited by some agent infinitely often. 
\end{definition}

\begin{definition}\label{def:correctalg}
    Let $C_1,C_2,\dots,C_t$ be the connected components of the graph~$G-\frakb$, resulting from the removal of~$\frakb$ and all its incident edges from~$G$. If $\frakb=\bot$, then $t=1$ and $C_1\equiv G$. Without loss of generality, let~$h\in C_1$.
    
    We say that $\calA$ \emph{solves \pbmPerpExpl\ on~$I$}, if for every execution starting from the initial configuration in which $k$ agents are co-located at node~$h$, at least one of the components $C_1,C_2,\dots,C_t$ is perpetually explored.

    We say that $\calA$ \emph{solves \pbmPerpExplHome\ on~$I$}, if for every execution starting from the initial configuration in which $k$ agents are co-located at node~$h$, the component $C_1$ (containing the home) is perpetually explored.
\end{definition}

Finally, we say that $\calA$ \emph{solves \pbmPerpExpl\ with~$k_0$ agents} if it solves the problem on any instance with~$k\geq k_0$ agents (similarly for \pbmPerpExplHome). Note that any algorithm that solves \pbmPerpExplHome\ also solves \pbmPerpExpl.

Some further preliminary notions that are required for the full technical presentation of the lower bounds in Appendices~\ref{Appendix: Proof of impossibility with 3 PerpExplorationBBH Path} and~\ref{Appendix: Proof of mpossibility with 5 PerpExplorationBBH-Home Path}, are given in Appendix~\ref{Appendix:prelims}.


\section{Perpetual exploration in path and tree networks}

In this section, our main aim is to establish the following two theorems, giving the optimal number of agents that solve \pbmPerpExpl\ and \pbmPerpExplHome\ in path graphs.

\begin{theorem}
    \label{Thm: necessity and sufficiency of 4 agents for perpExploration-BBH path} 4 agents are necessary and sufficient to solve \pbmPerpExpl\ in path graphs, without initial knowledge of the size of the graph.
\end{theorem}

\begin{theorem}
    \label{Thm: necessity and sufficiency of 6 agents for perpExploration-BBH-Home path} 6 agents are necessary and sufficient to solve \pbmPerpExplHome\ in path graphs, without initial knowledge of the size of the graph.
\end{theorem}

For the necessity part, we prove that there exists no algorithm solving \pbmPerpExpl\ (resp.\ \pbmPerpExplHome) with 3 agents (resp.\ 5 agents), even assuming knowledge of the size of the graph. For the sufficiency part of Theorem~\ref{Thm: necessity and sufficiency of 6 agents for perpExploration-BBH-Home path}, we provide an algorithm, which we call \textsc{Path\_PerpExplore-BBH-Home}, solving \pbmPerpExplHome\ with 6 initially co-located agents, even when the size of the path is unknown to the agents. Thereafter, we modify this algorithm to solve \pbmPerpExpl\ with 4 initially co-located agents, thus establishing the sufficiency part of Theorem~\ref{Thm: necessity and sufficiency of 4 agents for perpExploration-BBH path}.

In Section~\ref{sec:lbpath}, we give a brief sketch of the approach we have used to prove the lower bounds on the number of agents. Then, in Section~\ref{section: Path Graph Algorithm Description}, we describe the algorithms.

\subsection{Lower bounds in paths} \label{sec:lbpath}

\begin{theorem} 
    \label{thm:impossibility with 3 PerpExplorationBBH Path} At least 4 agents are necessary to solve \pbmPerpExpl\ in all path graphs with $n\geq 9$ nodes, even assuming initial knowledge of the size of the graph.
\end{theorem}

In order to prove Theorem~\ref{thm:impossibility with 3 PerpExplorationBBH Path}, we establish, in a series of lemmas, several different types of configurations from which the adversary can force an algorithm to fail, by destroying all agents or by failing to entirely perpetually explore any component. Omitting a lot of details, which can be found in Appendix~\ref{Appendix: Proof of impossibility with 3 PerpExplorationBBH Path}, these configurations are as follows:
\begin{itemize}
    \item If one agent remains alive, and at least two BBH locations are consistent with the history of the agent up to that time, then the algorithm fails.
    \item If two agents remain alive on the same side \textit{of the path} of at least three potential BBH locations, then the algorithm fails. Also, if two agents remain alive on either side of at least three potential BBH locations, then the algorithm fails.
    \item If three agents remain alive on the same side of at least eight potential BBH locations, of which at least the first three are consecutive nodes in the path, then the algorithm fails.
\end{itemize}
Theorem~\ref{thm:impossibility with 3 PerpExplorationBBH Path} then follows by observing that, in the initial configuration with three co-located agents, all nodes apart from the home node are potential BBH positions.

\begin{theorem}
    \label{thm:impossibility with 5 PerpExplorationBBH-Home Path} At least 6 agents are necessary to solve \pbmPerpExplHome\ in all path graphs with $n\geq 145$ nodes, even assuming initial knowledge of the size of the graph.
\end{theorem}

The first element of the proof of Theorem~\ref{thm:impossibility with 5 PerpExplorationBBH-Home Path} is that at least one agent must remain at the \emph{home} node up to at least the destruction of the first agent by the BBH (actually, up to the time of destruction of the first agent plus the time it would take for that information to arrive at the home node), otherwise the adversary can trap all agents in the wrong component, i.e., the one not containing the \emph{home} node.

Then, very informally, we argue that, no matter how the algorithm divides the 5 available agents into a group of \emph{exploring} agents that must go arbitrarily far from the home node, and a group of \emph{waiting} agents that must wait at the home node, either the number of \emph{exploring} agents is not enough for them to explore arbitrarily far from the \emph{home} node and, at the same time, be able to detect exactly the BBH position if it is activated, or the number of \emph{waiting} agents is not enough to allow them to recover the information of the exact BBH position from one of the remaining \emph{exploring} agents that has been stranded on the other side of the BBH.

The above argument is presented in full technical detail in Appendix~\ref{Appendix: Proof of mpossibility with 5 PerpExplorationBBH-Home Path}.

\subsection{Description of Algorithm \textsc{\textmd{Path\_PerpExplore-BBH-Home}}}\label{section: Path Graph Algorithm Description}

Here we first describe the algorithm, assuming that the BBH does not intervene and destroy any agent. Then we describe how the agents behave after the BBH destroys at least one agent (i.e., the BBH intervenes).

We call this algorithm \textsc{Path\_PerpExplore-BBH-Home}. Let $\left\langle P,6,h,\frakb\right\rangle$ be an instance of~\pbmPerpExplHome, where $P=(V,E,\lambda)$ is a port-labeled path. Per Definition~\ref{def:correctalg}, all agents are initially co-located at~$h$ (the home node). To simplify the presentation, we assume that $h$ is an extremity of the path, and we explain how to modify the algorithm to handle other cases in Remark~\ref{remark: h is not an extreme end} in Appendix~\ref{Appendix: Path Alg}. Our algorithm works with 6 agents. Initially, among them the four least ID agents will start exploring $P$, while the other two agents will wait at $h$, for the return of the other agents. We first describe the movement of the four least ID agents say, $a_0$, $a_1$, $a_2$ and $a_3$, on $P$. Based on their movement, they identify their role as follows: $a_0$ as $F$, $a_1$ as $I_2$, $a_2$ as $I_1$ and $a_3$ as $L$. The identities $F$, $I_2$, $I_1$ and $L$ are denoted as Follower, Intermediate$2$, Intermediate$1$ and Leader, respectively. The exploration is performed by these four agents in two steps, in the first step, they form a particular pattern on $P$. Then in the second step, they move collaboratively in such a way that the pattern is translated from the previous node to the next node in five rounds. Since the agents do not have the knowledge of $n$, where $|V|=n$, they do the exploration of $P$, in $\log n$ phases, and then this repeats. In the $i$-th phase, the four agents start exploring $P$ by continuously translating the pattern to the next node, starting from $h$, and moves up to a distance of $2^i$ from $h$. Next, it starts exploring backwards in a similar manner, until they reach $h$. It may be observed that, any phase after $j-$th phase (where $2^j\ge n$), the agents behave in a similar manner, as they behaved in $j-$th phase, i.e., they move up to $2^j$ distance from $h$, and then start exploring backwards in a similar manner, until each agent reaches $h$. Note that, this can be done with agents having $\calO(\log n)$ bits of memory, in order to store the current phase number.

We now describe how the set of four agents, i.e., $L$, $I_1$, $I_2$ and $F$ create and translate the pattern to the next node. 

\noindent\textbf{Creating pattern:} $L, I_1,I_2$ and $F$ takes part in this step from the very first round  of any phase, starting from $h$. In the first round, $L,~I_1$ and $I_2$ moves to the next node. Then in the next round only $I_2$ returns back to \emph{home} to meet with $F$. Note that, in this configuration the agents $L$, $I_1$, $I_2$ and $F$ are at two adjacent nodes while $F$ and $I_2$ are together and $L$ and $I_1$ are together on the same node. We call this particular configuration the \textit{pattern}. We name the exact procedure as \textsc{Make\_Pattern}.  



\noindent\textbf{Translating pattern:} After the pattern is formed in first two rounds of a phase, the agents translate the pattern to the next node until the agent $L$ reaches either at one end of the path graph $P$, or, reaches a node at a distance $2^i$ from $h$ in the $i$-th phase. Let, $v_0,~v_1,~v_2$ be three consecutive nodes on $P$, where, suppose $L$ and $I_1$ is on $v_1$ and $F$ and $I_2$ is on $v_0$. This translation of the pattern makes sure that after 5 consecutive rounds $L$ and $I_1$ is on $v_2$ and $F$ and $I_2$ is on $v_1$, thus translating the pattern by one node. We call these 5 consecutive rounds where the pattern translates starting from a set of 2 adjacent nodes, say $v_0,~v_1$, to the next two adjacent nodes, say $v_1,~v_2$, a \textit{sub-phase} in the current phase. The description of the 5 consecutive rounds i.e., a sub-phase without the intervention of the BBH at $\frakb$ (such that $\frakb\in\{v_0,v_1,v_2\}$) are as follows.

\noindent\underline{\textbf{Round 1:}} $L$ moves to $v_2$ from $v_1$.

\noindent\underline{\textbf{Round 2:}} $I_2$ moves to $v_1$ from $v_0$ and $L$ moves to $v_1$ back from $v_2$.

\noindent\underline{\textbf{Round 3:}} $I_2$ moves back to $v_0$ from $v_1$ to meet with $F$. Also, $L$ moves back to $v_2$ from $v_1$.

\noindent\underline{\textbf{Round 4:}} $F$ and $I_2$ moves to $v_1$ from $v_0$ together.

\noindent\underline{\textbf{Round 5:}} $I_1$ moves to $v_2$ from $v_1$ to meet with $L$.

\noindent So after completion of round 5 the pattern is translated from nodes $v_0$, $v_1$ to $v_1$, $v_2$. The pictorial description of the translate pattern is explained in Fig. \ref{fig:translate-pattern-MainVersion}.

\begin{figure}[h!]
    \centering
    \begin{minipage}{0.48\textwidth}
        \centering
        \includegraphics[width=0.7\textwidth]{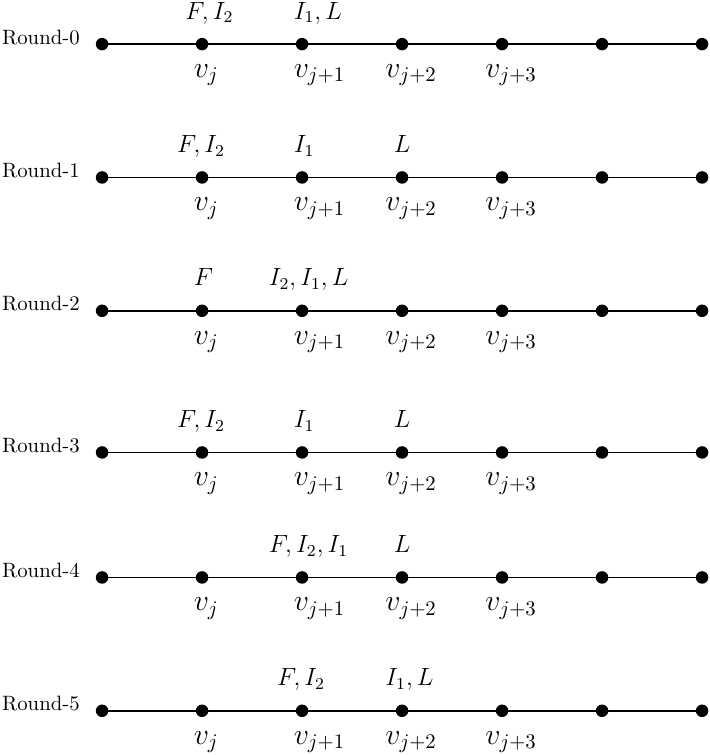}
        \caption{Depicts translating pattern steps}
        \label{fig:translate-pattern-MainVersion}
    \end{minipage}
    \hfill
    \begin{minipage}{0.48\textwidth}
        \centering
        \includegraphics[width=0.6\textwidth]{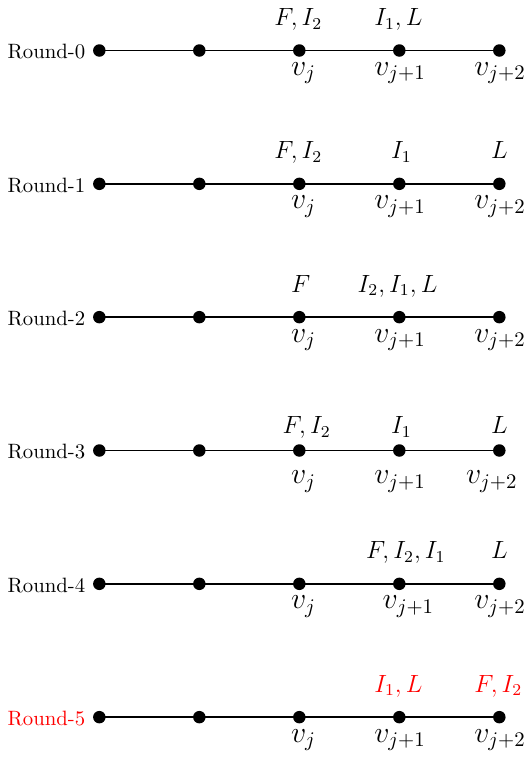}
        \caption{Depicts the step-wise interchange of roles, when the agents reach the end of the path graph, while performing translate pattern}
        \label{fig:translate-pattern-flip-MainVersion}
    \end{minipage}%
\end{figure}
Now, suppose $v_2$ is the node upto which the pattern was supposed to translate at the current phase (or, it can be the end of the path graph also). So, when $L$ visits $v_2$ for the first time, in round 1 of  some sub-phase it knows that it has reached the end of the path graph for the current phase. Then in the same sub-phase at round 2 it conveys this information to $I_2$ and $I_1$. In round 3 of the same sub-phase, $F$ gets that information from $I_2$. So at the end of the current sub-phase all agents has the information that they have explored either one end of the path graph or the node upto which they were supposed to explore in the current phase. In this case, they interchange the roles as follows: the agent which was previously had role $L$ changes role to $F$, the agent having role $F$ changes it to $L$, $I_1$ changes role to $I_2$ and $I_2$ changes role to $I_1$ (refer to Fig. \ref{fig:translate-pattern-flip-MainVersion}). Then from the next sub-phase onwards they start translating the pattern towards $h$. It may be noted that, once $L$ (previously $F$) reaches $h$ in round 1 of a sub-phase, it conveys this information to the remaining agents in similar manner as described above. So, at round 5 of this current sub-phase, $F$ (previously $L$) and $I_2$ (previously $I_1$) also reaches $h$, and meets with $L$ (previously $F$) and $I_1$ (previously $I_2$). We name the exact procedure as \textsc{Translate\_Pattern}.

\noindent \textbf{Intervention by the BBH:} 
Now we describe the behaviour of the agents while the BBH kills at least one exploring agent during create pattern or, during translate pattern. We first claim that, however the BBH intervenes while the four agents are executing creating pattern or translating pattern, there must exists at least one agent, say $A_{alive}$ which knows the exact location of the BBH. The justification of the claim is immediate if we do case by case studies fixing a BBH position and the round it is intervening with the agents. This can be during creating pattern or in a particular sub-phase of translating pattern. We have added this study in Appendix \ref{Appendix: Path Alg}. 
Now, there can be two cases:  

\noindent\textbf{Case-I :} Let $A_{alive} \in C_1$ ($C_1$ being the connected component of $G-\frakb$ that contains $h$) then it can perpetually explore every vertex in $C_1$ satisfying the requirement. 

\noindent\textbf{Case-II :} Let $A_{alive} \in C_2$ ($C_2$ being the connected component of $G-\frakb$ such that $h \notin C_2$).  In this case $A_{alive}$ places itself to the adjacent node of the BBH on $C_2$.
Let $T_i$ be the maximum time, required for the 4 agents (i.e., $L,~I_1,~I_2$ and $F$) to return back to $h$ in $i-$th phase if BBH does not intervene. Let us denote the waiting agents at $h$, i.e., $a_4,~a_5$ as $F_1,~F_2$. Starting from the $i$-th phase, they wait for $T_i$ rounds for the other agents to return. Now, if the set of agents $L,~I_1,~I_2$ and $F$ fail to return back to \emph{home} within $T_i$ rounds in phase, $i$, the agents $F_1$ and $F_2$ starts moving \textit{cautiously}. In the cautious move, first $F_1$ visits the next node and in the next round, returns back to the previous node to meet with $F_2$, that was waiting there for $F_1$. If $F_1$ fails to return then $F_2$ knows the exact location of $\frakb$, which is the next node $F_1$ visited. In this case, $F_2$ remains in the component of $h$, hence can explore it perpetually. Otherwise, if $F_1$ returns back to $F_2$, then in the next round both of them moves together to the next node.

\begin{figure}
    \centering
\includegraphics[width=0.7\linewidth]{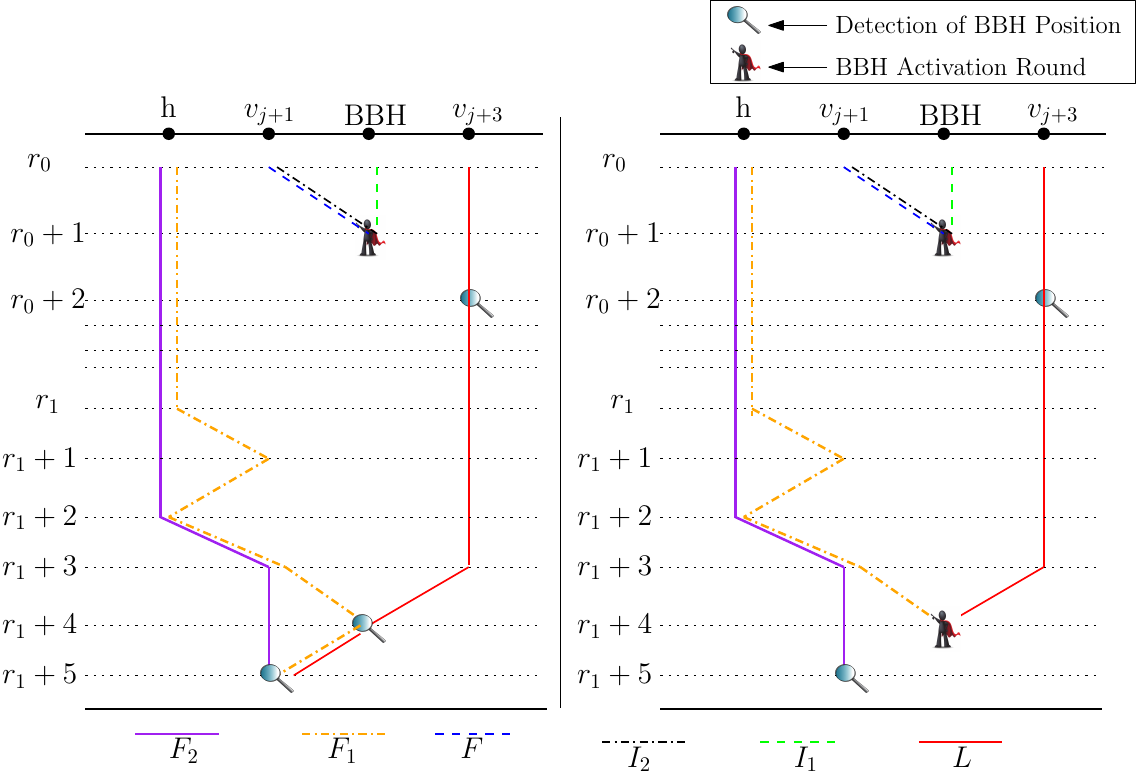}
    \caption{Represents the time diagram, in which at least one among $F_1$ and $F_2$ detects $\frakb$, and perpetually explores the path graph.}
    \label{fig:perpetuallyexplore-timediagram-MainVersion}
\end{figure}

 It may be noted that, $A_{alive}$ knows which phase is currently going on and so it knows the exact round at which $F_1$ and $F_2$ starts moving cautiously. Also, it knows the exact round at which $F_1$ first visits $\frakb$, say at round $r$. $A_{alive}$ waits till round $r-1$, and at round $r$ it moves to $\frakb$. Now at round $r$, if adversary activates $\frakb$, it destroys both $F_1$ and $A_{alive}$ then, $F_1$ fails to return back to $F_2$ in the next round. This way, $F_2$ knows the exact location of $\frakb$, while it remains in $C_1$. So it can explore $C_1$ by itself perpetually. The right figure in Fig. \ref{fig:perpetuallyexplore-timediagram-MainVersion}, represents the case where $L$ detects $\frakb$ at round $r_0+2$, and waits till round $r_1+3$. In the meantime, at round $r_1$ (where $r_1=r'_0+T_i$, $r'_0<r_0$ and $r'_0$ is the first round of phase $i$), $F_1$ and $F_2$ starts moving cautiously. Notably, along this movement, at round $r_1+4$, $F_1$ visits $\frakb$, and at the same time $L$ as well visits $\frakb$ from $v_{j+3}$. The adversary activates $\frakb$, and both gets destroyed. So, at round $r_1+5$, $F_2$ finds failure of $F_1$'s return and understands the next node to be $\frakb$, while it is present in $C_1$. Accordingly, it perpetually explores $C_1$.

On the other hand, if at round $r$, adversary doesn't activate $\frakb$, then $F_1$ meets with $A_{alive}$ and knows that they are located on the inactivated $\frakb$. In this case they move back to $C_1$ and starts exploring the component $C_1$, avoiding $\frakb$. The left figure in Fig. \ref{fig:perpetuallyexplore-timediagram-MainVersion} explores this case, where $L$ detects the position of $\frakb$ at round $r_0+2$, and stays at $v_{j+3}$ until $r_1+3$, then at round $r_1+4$, when $F_1$ is also scheduled to visit $\frakb$, $L$ also decides to visit $\frakb$. But, in this situation, the adversary does not activate $\frakb$ at round $r_1+4$, so both $F_1$ and $L$ meets, gets the knowledge from $L$ that they are on $\frakb$. In the next round, they move to $v_{j+1}$ which is a node in $C_1$, where they meet $F_2$ and shares this information. After which, they perpetually explore $C_1$.

The correctness of the algorithm follows from the description of the algorithm and the case by case study of all possible interventions by the BBH, which can be found in Appendix~\ref{Appendix: Path Alg}. We have thus established the sufficiency part of Theorem~\ref{Thm: necessity and sufficiency of 6 agents for perpExploration-BBH-Home path}.

  \begin{note}
  It may be noted that, our algorithm \textsc{Path\_PerpExplore-BBH-Home}, without $F_1$ and $F_2$, is sufficient to solve \pbmPerpExpl. It is because, if $A_{alive}$ in $C_i$ ($i \in \{1,2\}$) after one intervention by the BBH, it can perpetually explore $C_i$ without any help from $F_1$ and $F_2$ as it knows the exact location of the BBH (refer to the correctness under \textbf{Intervention by the BBH} in Appendix \ref{Appendix: Path Alg}). This proves the sufficiency part of Theorem~\ref{Thm: necessity and sufficiency of 4 agents for perpExploration-BBH path}.
\end{note}

\begin{remark} \label{rem:subsume 5 agent ring alg}
    Note that our 4-agent \pbmPerpExpl\ algorithm for paths directly improves the 5-agent algorithm for rings with knowledge of the 
    network size, given in~\cite{DBLP:conf/opodis/GoswamiBD024}, in the \emph{face-to-face} model with initially co-located agents. Indeed, the 4-agent pattern can keep moving around the ring as if in a path graph. Naturally, it will never reach a node of degree~$1$. If and when the BBH destroys an agent, then, as we showed, at least one agent remains alive and knows the position of the BBH. As the agent knows the size of the ring in this case, it can perform perpetual exploration of the non-malicious ring nodes, without ever visiting the BBH.
\end{remark}%

  \subsubsection{Modification of the path algorithms to work in trees}

We modify the algorithms for path graphs of Section~\ref{section: Path Graph Algorithm Description} to work for trees, with the same number of agents. The algorithms for trees work by translating the pattern from one node to another by following the $k-\texttt{Increasing-DFS}$ \cite{fraigniaud2005graph} algorithm up to a certain number of nodes in a certain phase.
  The details of the modification can be found in Appendix \ref{Appendix: Tree Alg}.


\section{Perpetual exploration in general graphs}

In this section, we establish upper and lower bounds on the optimal number of agents required to solve Perpetual Exploration in arbitrary graphs with a BBH, without any initial knowledge about the graph. In particular, we give a lower bound of~$2\Delta-1$ agents for \pbmPerpExpl\ (Theorem~\ref{thm: equivalent statement} below), which carries over directly to \pbmPerpExplHome, and an algorithm for \pbmPerpExplHome\ using $3\Delta+3$ agents (Theorem~\ref{theorem:Final3DeltaUpperBound-mainversion} below), which also solves \pbmPerpExpl.






We give a brief sketch of the lower bound proof in Section~\ref{subsection: lower bound general graph}, and we present the algorithm in Section~\ref{subsection: Algorithm description general graph}.

\subsection{Lower bound in general graphs}\label{subsection: lower bound general graph}


For the lower bound, we construct a particular class of graphs in which any algorithm using $2\Delta-2$ agents or less must fail. The complete proof can be found in Appendix~\ref{appendix: lower bound proof General BBH}.

\begin{theorem}\label{thm: equivalent statement}
For every $\Delta\geq 4$, there exists a class of graphs~$\calG$ with maximum degree~$\Delta$, such that any algorithm using at most~$2\Delta-2$ agents, with no initial knowledge about the graph, fails to solve \pbmPerpExpl\ in at least one of the graphs in~$\calG$.
\end{theorem}

\begin{proof}[Proof sketch] For a fixed $\Delta\geq 4$, we construct the corresponding graph class~$\calG$ by taking an underlying path $\mathcal{P}$, consisting of two types of vertices $\{v_i:1\leq i\leq\Delta\}$ and, for each~$i\leq \Delta$, the nodes $\{u^i_j:1\leq j\leq l_i\}$, where $l_i\geq 1$. The nodes $\{u^i_j\}$ form a subpath of length~$l_i+1$ connecting~$v_i$ to $v_{i+1}$. 
The nodes $v_i$ (for $\in\{1,2,\dots,\Delta-1\}$) are special along $\mathcal{P}$, because they are connected to the BBH~$\frakb$ either directly or via a new node~$w_i$. Every node of~$\calP$, as well as~$\frakb$, now completes its degree up to~$\Delta$ by connecting to at most~$\Delta-1$  trees of height~$2$ (see example in Figure~\ref{fig:example graph of G class}). Every node~$w_i$ completes its degree up to~$\Delta$ by connecting to $\Delta-2$ new nodes.

Note that, to reach $\frakb$ from a vertex of the form $u^i_j$, it is required to visit either $v_i$ or $v_{i+1}$. In the example of Figure~\ref{fig:example graph of G class}, we have $\Delta=4$, $l_1=l_2=2$ and $l_3=1$, so in $\mathcal{P}$ there are two vertices $u^1_1,u^1_2$ between~$v_1$ and~$v_2$, two vertices $u^2_1,u^2_2$ between $v_2$ and $v_3$, and one vertex $u^3_1$ between $v_3$ and $v_4$.
In addition, $v_2$ is directly connected to $\frakb$ (or BBH) whereas $v_1$ and $v_3$ are connected to $\frakb$ via a path of length 2.

\begin{figure}
    \centering
\includegraphics[width=0.6\linewidth]{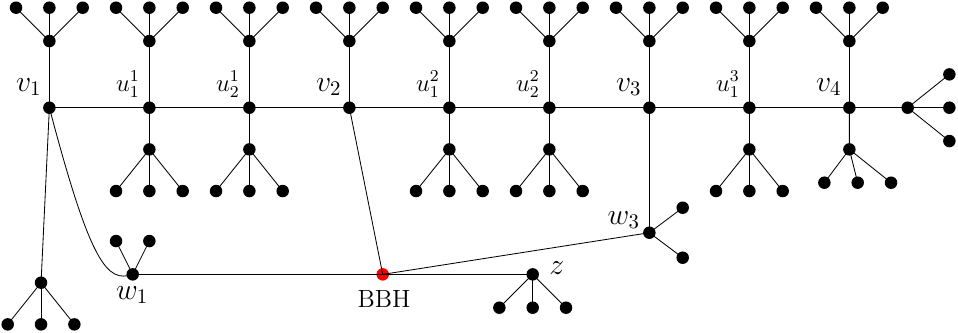}
    \caption{An example graph of the class $\calG$, used in the proof of Theorem~\ref{thm: equivalent statement}.}
    \label{fig:example graph of G class}
\end{figure}

Given an algorithm $\calA$ which claims to solve \pbmPerpExpl\ with $2\Delta-2$ co-located agents, the adversary  returns a port-labeled graph $G=(V,E,\lambda)\in\mathcal{G}$ in which $\mathcal{A}$ fails, by choosing the lengths~$l_i$ and the connection between $v_i$ and $\frakb$ (direct or through $w_i$). A high-level idea of the proof is as follows. First, it can be ensured that, $\mathcal{A}$ must instruct at least one agent to visit a node which is at 2 hop distance from $v_i$, for all $i\in\{1,2,\dots,\Delta-1\}$. Let the agent which first visits a node which is at 2 hop distance from $v_i$, starts from $v_i$ at time $t_i$. \textbf{Case-1}: Now, based on $\mathcal{A}$, if at $t_i$, at least 2 agents are instructed to move from $v_i$ along the same port, then the adversary chooses a graph from $\mathcal{G}$, such that $v_i$ is connected to $\frakb$, and returns a port labeling, such that at $t_{i}+1$, these agents reach $\frakb$. \textbf{Case-2}: Otherwise, the adversary chooses a graph from $\mathcal{G}$ such that, either $v_i$ is connected to $\frakb$ or there exists an intermediate vertex $w_i$ connecting $v_i$ with $\frakb$. In this case also, the adversary returns the port labeling such that, the agent must move to $\frakb$ (or $w_i$) at $t_{i}+1$ and to some vertex $z$ (or $\frakb$) at $t'_i$ (where $t'_i>t_i+1$).

In addition, the graph chosen by the adversary sets the distance between $v_{i-1}$ and $v_i$ in a way to ensures that, within the time interval $[t_i,t’_i]$, no agent from any $v_{\alpha}$ ($\alpha\in\{0,1,\dots, i-1\}$) attempts to visit $\frakb$. So, for \textbf{Case-1}, two agents directly gets destroyed. For \textbf{Case-2}, after the first agent is destroyed, remaining agents do not understand which among the next 2 nodes from $v_i$ is $\frakb$. Now, for $\mathcal{A}$ to succeed, at least one other agent from $v_i$ gets destroyed, at some round $t''_i$. This shows that from each $v_i$, for all $i\in\{1,2,\dots,\Delta-1\}$ at least 2 agents each, are destroyed. Thus, $\mathcal{A}$ fails. 
\end{proof}

\subsection{Description of Algorithm \textmd{\textsc{Graph\_PerpExplore-BBH-Home}}}\label{subsection: Algorithm description general graph}

Here we discuss the algorithm, termed as \textsc{Graph\_PerpExplore-BBH-Home} that solves \textsc{PerpExploration-BBH-Home} on a general graph, $G=(V,E,\lambda)$. We will show that our algorithm requires at most $3\Delta+3$ agents. Let $\left\langle G,3\Delta+3,h,\frakb\right\rangle$ be an instance of the problem~\pbmPerpExplHome, where $G=(V,E,\lambda)$ is a simple port-labeled graph. The structure of our algorithm depends upon four separate algorithms \textsc{Translate\_Pattern} along with \textsc{Make\_Pattern} (discussed in Section \ref{section: Path Graph Algorithm Description}), \textsc{Explore} (explained in this section) and \textsc{BFS-Tree-Construction} \cite{chalopin2010constructing}. So, before going in to details of our algorithm that solves \textsc{PerpExploration-BBH-Home}, we recall the idea of \textsc{BFS-Tree-Construction}.

 An agent starts from a node $h\in V$ (also termed as \textit{home}), where among all nodes in $G$, only $h$ is marked. The agent performs breadth-first search (BFS) traversal, while constructing a BFS tree rooted at $h$. The agent maintains a set of edge-labeled paths, $\mathcal{P}=\{P_v~:~\text{edge labeled shortest path}$ $ \text{from $h$ to $v$,}~\forall v\in V~\text{such that the agent has visited $v$}\}$ while executing the algorithm. During its traversal, whenever the agent visits a node $w$ from a node $u$, then to check whether the node $w$ already belongs to the current BFS tree of $G$ constructed yet, it traverses each stored edge labeled paths in the set $\mathcal{P}$ from $w$ one after the other, to find if one among them takes it 
to the marked node $h$. If yes, then it adds to its map a cross-edge $(u,w)$. Otherwise, it adds to the already constructed BFS tree, the node $w$, accordingly $\mathcal{P}=\mathcal{P}\cup P_w$ is updated. The underlying data structure of \textsc{Root\_Paths} \cite{chalopin2010constructing} is used to perform these processes. This strategy guarantees as per Proposition 9 of \cite{chalopin2010constructing}, that \textsc{BFS-Tree-Construction} algorithm constructs a map of $G$, in presence of a marked node, within $\mathcal{O}(n^3\Delta)$ steps and using $\mathcal{O}(n\Delta\log n)$ memory, where $|V|=n$ and $\Delta$ is the maximum degree in $G$.

 In our algorithm, we use $k$ agents (in Theorem \ref{theorem:Final3DeltaUpperBound-mainversion}, it is shown that $k=3\Delta+3$ agents are sufficient), where they are initially co-located at a node $h\in V$, which is termed as \emph{home}. Initially, at the start our algorithm asks the agents to divide in to three groups, namely, \texttt{Marker}, \texttt{SG} and \texttt{LG}$_0$, where \texttt{SG} (or smaller group) contains the least four ID agents, the highest ID agent among all $k$ agents, denoted as \texttt{Marker} stays at $h$ (hence $h$ acts as a marked node), and the remaining $k-5$ agents are denoted as \texttt{LG}$_0$ (or larger group). During the execution of our algorithm, if at least one member of \texttt{LG}$_0$ detects one port leading to the BBH from one of its neighbor, in that case at least one member of \texttt{LG}$_0$ settles down at that node, acting as an \textit{anchor} blocking that port which leads to the BBH, and then some of the remaining members of \texttt{LG}$_0$ forms \texttt{LG}$_1$. In general, if at least one member of \texttt{LG}$_i$ detects the port leading to the BBH from one of its neighbors, then again at least one member settles down at that node acting as an \textit{anchor} to block that port leading to the BBH, and some of the remaining members of \texttt{LG}$_i$ forms \texttt{LG}$_{i+1}$, such that |\texttt{LG}$_{i+1}|<$|\texttt{LG}$_{i}$|. It may be noted that, a member of \texttt{LG}$_i$ only settles at a node $v$ (say) acting as an \textit{anchor}, only if no other \textit{anchor} is already present at $v$. Also, only if a member of \texttt{LG}$_i$ settles as an \textit{anchor}, then only some of the members of \texttt{LG}$_i$ forms \texttt{LG}$_{i+1}$.
 
 In addition to the groups \texttt{LG}$_0$ and \texttt{SG}, the \texttt{Marker} agent permanently remains at $h$. In a high-level the goal of our \textsc{Graph\_PerpExplore-BBH-Home} algorithm is to create a situation, where eventually at least one agent blocks, each port of $C_1$ that leads to the BBH (where $C_1,C_2,\dots,C_t$ are the connected components of $G-\frakb$, such that $h\in C_1$), we term these blocking agents as \textit{anchors}, whereas the remaining alive agents must perpetually explore at least $C_1$.

Initially from $h$, the members of \texttt{SG} start their movement, and the members of \texttt{LG}$_0$ stays at $h$ until they find that, none of the members of \texttt{SG} reach $h$ after a certain number of rounds. Next, we explain one after the other how both these groups move in $G$.

\noindent\textbf{Movement of \texttt{SG}}: The members (or agents) in \texttt{SG} works in phases, where in each phase the movement of these agents are based on the algorithms \textsc{Make\_Pattern} and \textsc{Translate\_Pattern} (both of these algorithms are described in Section \ref{section: Path Graph Algorithm Description}). Irrespective of which, the node that they choose to visit during \textit{making pattern} or \textit{translating pattern} is based on the underlying algorithm \textsc{BFS-Tree-Construction}.

More specifically, the $i$-th phase (for some $i>0$) is divided in two sub-phases: $i_1$-th phase and $i_2$-th phase. In the $i_1$-th phase, the members of \texttt{SG} makes at most $2^i$ translations, while executing the underlying algorithm \textsc{BFS-Tree-Construction}. Next, in the $i_2$-th phase, irrespective of their position after the end of $i_1$-th phase, they start translating back to reach $h$. After they reach $h$ during the $i_2$-th phase, they start $(i+1)$-th phase (which has again, ${(i+1)}_1$ and ${(i+1)}_2$ sub-phase). Note that, while executing $i_1$-th phase, if the members of \texttt{SG} reach $h$, in that case they continue executing $i_1$-th phase. We already know as per ~Section \ref{section: Path Graph Algorithm Description}, each translation using \textsc{Translate\_Pattern} requires 5 rounds and for creating the pattern using \textsc{Make\_Pattern} it requires 2 rounds. This concludes that, it requires at most ${Ti}_j=5\cdot 2^i+2$ rounds to complete $i_j$-th phase, for each $i>0$ and $j\in\{1,2\}$.

If at any point, along their traversal, the adversary activates the BBH, such that it interrupts the movement of \texttt{SG}. In that scenario, at least one member of \texttt{SG} must remain alive, exactly knowing the position of the BBH from its current node (refer to the discussion of \textbf{Intervention by the BBH} in Section \ref{section: Path Graph Algorithm Description} and Appendix \ref{Appendix: General Graph Algorithm}). The agent (or agents) which knows the exact location of the BBH, stays at the node adjacent to the BBH, such that from its current node, it knows the exact port that leads to the BBH, or in other words they act as \textit{anchors} with respect to one port, leading to the BBH. In particular, let us suppose, the agent holds the adjacent node of BBH, with respect to port $\alpha$ from BBH, then this agent is termed as \texttt{Anchor($\alpha$)}.

\noindent\textbf{Movement of \texttt{LG}$_0$}: These group members stay at $h$ with \texttt{Marker}, until the members of \texttt{SG} are returning back to $h$ in the $i_2$-th phase, for each $i>0$. If all members of \texttt{SG} do not reach $h$, in the $i_2$-th phase, i.e., within ${Ti}_2$ rounds since the start of $i_2$-th phase, then the members of \texttt{LG}$_0$ start their movement.

Starting from $h$, the underlying movement of the members of \texttt{LG}$_0$ is similar to \textsc{BFS-Tree-Construction}, but while moving from one node to another they do not execute neither \textsc{Make\_Pattern} nor \textsc{Translate\_Pattern}, unlike the members of \texttt{SG}. In this case, if all members of \texttt{LG}$_0$ are currently at a node $u\in V$, then three lowest ID members of \texttt{LG}$_0$ become the explorers, they are termed as $E^0_1$, $E^0_2$ and $E^0_3$ in increasing order of their IDs, respectively. If based on the \textsc{BFS-Tree-Construction}, the next neighbor to be visited by the members of \texttt{LG}$_0$ is $v$, where $v\in N(u)$, then the following procedure is performed by the explorers of \texttt{LG}$_0$, before \texttt{LG}$_0$ finally decides to visit $v$. 

Suppose at round $r$ (for some $r>0$), \texttt{LG}$_0$ members reach $u$, then at round $r+1$ both $E^0_2$ and $E^0_3$ members reach $v$. Next at round $r+2$, $E^0_3$ traverses to the first neighbor of $v$ and returns to $v$ at round $r+3$. At round $r+4$, $E^0_2$ travels to $u$ from $v$ and meets $E^0_1$ and then at round $r+5$ it returns back to $v$. This process iterates for each neighbor of $v$, and finally after each neighbor of $v$ is visited by $E^0_3$, at round $r+4\cdot (\delta_v-1)+1$ both $E^0_2$ and $E^0_3$ returns back to $u$. And in the subsequent round each members of \texttt{LG}$_0$ visit $v$. The whole process performed by $E^0_1$, $E^0_2$ and $E^0_3$ from $u$ is termed as \textsc{Explore}$(v)$, where $v$ symbolizes the node at which the members of \texttt{LG}$_0$ choose to visit from a neighbor node $u$. After the completion of \textsc{Explore}$(v)$, each member of \texttt{LG}$_0$ (including the explorers) visit $v$ from $u$. A pictorial description is explained in Fig. \ref{fig:Explore(v)-mainversion}.
\begin{figure}
    \centering
    \includegraphics[width=0.8\linewidth]{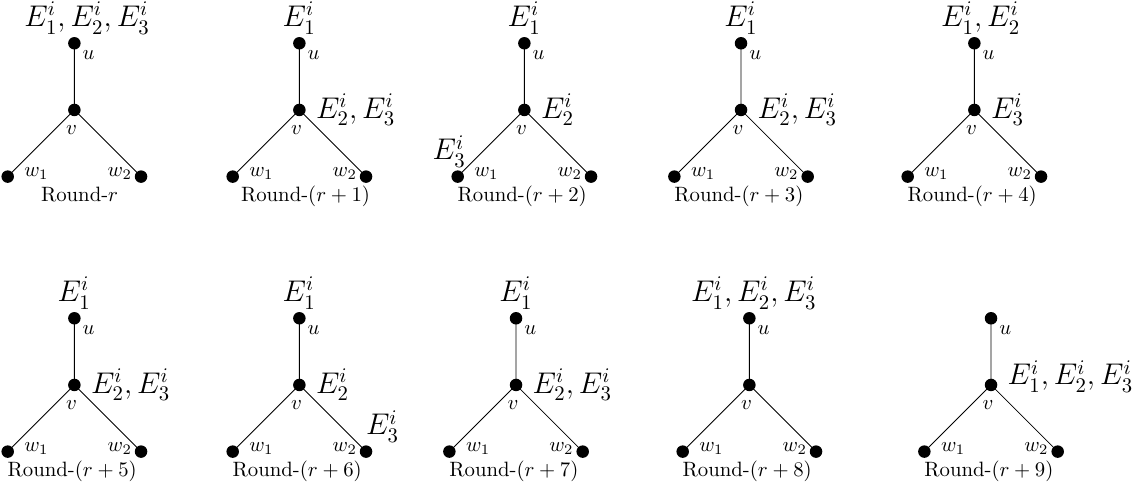}
    \caption{Depicts the round-wise execution of \textsc{Explore}$(v)$ from $u$ by the explorer agents of \texttt{LG}$_i$, for some $i\ge 0$ on the neighbors $w_1$ and $w_2$ of $v$.}
    \label{fig:Explore(v)-mainversion}
\end{figure}

It may be noted that, if the members of \texttt{LG}$_0$ at the node $u$, according to the \textsc{BFS-Tree-Construction} algorithm, are slated to visit the neighboring node $v$, then before executing \textsc{Explore}$(v)$, the members of \texttt{LG}$_0$ checks if there exists an \textit{anchor} agent blocking that port which leads to $v$. If that is the case, then \texttt{LG}$_0$ avoids visiting $v$ from $u$, and chooses the next neighbor, if such a neighbor exists and no \textit{anchor} agent is blocking that edge. If no such neighbor exists from $u$ to be chosen by \texttt{LG}$_0$ members, then they backtrack to the parent node of $u$, and start iterating the same process.

From the above discussion we can have the following remark.
\begin{remark}
    If at some round $t$, the explorer agents of \texttt{LG}$_i$  (i.e., $E^i_1, E^i_2$ and $E^i_3$), are exploring a two length path, say $P= u \rightarrow v \rightarrow w$, from $u$, then all members of \texttt{LG}$_i$ 
 agrees on $P$ at $t$. This is due to the fact that the agents while executing \textsc{Explore}($v$) from $u$ must follow the path $u \rightarrow
 v$ first. Now from $v$, $E^i_3$ chooses the next port in a particular pre-decided order (excluding the port through which it entered $v$). So, whenever it returns back to $v$ to meet $E^i_2$ after visiting a node $w$, $E^i_2$ knows which port it last visited and which port it will chose next and relay that information back to other agents of \texttt{LG}$_i$ on $u$. So, after $E^i_2$ returns back to $v$ again from $u$ when $E^i_3$ starts visiting the next port, all other agents know about it.
 \label{remark: agreement of path chosen-mainversion}
\end{remark}

During the execution of \textsc{Explore}$(v)$ from $u$, the agent $E^0_3$ can face one of the following situations:

\begin{itemize}
    \item It can find an \textit{anchor} agent at $v$, acting as \texttt{Anchor}$(\beta)$, for some $\beta\in\{1,\dots,\delta_v\}$. In that case, during its current execution of \textsc{Explore}$(v)$, $E^0_3$ does not visit the neighbor of $v$ with respect to the port $\beta$.
    \item It can find an \textit{anchor} agent at a neighbor $w$ (say) of $v$, acting as \texttt{Anchor}$(\beta')$, where $\beta'\in\{1,\dots, \delta_w\}$. If the port connecting $w$ to $v$ is also $\beta'$, then $E^0_3$ understands $v$ is the BBH, and accordingly tries to return to $u$, along the path $w\rightarrow v \rightarrow u$, and if it is able to reach to $u$, then it acts as an \textit{anchor} at $u$, with respect to the edge $(u,v)$. On the other hand, if port connecting $w$ to $v$ is not $\beta'$, then $E^0_3$ continues its execution of \textsc{Explore}$(v)$.
\end{itemize}

The agent $E^0_2$ during the execution of \textsc{Explore}$(v)$, can encounter one of the following situations, and accordingly we discuss the consequences that arise due to the situations encountered.

\begin{itemize}
    \item It can find an \textit{anchor} agent at $v$ where the \textit{anchor} agent is not $E^0_3$, in which case it continues to execute \textsc{Explore}$(v)$.
    \item It can find an \textit{anchor} agent at $v$ and finds the \textit{anchor} agent to be $E^0_3$. In this case, $E^0_2$ returns back to $u$, where \texttt{LG}$_1$ is formed, where \texttt{LG}$_1=\texttt{LG}_0\setminus\{E^0_3\}$. Next, the members of \texttt{LG}$_1$ start executing the same algorithm from $u$, with new explorers as $E^1_1$, $E^1_2$ and $E^1_3$.
    \item $E^0_2$ can find that $E^0_3$ fails to return to $v$ from a node $w$ (say), where $w\in N(v)$. In this case, $E^0_2$ understands $w$ to be the BBH, and it visits $u$ in the next round to inform this to remaining members of \texttt{LG}$_0$ in the next round, and then returns back to $v$, and becomes \texttt{Anchor}$(\beta)$, where $\beta\in\{1,\dots,\delta_v\}$ and $\beta$ is the port for $(v,w)$. On the other hand, \texttt{LG}$_0$ after receiving this information from $E^0_2$, transforms to \texttt{LG}$_1$ (where \texttt{LG}$_1=\texttt{LG}_0\setminus\{E^0_2,E^0_3\}$) and starts executing the same algorithm, with $E^1_1$, $E^1_2$ and $E^1_3$ as new explorers.
    
\end{itemize}

Lastly, during the execution of \textsc{Explore}$(v)$ the agent $E^0_1$ can face the following situation. 

\begin{itemize}
    \item $E^0_2$ fails to return from $v$, in this situation $E^0_1$ becomes \texttt{Anchor}$(\beta)$ at $u$, where $\beta$ is the port connecting $u$ to $v$. Moreover, the remaining members of \texttt{LG}$_0$, i.e., \texttt{LG}$_0\setminus\{E^0_1,E^0_2,E^0_3\}$ forms \texttt{LG}$_1$ and they start executing the same algorithm from $u$, with new explorers, namely, $E^1_1$, $E^1_2$ and $E^1_3$, respectively.
\end{itemize}

 For each $E^0_1$, $E^0_2$ and $E^0_3$, if they do not face any of the situations discussed above, then they continue to execute \textsc{Explore}$(v)$. 

A pictorial explanation of two scenarios, of how an \textit{anchor} settles at neighbor nodes of BBH is explained in Fig. \ref{fig:Time-Diagram-Explore-on-path-P-MainVersion}.

 \begin{figure}
    \centering
    \includegraphics[width=0.75\linewidth]{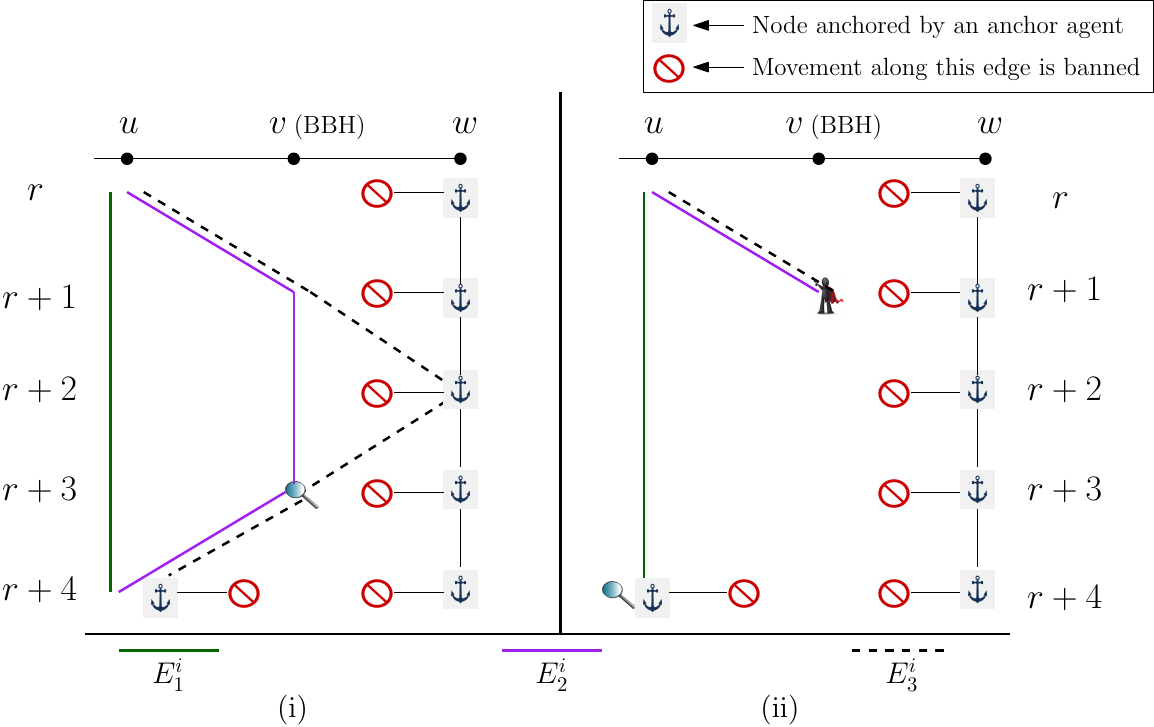}
    \caption{Depicts the time diagram of \textsc{Explore}$(v)$ of \texttt{LG}$_i$ along a specific path $P=u\rightarrow v \rightarrow w$, where $v=\frakb$, in which $w$ contains an \textit{anchor} agent for the edge $(v,w)$. In (i), it depicts even if $\frakb$ is not activated, an explorer agent settles as an \textit{anchor} at $u$ for the edge $(u,v)$. In (ii), it depicts that activation of $\frakb$ destroys both $E^i_2$ and $E^i_3$, then an explorer $E^i_1$ gets settled as \textit{anchor} at $u$ for $(u,v)$.}
    \label{fig:Time-Diagram-Explore-on-path-P-MainVersion}
\end{figure}

\noindent\textbf{Correctness:}
In order to prove the correctness of the algorithm \textsc{Graph\_PerpExplore-BBH-Home}, we give a brief sketch of the proof of the following theorem. The complete sequence of lemmas leading to this final theorem can be found in Appendix~\ref{Appendix: General Graph Algorithm}.

\begin{restatable}{theorem}{ubgeneral}
\label{theorem:Final3DeltaUpperBound-mainversion}
    Algorithm~\textsc{Graph\_PerpExplore-BBH-Home} solves \pbmPerpExplHome\ in arbitrary graphs with $3\Delta+3$ agents, having $\calO(n\Delta\log n)$ memory, without initial knowledge about the graph.
\end{restatable}

\begin{proof}[Proof sketch]
We first show that if $\frakb$ (i.e., BBH) never destroys any agent then the problem \pbmPerpExplHome~is solved only by the four agents in \texttt{SG} (follows from the correctness of \textsc{BFS-Tree-Construction}). On the otherhand if BBH destroys any agent, then we prove that there exists an agent from \texttt{SG} acting as an \textit{anchor} at a node in $N(\frakb)$, where $N(\frakb)$ indicates the neighbors of $\frakb$. The node can be either in $C_1$ or $C_j$, where $G-\frakb=C_1\cup \dots \cup C_t$ for some $t\ge1$ and $home\in C_1$ and $j\neq 1$. Now if we assume, BBH destroys at least one agent then, we define a set $\mathcal{U}$ of nodes in $G$, where $\mathcal{U}\subseteq N(\frakb)\cap C_1$ such that no vertex in $\mathcal{U}$ contains an \textit{anchor}. A node $u\in\mathcal{U}$ ceases to exist in $\mathcal{U}$, whenever an agent visits $u$ and blocks the edge $(u,\frakb)$, by becoming an \textit{anchor}. We show that during the execution of our algorithm, $\mathcal{U}$ eventually becomes empty (in bounded time, refer to Remark \ref{remark: Eventual time for placing anchor} in Appendix \ref{appendix: lower bound proof General BBH}). In addition to that, we also show that, the agents in \texttt{LG}$_i$ (for some $i\ge 0$) never visit a node which is not in $C_1$. By \texttt{LG}$_i$ visiting a node, we mean to say that all agents are located at that node simultaneously. This concludes that, all neighbors in $C_1$ of $\frakb$, gets anchored by an \textit{anchor} agent. So, eventually all the agents, which are neither an \textit{anchor} nor a \texttt{Marker}, can perpetually explore $C_1$. Also, to set up an \textit{anchor} at each neighbor of $\frakb$, at most 2 agents are destroyed. So, this shows that $3(\Delta-1)$ agents are required to \textit{anchor} each neighbor of $C_1$ by the members of \texttt{LG}$_i$ (for any $i\ge 0$). So, in total $3\Delta+3$ agents are sufficient to execute \textsc{Graph\_PerpExplore-BBH-Home}, including \texttt{Marker}, 4 agents in \texttt{SG}$_0$, and one agent which is neither a \texttt{Marker} nor an \textit{anchor}, i.e., the one which perpetually explores $C_1$. Moreover, to execute \textsc{BFS-Tree-Construction} as stated in  \cite[Proposition~9]{chalopin2010constructing}, each agent requires $\calO(n\Delta\log n)$ memory. This concludes the proof. 
\end{proof}

\subsection{Perpetual exploration in presence of a black hole}

In the special case in which the Byzantine black hole is activated in each round, i.e., behaves as a classical black hole, we show that the optimal number of agents for perpetual exploration (we call this problem \textsc{PerpExploration-BH}) drops drastically to between~$\Delta+1$ and~$\Delta+2$ agents.

The lower bound holds even with initial knowledge of~$n$, and even if we assume that agents have knowledge of the structure of the graph, minus the position of the black hole and the local port labelings at nodes.  This can be seen as an analogue, in our model, of the well-known $\Delta+1$ lower bound for black hole search in asynchronous networks~\cite{DBLP:journals/dc/DobrevFPS06}.


 Details of the proof and a discussion of the algorithm that solves the problem with~$\Delta+2$ agents can be found in Appendix~\ref{Appendix: PerpExploration-BH}.

\section{Conclusion}
 We gave the first non-trivial upper and lower bounds on the optimal number of agents for perpetual exploration, in presence of at most one Byzantine black hole in general, unknown graphs.

 One noteworthy point, as regards the related problem of \emph{pinpointing} the location of the malicious node, is that all our algorithms have the property that, even in an execution in which the Byzantine black hole destroys only one agent, all remaining agents manage to determine exactly the position of the malicious node, at least within the component which ends up being perpetually explored. By contrast, this is not the case for the algorithms proposed in~\cite{DBLP:conf/opodis/GoswamiBD024} for \pbmPerpExpl\ in ring networks, as the adversary can time a single activation of the Byzantine black hole, destroying at least one agent, so that the remaining agents manage to perpetually explore the ring, but without ever being able to disambiguate which one of two candidate nodes is the actual malicious node.


A few natural open problems remain. First, close the important gap between~$2\Delta-1$ and~$3\Delta+3$ for the optimal number of agents required for \pbmPerpExpl\ and \pbmPerpExplHome, in general graphs. Second, note that our general graph lower bound of $2\Delta-1$ holds only for graphs of maximum degree at least~$4$ (Theorem~\ref{thm: equivalent statement}). Could there be a $4$-agent algorithm for graphs of maximum degree~$3$? Third, in the special case of a black hole (or, equivalently, if we assume that the Byzantine black hole is activated in each round), close the gap between~$\Delta+1$ and~$\Delta+2$ agents. Fourth, our 4-agent algorithm in paths induced a direct improvement of the optimal number of agents for perpetual exploration in a ring with known~$n$, under the face-to-face communication model, from~5 agents (which was shown in~\cite{DBLP:conf/opodis/GoswamiBD024}) to~4 agents. However, it is still unknown whether the optimal number of agents is~3 or~4 in this case.


Finally, it would be interesting to consider stronger or weaker variants of the Byzantine black hole, and explore the tradeoffs between the power of the malicious node and the optimal number of agents for perpetual exploration. Orthogonally to this question, one can consider different agent communication models or an asynchronous scheduler.

\bibliography{Fullversion}

\begin{thebibliography}{10}

\bibitem{DBLP:journals/tcs/BampasLMPP15}
Evangelos Bampas, Nikos Leonardos, Euripides Markou, Aris Pagourtzis, and Matoula Petrolia.
\newblock Improved periodic data retrieval in asynchronous rings with a faulty host.
\newblock {\em Theor. Comput. Sci.}, 608:231--254, 2015.
\newblock URL: \url{https://doi.org/10.1016/j.tcs.2015.09.019}, \href {https://doi.org/10.1016/J.TCS.2015.09.019} {\path{doi:10.1016/J.TCS.2015.09.019}}.

\bibitem{chalopin2010constructing}
J{\'e}r{\'e}mie Chalopin, Shantanu Das, and Adrian Kosowski.
\newblock Constructing a map of an anonymous graph: Applications of universal sequences.
\newblock In {\em Principles of Distributed Systems: 14th International Conference, OPODIS 2010, Tozeur, Tunisia, December 14-17, 2010. Proceedings 14}, pages 119--134. Springer, 2010.

\bibitem{CzyzowiczComplexityBHS}
Jurek Czyzowicz, Dariusz Kowalski, Euripides Markou, and Andrzej Pelc.
\newblock Complexity of searching for a black hole.
\newblock {\em Fundamenta Informaticae}, 71(2-3):229--242, 2006.

\bibitem{CzyzowiczBHSSyncTree}
Jurek Czyzowicz, Dariusz Kowalski, Euripides Markou, and Andrzej Pelc.
\newblock Searching for a black hole in synchronous tree networks.
\newblock {\em Combinatorics, Probability and Computing}, 16(4):595--619, 2007.

\bibitem{DiLunaBHSDynamicRing}
Giuseppe~Antonio Di~Luna, Paola Flocchini, Giuseppe Prencipe, and Nicola Santoro.
\newblock Black hole search in dynamic rings.
\newblock In {\em 2021 IEEE 41st International Conference on Distributed Computing Systems (ICDCS)}, pages 987--997. IEEE, 2021.

\bibitem{DPP14}
Yoann Dieudonn{\'{e}}, Andrzej Pelc, and David Peleg.
\newblock Gathering despite mischief.
\newblock {\em {ACM} Trans. Algorithms}, 11(1):1:1--1:28, 2014.
\newblock \href {https://doi.org/10.1145/2629656} {\path{doi:10.1145/2629656}}.

\bibitem{DobrevDangerousGraphTokens}
Stefan Dobrev, Paola Flocchini, Rastislav Kr{\'a}lovi{\v{c}}, and Nicola Santoro.
\newblock Exploring an unknown dangerous graph using tokens.
\newblock {\em Theoretical Computer Science}, 472:28--45, 2013.

\bibitem{DBLP:journals/dc/DobrevFPS06}
Stefan Dobrev, Paola Flocchini, Giuseppe Prencipe, and Nicola Santoro.
\newblock Searching for a black hole in arbitrary networks: optimal mobile agents protocols.
\newblock {\em Distributed Comput.}, 19(1):1--35, 2006.
\newblock URL: \url{https://doi.org/10.1007/s00446-006-0154-y}, \href {https://doi.org/10.1007/S00446-006-0154-Y} {\path{doi:10.1007/S00446-006-0154-Y}}.

\bibitem{DobrevAnonymousRingBHS}
Stefan Dobrev, Paola Flocchini, Giuseppe Prencipe, and Nicola Santoro.
\newblock Mobile search for a black hole in an anonymous ring.
\newblock {\em Algorithmica}, 48:67--90, 2007.

\bibitem{DobrevBHSUnorientedRingScattered}
Stefan Dobrev, Nicola Santoro, and Wei Shi.
\newblock Locating a black hole in an un-oriented ring using tokens: The case of scattered agents.
\newblock In {\em European Conference on Parallel Processing}, pages 608--617. Springer, 2007.

\bibitem{FlocchiniPingPongBHSPebble}
Paola Flocchini, David Ilcinkas, and Nicola Santoro.
\newblock Ping pong in dangerous graphs: Optimal black hole search with pebbles.
\newblock {\em Algorithmica}, 62:1006--1033, 2012.

\bibitem{fraigniaud2005graph}
Pierre Fraigniaud, David Ilcinkas, Guy Peer, Andrzej Pelc, and David Peleg.
\newblock Graph exploration by a finite automaton.
\newblock {\em Theoretical Computer Science}, 345(2-3):331--344, 2005.

\bibitem{DBLP:conf/opodis/GoswamiBD024}
Pritam Goswami, Adri Bhattacharya, Raja Das, and Partha~Sarathi Mandal.
\newblock Perpetual exploration of a ring in presence of byzantine black hole.
\newblock In Silvia Bonomi, Letterio Galletta, Etienne Rivi{\`{e}}re, and Valerio Schiavoni, editors, {\em 28th International Conference on Principles of Distributed Systems, {OPODIS} 2024, December 11-13, 2024, Lucca, Italy}, volume 324 of {\em LIPIcs}, pages 17:1--17:17. Schloss Dagstuhl - Leibniz-Zentrum f{\"{u}}r Informatik, 2024.
\newblock URL: \url{https://doi.org/10.4230/LIPIcs.OPODIS.2024.17}, \href {https://doi.org/10.4230/LIPICS.OPODIS.2024.17} {\path{doi:10.4230/LIPICS.OPODIS.2024.17}}.

\bibitem{IntrusionDetection}
Wayne~A Jansen.
\newblock Intrusion detection with mobile agents.
\newblock {\em Computer Communications}, 25(15):1392--1401, 2002.

\bibitem{DBLP:conf/sirocco/KralovicM10}
Rastislav Kr{\'{a}}lovic and Stanislav Mikl{\'{\i}}k.
\newblock Periodic data retrieval problem in rings containing a malicious host.
\newblock In Boaz Patt{-}Shamir and T{\'{\i}}naz Ekim, editors, {\em Structural Information and Communication Complexity, 17th International Colloquium, {SIROCCO} 2010, Sirince, Turkey, June 7-11, 2010. Proceedings}, volume 6058 of {\em Lecture Notes in Computer Science}, pages 157--167. Springer, 2010.
\newblock \href {https://doi.org/10.1007/978-3-642-13284-1\_13} {\path{doi:10.1007/978-3-642-13284-1\_13}}.

\bibitem{DBLP:journals/eatcs/Markou12}
Euripides Markou.
\newblock Identifying hostile nodes in networks using mobile agents.
\newblock {\em Bull. {EATCS}}, 108:93--129, 2012.
\newblock URL: \url{http://eatcs.org/beatcs/index.php/beatcs/article/view/52}.

\bibitem{DBLP:series/lncs/MarkouS19}
Euripides Markou and Wei Shi.
\newblock Dangerous graphs.
\newblock In Paola Flocchini, Giuseppe Prencipe, and Nicola Santoro, editors, {\em Distributed Computing by Mobile Entities, Current Research in Moving and Computing}, volume 11340 of {\em Lecture Notes in Computer Science}, pages 455--515. Springer, 2019.
\newblock \href {https://doi.org/10.1007/978-3-030-11072-7\_18} {\path{doi:10.1007/978-3-030-11072-7\_18}}.

\bibitem{:Miklik2010}
Stanislav Mikl{\'i}k.
\newblock {\em Exploration in faulty networks}.
\newblock PhD thesis, Comenius University, Bratislava, 2010.

\bibitem{WeiShiSurveyBHS}
Mengfei Peng, Wei Shi, Jean-Pierre Corriveau, Richard Pazzi, and Yang Wang.
\newblock Black hole search in computer networks: State-of-the-art, challenges and future directions.
\newblock {\em Journal of Parallel and Distributed Computing}, 88:1--15, 2016.

\bibitem{Shannon-First-Exploration}
Claude~E Shannon.
\newblock Presentation of a maze-solving machine.
\newblock {\em Claude Elwood Shannon Collected Papers}, pages 681--687, 1993.

\end{thebibliography}

\clearpage
\appendix

\section{Preliminary notions} \label{Appendix:prelims}

In this section, we discuss some preliminary notions, used to prove Theorem~\ref{thm:impossibility with 3 PerpExplorationBBH Path} in Appendix~\ref{Appendix: Proof of impossibility with 3 PerpExplorationBBH Path} and Theorem~\ref{thm:impossibility with 5 PerpExplorationBBH-Home Path} in Appendix~\ref{Appendix: Proof of mpossibility with 5 PerpExplorationBBH-Home Path}.

For the following, fix a \pbmPerpExpl\ instance $I=\left\langle G,k,h,\frakb \right\rangle$, where $G=(V,E,\lambda)$, let $\calA$ be an algorithm, and fix an execution~$\calE$ of~$\calA$ on~$I$.

\begin{definition}[Destruction Time]\label{definition: destruction time}
 The \emph{destruction time}, denoted by $T_d$, is the first round in which the adversary activates the Byzantine black hole and destroys an agent.
\end{definition}

Note that, for a given execution, $T_d$ can be infinite if $\frakb=\bot$, or if the adversary never chooses to activate the black hole, or if it never destroys any agent.

\begin{definition}[Benign execution]\label{def:benign}
    We say that an execution of~$\calA$ on~$I$ is \emph{benign up to time~$T$} if $T_d > T$. We say that it is \emph{benign} if $T_d$ is infinite. 
\end{definition}

\begin{definition}[Benign continuation]
    The \emph{benign continuation from time~$t_0$} of an execution~$\calE$ of~$\calA$ on~$I$ is an execution which is identical to~$\calE$ up to time~$t_0$, and for~$t > t_0$ the Byzantine black hole is never activated. 
\end{definition}

Note that agents are deterministic and the system is synchronous. Therefore, if~$\frakb=\bot$, then there exists a unique execution of~$\calA$ on~$I$, which, in addition, is benign. If the system does contain a Byzantine black hole, then there exist different executions for different sequences of actions by the adversary. However, the benign execution is still unique in this case.

\begin{definition}[History of an agent]
The \emph{history} of an agent~$a$ at time~$t\geq 1$ is the finite sequence of inputs received by agent~$a$, at the beginning of rounds~$1,\dots,t$.
\end{definition}

\begin{definition}[Consistent instances] \label{def:consistent}
An instance~$I'$ is said to be $(\calE,I,t)$-consistent if there exists an execution~$\calE'$ of~$\calA$ on~$I'$ such that:
\begin{itemize}
    \item At the beginning of round~$t$, the sets of histories of all alive agents in~$\calE$ and in~$\calE'$ are identical.
    \item If $\tilde{\calE}$ and $\tilde{\calE}'$ are the infinite benign continuations of~$\calE$ and~$\calE'$, respectively, from time~$t$, then the suffixes of~$\tilde{\calE}$ and~$\tilde{\calE}'$ starting from time~$t$ are identical.    
\end{itemize}

We denote by $\cons(\calE,I,t)$ the set of all $(\calE,I,t)$-consistent instances.
\end{definition}

The usefulness of Definition~\ref{def:consistent} lies in the fact that agents with a certain history in a given execution in a given instance can be swapped into any other consistent instance and, in a benign continuation, behave exactly as they would in the original instance. We use this in the lower bound proofs of Appendices~\ref{Appendix: Proof of impossibility with 3 PerpExplorationBBH Path} and~\ref{Appendix: Proof of mpossibility with 5 PerpExplorationBBH-Home Path}.

\section{Proof of Theorem \ref{thm:impossibility with 3 PerpExplorationBBH Path}}\label{Appendix: Proof of impossibility with 3 PerpExplorationBBH Path}

In this section we establish the basic impossibility result for \pbmPerpExpl~in turn prove Theorem \ref{thm:impossibility with 3 PerpExplorationBBH Path}: Three agents cannot solve \pbmPerpExpl\ in sufficiently large paths, even assuming knowledge of the size of the graph (Theorem~\ref{thm:3 agents not sufficient}).

\begin{lemma} \label{lem:impossibility:oneagent-twosusp}
    If $\calA$ solves \pbmPerpExpl\ with~$k_0$ agents, then for every instance~$I$ on a path graph with~$k\geq k_0$ agents, there is no execution~$\calE$ of~$\calA$ on~$I$ such that all of the following hold:
    \begin{itemize}
        \item At some time~$t \geq 1$, exactly one agent~$a$ remains alive.
        \item $\cons(\calE,I,t)$ contains instances $I_1=\left\langle P,h,k,\frakb_1\right\rangle$ and~$I_2=\left\langle P,h,k,\frakb_2\right\rangle$, where~$P$ is a port-labeled path and $\frakb_1,\frakb_2$ are two distinct nodes of~$P$.
    \end{itemize}
\end{lemma}

\begin{proof}
For a contradiction, let~$\calE$ be an execution of~$\calA$ on~$I$ that satisfies both conditions simultaneously at time~$t\geq 1$, and let~$h_a$ be the history of the only remaining alive agent~$a$ at time~$t$. 
By consistency, there exist executions~$\calE_1$ on~$I_1$ and~$\calE_2$ on~$I_2$ such that, under both, the history of~$a$ at time~$t$ is exactly~$h_a$. Moreover, the benign continuations~$\tilde{\calE}_1,\tilde{\calE}_2$ share a common suffix starting from time~$t$.

We distinguish the following cases:
\begin{itemize}
    \item In $\tilde{\calE}_1$ and $\tilde{\calE}_2$, at some $t'\geq t$, agent~$a$ visits~$\frakb_1$ (resp.\ $\frakb_2$). Then, the algorithm fails in instance~$I_1$ (resp.\ $I_2$) by activating the Byzantine black hole at time~$t'$.

    \item In $\tilde{\calE}_1$ and $\tilde{\calE}_2$,  agent~$a$ never visits~$\frakb_1$ or $\frakb_2$ at any time~$t'\geq t$. Let~$p$ be the position of the agent at time~$t$. If~$\frakb_1$ and~$\frakb_2$ are on the same side of~$p$, and $\frakb_1$ (resp.\ $\frakb_2$) is the farthest from~$p$, then the algorithm fails on instance~$I_1$ (resp.\ $I_2$) because, after time~$t$, it never visits node~$\frakb_2$ (resp.\ $\frakb_1$), which is in the same component as the agent and is not the Byzantine black hole. If $p$ is between $\frakb_1$ and~$\frakb_2$, then the algorithm fails on both instances~$I_1$ and~$I_2$. To see why, consider the instance~$I_1$: In~$I_1$, $\frakb_2$ is not the Byzantine black hole and it is in the same component as the agent, but is never visited after time~$t$. The reasoning is similar for instance~$I_2$.
\end{itemize}
In all cases, we obtain a contradiction with the correctness of algorithm~$\calA$.
\end{proof}

\begin{lemma} \label{lemma: 2 agents are not sufficient}
If $\calA$ solves \pbmPerpExpl\ with~$k_0$ agents, then for every instance~$I$ on a path graph with~$k\geq k_0$ agents, there is no execution~$\calE$ of~$\calA$ on~$I$ such that all of the following hold:
    \begin{itemize}
        \item At some time~$t \geq 1$, exactly two agents~$a,b$ remain alive.
        \item $\cons(\calE,I,t)$ contains instances~$I_i=\left\langle P,h,k,\frakb_i\right\rangle$, $1\leq i\leq 3$, where~$P$ is a port-labeled path and $\{\frakb_i: 1\leq i\leq 3\}$ are three distinct nodes of~$P$. Let~$\calB$ be the common suffix, starting from time~$t$, of all benign continuations of the executions witnessing the $(\calE,I,t)$-consistency of~$I_i$, for $1\leq i\leq 3$.
        \item If $p_a,p_b$ are the positions of the agents at time~$t$ under $\calB$, then $p_a,p_b$ are on the same side of all of~$\{\frakb_i: 1\leq i\leq 3\}$.
    \end{itemize}
\end{lemma}

\begin{proof}
    For a contradiction, let~$\calE$ be an execution of~$\calA$ on~$I$ that satisfies all of the conditions simultaneously at time~$t\geq 1$, let~$\frakb_1,\frakb_2,\frakb_3$ be ordered by increasing distance from~$p_a$ and~$p_b$, and
let $P_1$ denote the subpath of~$P$ from~$\frakb_1$ to an extremity of~$P$, such that~$P_1$ does not contain any of~$p_a$, $p_b$.

We distinguish the following cases:

\emph{Case 1}: In $\tilde{\calE}_1$, $\tilde{\calE}_2$, and $\tilde{\calE}_3$, there exists some time~$t'\geq t$ such that agent~$a$ (without loss of generality) enters~$P_1$ (hence moves to~$\frakb_1$ from outside~$P_1$), and agent~$b$ moves to, from, or stays on some node of~$P_1$.

In this case, consider the smallest such~$t'$, and let~$t'_1\leq t'$ be the last time that agent~$b$ enters~$P_1$ before~$t'$. Note that agent~$a$ enters~$P_1$ at time~$t'$ without waiting for agent~$b$ to exit~$P_1$. Therefore, even if agent~$b$ is destroyed by~$\frakb_1$ at time~$t'_1$, agent~$a$ will still enter~$P_1$ at time~$t'$. We obtain a contradiction with the correctness of algorithm~$\calA$, as it fails on instance~$I_1$ if the Byzantine black hole~$\frakb_1$ is activated at times~$t'$ and~$t'_1$.

\emph{Case 2}: In $\tilde{\calE}_1$, $\tilde{\calE}_2$, and $\tilde{\calE}_3$, whenever an agent enters~$P_1$ at some time~$t'\geq t$, then at that time the other agent is not inside~$P_1$, nor is entering or exiting~$P_1$.

In this case, there must exist at least one time~$t'\geq t$ such that agent~$a$ (without loss of generality) enters~$P_1$ and moves on to reach~$\frakb_2$ before exiting~$P_1$, while agent~$b$ remains outside of~$P_1$. Indeed, if this never happens, then the algorithm fails on instance~$I_3$ as it never explores node~$\frakb_2$.

Let~$t'_1>t'$ be the first time after~$t'$ such that agent~$a$ visits~$\frakb_2$. Now, consider the following executions:
\begin{itemize}
    \item $\calE'_1$ on $I_1$: same as~$\tilde{\calE}_1$ up to time~$t'$, then $\frakb_1$ is activated at time~$t'$ and destroys agent~$a$, then $\frakb_1$ is never activated after~$t'$.
    \item $\calE'_2$ on $I_2$: same as~$\tilde{\calE}_2$ up to time~$t'_1$, then $\frakb_2$ is activated at time~$t'_1$ and destroys agent~$a$, then $\frakb_2$ is never activated after~$t'_1$.
\end{itemize}
Clearly, after time~$t'$ and as long as agent~$a$ hasn't exited~$P_1$, agent~$b$ receives exactly the same inputs in both executions. Therefore, the history of agent~$b$ is identical in both executions up to at least time~$t'_1$. Let~$h'_b$ be the history of agent~$b$ at that time. Moreover, the benign continuations of both~$\calE'_1$ and~$\calE'_2$ after time~$t'_1$ must be identical.
It follows that $\cons(\calE'_1,I_1,t'_1)$ contains~$I_1$ and~$I_2$. This contradicts Lemma~\ref{lem:impossibility:oneagent-twosusp}.
\end{proof}

\begin{lemma} \label{lemma: 2 agents are not sufficient enclosing susp}
If $\calA$ solves \pbmPerpExpl\ with~$k_0$ agents, then for every instance~$I$ on a path graph with~$k\geq k_0$ agents, there is no execution~$\calE$ of~$\calA$ on~$I$ such that all of the following hold:
    \begin{itemize}
        \item At some time~$t \geq 1$, exactly two agents~$a,b$ remain alive.
        \item $\cons(\calE,I,t)$ contains instances~$I_i=\left\langle P,h,k,\frakb_i\right\rangle$, $1\leq i\leq 3$, where~$P$ is a port-labeled path and $\{\frakb_i: 1\leq i\leq 3\}$ are three distinct nodes of~$P$. Let~$\calB$ be the common suffix, starting from time~$t$, of all benign continuations of the executions witnessing the $(\calE,I,t)$-consistency of~$I_i$, for $1\leq i\leq 3$.
        \item If $p_a,p_b$ are the positions of the agents at time~$t$ under~$\calB$, then the path connecting $p_a$ to $p_b$ contains all of~$\{\frakb_i: 1\leq i\leq 3\}$.
    \end{itemize}
\end{lemma}

\begin{proof}
          For a contradiction, let~$\calE$ be an execution of~$\calA$ on~$I$ that satisfies all of the conditions simultaneously at time~$t\geq 1$, and let~$\frakb_1,\frakb_2,\frakb_3$ be ordered by increasing distance from~$p_a$ (and, therefore, decreasing distance from~$p_b$).

    \begin{claim}
        In~$\calB$, no agent ever visits~$\frakb_2$ after time~$t$.
    \end{claim}

    \begin{claimproof}
        Suppose the contrary and, without loss of generality, let~$a$ be the first agent to visit~$\frakb_2$ at time~$t_2\geq t$. As the two agents are initially on either side of~$\frakb_2$, it follows that the agents do not meet before~$t_2$. Moreover, in order to reach~$\frakb_2$, agent~$a$ must also visit~$\frakb_1$ at some time~$t_1<t_2$. By activating the Byzantine black hole at times~$t_1,t_2$ in the respective instances, agent~$a$ may be destroyed on either of~$\frakb_1,\frakb_2$. Let~$\calB'$ be the same as~$\calB$ up to~$t_2$, when~$\frakb_2$ is activated once to destroy agent~$a$ and then remains benign. By the lack of communication between agents, $\cons(\calB',I_2,t_2)$, contains~$I_1,I_2$, which contradicts Lemma~\ref{lem:impossibility:oneagent-twosusp}.
    \end{claimproof}

    A corollary of the above claim is that agents never meet after time~$t$. We can further show that agent~$a$ never visits~$\frakb_1$ and, by symmetry, agent~$b$ never visits~$\frakb_3$. Indeed, if agent~$a$ ever visits~$\frakb_1$, then the algorithm fails on instance~$I_1$ by destroying agent~$a$. The remaining agent~$b$ will never visit~$\frakb_2$, thus failing to perpetually explore its component.
    We conclude that no agent ever visits any of~$\frakb_1,\frakb_2,\frakb_3$ after time~$t$, and therefore the algorithm fails on instance~$I_2$ because no agent perpetually explores its component.
\end{proof}

\begin{lemma} \label{lemma: 3 agents are not sufficient-2simult}
If $\calA$ solves \pbmPerpExpl\ with~$k_0$ agents, then for every instance~$I$ on a path graph with~$k\geq k_0$ agents, there is no execution~$\calE$ of~$\calA$ on~$I$ such that all of the following hold:
    \begin{itemize}
        \item At some time~$t \geq 1$, exactly three agents~$a,b,c$ remain alive.
        \item $\cons(\calE,I,t)$ contains instances~$I_i=\left\langle P,h,k,\frakb_i\right\rangle$, $1\leq i\leq 7$, where~$P$ is a port-labeled path and $\{\frakb_i : 1\leq i\leq 7\}$ are seven distinct nodes of~$P$. Let~$\calB$ be the common suffix, starting from time~$t$, of all benign continuations of the executions witnessing the $(\calE,I,t)$-consistency of $I_i$, for $1\leq i\leq 7$.
        \item If $p_a,p_b,p_c$ are the positions of the agents at time~$t$ under~$\calB$, then $p_a,p_b,p_c$  are on the same side of all of~$\{\frakb_i : 1\leq i\leq 7\}$.
        \item Assuming without loss of generality that nodes~$\frakb_i$ are ordered by increasing distance from~$p_a,p_b,p_c$, and letting~$P_1$ denote the subpath of~$P$ from~$\frakb_1$ to an extremity of~$P$ such that~$P_1$ does not contain~$p_a,p_b,p_c$, then in~$\calB$, at most two agents are ever simultaneously in~$P_1$.
    \end{itemize}
\end{lemma}

\begin{proof}
      For a contradiction, let~$\calE$ be an execution of~$\calA$ on~$I$ that satisfies all of the conditions simultaneously at time~$t\geq 1$, and let, for $i\geq 2$, $P_i$ denote the subpath of~$P_1$ spanning the nodes from~$\frakb_i$ to an extremity of~$P$. We say that a $P_i$-event occurs at time~$t_0\geq t$ ($1\leq i\leq 7$), if an agent enters~$P_i$ at time~$t_0$ while another agent is already in~$P_i$.
      
\begin{claim}
    In~$\calB$, a $P_4$-event must occur at some time after~$t$.
\end{claim}

\begin{claimproof}
Indeed, suppose that in $\calB$, at most one agent is ever inside~$P_4$ after time~$t$. Then, there must exist some time~$t'_1\geq t$ at which agent~$a$ (without loss of generality) enters~$P_4$ and does not exit~$P_4$ before reaching~$\frakb_6$ at time~$t'_2\geq t'_1$, while the other agents stay outside of~$P_4$. Agent~$a$ may be destroyed on any of~$\frakb_4,\frakb_5,\frakb_6$ by activating the Byzantine black hole in the respective instance. It follows that $I_4,I_5,I_6$ are all in $\cons(\calB',I_6,t'_2)$, where $\calB'$ is the same as~$\calB$ up to time~$t'_2$, when $\frakb_6$ is activated. This contradicts Lemma~\ref{lemma: 2 agents are not sufficient}.    
\end{claimproof}

Note that a $P_4$-event requires a prior $P_3$-event. Let~$t_4$ be the time of the first~$P_4$-event in~$\calB$, involving agents~$a,b$ without loss of generality. Let $t_3\leq t_4$ be the time of the last~$P_3$-event before~$t_4$. By definition, the~$P_3$-event at time~$t_3$ must involve the same agents~$a,b$ and none of these agents may have exited~$P_3$ between~$t_3$ and~$t_4$. Without loss of generality, let agent~$b$ be the one that enters~$P_3$ at time~$t_3$. Therefore, this is the last time that agent~$b$ enters~$P_3$ before~$t_4$. Let~$t_1\leq t_3$ denote the last time agent~$a$ entered~$P_3$.

\begin{claim} \label{clm:c1}
    In~$\calB$, agent~$c$ does not enter~$P_3$ between $t_1$ and $t_3$.
\end{claim}

\begin{claimproof}
    Suppose that agent~$c$ enters~$P_3$ at time~$t'_1$ ($t_1\leq t'_1 < t_3$). Given that agent~$b$ enters~$P_3$ at time~$t_3$, and by the assumption that at most two agents may ever be in~$P_1$, it follows that agent~$c$ must reach~$\frakb_1$ at some time~$t'_2$ between~$t'_1$ and~$t_3$, while agent~$a$ remains in~$P_3$ and agent~$b$ remains outside of~$P_1$. On the way out, agent~$c$ may be destroyed by any of~$\frakb_1,\frakb_2,\frakb_3$ by activating the Byzantine black hole in the respective instance. It follows that~$I_1,I_2,I_3$ are all in~$\cons(\calB',I_1,t'_2)$, where $\calB'$ is the same as~$\calB$ up to time~$t'_2$, when~$\frakb_1$ is activated. This contradicts Lemma~\ref{lemma: 2 agents are not sufficient enclosing susp}.
\end{claimproof}

It follows that only agent~$b$ may enter~$P_3$ between~$t_1$ and~$t_3$. Let~$t_2\geq t_1$ be the first time this happens. Then, agent~$b$ may not exit~$P_1$ after time~$t_2$, by a similar argument as in the proof of Claim~\ref{clm:c1}. There is, therefore, no communication after time~$t_1$ between agent~$c$ and either of agents~$a,b$. Finally, let~$t'_4\leq t_4$ be the time when the agent that is already in~$P_4$ at time~$t_4$ enters~$P_4$ for the last time. Now, consider the following executions:
\begin{itemize}
    \item $\calB_1$ on~$I_3$: same as~$\calB$ up to time~$t_1$, when~$\frakb_3$ is activated destroying agent~$a$. Then, benign up to time~$t_3$, when~$\frakb_3$ is activated destroying agent~$b$. Then, $\frakb_3$ is never activated after~$t_3$.
    \item $\calB_2$ on~$I_4$: same as~$\calB$ up to time~$t'_4$, when~$\frakb_4$ is activated. Then, benign up to time~$t_4$, when~$\frakb_4$ is activated. At this point, both agents~$a,b$ are destroyed and $\frakb_4$ is never activated after~$t_4$.
\end{itemize}
Clearly, $\cons(\calB_1,I_3,t_4)$ contains both of~$I_3,I_4$, which contradicts Lemma~\ref{lem:impossibility:oneagent-twosusp}.
\end{proof}

\begin{lemma} \label{lemma: 3 agents fail}
If $\calA$ solves \pbmPerpExpl\ with~$k_0$ agents, then for every instance~$I$ on a path graph with~$k\geq k_0$ agents, there is no execution~$\calE$ of~$\calA$ on~$I$ such that all of the following hold:
    \begin{itemize}
        \item At some time~$t \geq 1$, exactly three agents~$a,b,c$ remain alive.
        \item $\cons(\calE,I,t)$ contains instances~$I_i=\left\langle P,h,k,\frakb_i\right\rangle$, $1\leq i\leq 8$, where~$P$ is a port-labeled path and $\{\frakb_i : 1\leq i\leq 8\}$ are eight distinct nodes of~$P$. 
        \item If $p_a,p_b,p_c$ are the positions of the agents at time~$t$ under~$\calE$, then $p_a,p_b,p_c$  are on the same side of all of~$\{\frakb_i : 1\leq i\leq 8\}$.
        \item Assuming without loss of generality that nodes~$\frakb_i$ are ordered by increasing distance from~$p_a,p_b,p_c$, nodes $\frakb_1,\frakb_2,\frakb_3$ are consecutive in~$P$.
    \end{itemize}
\end{lemma}

\begin{proof}
      For a contradiction, let~$\calE$ be an execution of~$\calA$ on~$I$ that satisfies all of the conditions simultaneously at time~$t\geq 1$. Let~$\calB$ be the common suffix, starting from time~$t$, of all benign continuations of the executions witnessing the $(\calE,I,t)$-consistency of $I_i$, for $1\leq i\leq 8$. Let, for $i\geq 1$, $P_i$ denote the subpath of~$P$ spanning the nodes from~$\frakb_i$ to an extremity of~$P$, such that~$P_i$ does not contain any of~$p_a,p_b,p_c$. We say that a $P^j_i$-event occurs at time~$t_0\geq t$ ($1\leq i\leq 8$, $1\leq j\leq 3$), if some agent enters~$P_i$ at time~$t_0$ and at least $j$ agents  are in~$P_i$ at time~$t_0$ (including the one or the ones that entered at time~$t_0$).

      By Lemma~\ref{lemma: 3 agents are not sufficient-2simult}, a $P_2^3$-event must occur in~$\calB$. Let~$T$ be the time of the first $P_2^3$-event, and assume, without loss of generality, that agent~$c$ enters~$P_2$ at time~$T$ while agents~$a,b$ are in~$P_2$.

      \begin{claim} \label{clm:c2}
          In~$\calB$, there is no~$P_3^2$-event involving agents~$a,b$ at or before time~$T$.
      \end{claim}

      \begin{claimproof}
          For a contradiction, let~$t_4\leq T$ be the time of the last~$P_3^2$-event before time~$T$, and assume, without loss of generality, that agent~$a$ enters~$P_3$ while agent~$b$ is in~$P_3$. Let~$t_3\leq t_4$ be the time of the last~$P_2^2$-event before~$t_4$ involving agents~$a,b$.  Without loss of generality, let agent~$a$ be the one that enters~$P_2$ at time~$t_3$, and let~$t_1\leq t_3$ be the last time agent~$b$ entered~$P_2$.

          Consider the following executions:
          \begin{itemize}
              \item $\calB_1$ on~$I_1$: same as~$\calB$ up to time~$t_3$, when~$\frakb_1$ is activated destroying agent~$b$. Then, benign up to time~$T$, when~$\frakb_1$ is activated destroying agent~$c$. Then, $\frakb_1$ is never activated after~$T$.
              \item $\calB_2$ on~$I_2$: same as~$\calB$ up to time~$t_3$, when~$\frakb_1$ is activated destroying agent~$b$. Then, benign up to time~$T$, when~$\frakb_1$ is activated destroying agent~$c$. Then, $\frakb_2$ is never activated after~$T$.
          \end{itemize}
          Clearly, $\cons(\calB_1,I_1,T)$ contains both of~$I_1,I_2$, which contradicts Lemma~\ref{lem:impossibility:oneagent-twosusp}.
      \end{claimproof}

      By Claim~\ref{clm:c2}, either none of $a,b$ has entered~$P_3$ until~$T$, or only one of~$a,b$ has entered~$P_3$ until~$T$. We complete the proof by distinguishing the following cases:

      \emph{Case 1: None of $a,b$ has entered~$P_3$ until time~$T$}. Note that, by assumption, a $P_2^3$-event occurs at time~$T$. It follows that, in that round, all agents are either on, moving to, or moving from node~$\frakb_2$. Therefore, they can all be destroyed in instance~$I_2$ by activating~$\frakb_2$ at time~$T$.

      \emph{Case 2: Agent~$a$ (without loss of generality) enters or is in~$P_3$ at time~$T$, and agents $b,c$ enter simultaneously~$P_2$ at time~$T$}. Then, the adversary can construct executions~$\calE_1$ (resp.\ $\calE_2$) in which it follows~$\calB$ up to time~$T$, when it activates the Byzantine black hole $\frakb_1$ (resp.\ $\frakb_2$) in instance~$I_1$ (resp.\ $I_2$) to destroy both agents~$b,c$. $\cons(\calE_1,I_1,T)$ contains both of~$I_1,I_2$, which contradicts Lemma~\ref{lem:impossibility:oneagent-twosusp}.

      \emph{Case 3: Each of agents~$a,b$ is either on or moving to~$\frakb_2$ at time~$T$}. In this case, the adversary can destroy all of the agents in instance~$I_2$ by activating~$\frakb_2$ at time~$T$.

      \emph{Case 4: At round~$T$, agent~$a$ is at~$\frakb_3$, agent~$b$ is at~$\frakb_2$, and agent~$c$ moves to~$\frakb_2$ from~$\frakb_1$.} Let $t_a\leq T$ be the last time before~$T$ that $a$ moved from~$\frakb_2$ to~$\frakb_3$, and let $t_b\leq T$ be the last time before~$T$ that $b$ moved from~$\frakb_1$ to~$\frakb_2$. If $t_b \leq t_a$, then the adversary can activate~$\frakb_2$ in instance~$I_2$ between times~$t_a$ and~$T$ to destroy all the agents. If~$t_a<t_b$, then the adversary can construct executions~$\calE_1$ (resp.\ $\calE_2$) in which it follows~$\calB$ up to~$t_b$, and then it keeps activating the Byzantine black hole~$\frakb_1$ (resp.\ $\frakb_2$) in instance~$I_1$ (resp.\ $I_2$) to destroy both agents~$b,c$. We then have $I_1,I_2\in\cons(\calE_1,I_1,T)$, which contradicts Lemma~\ref{lem:impossibility:oneagent-twosusp}.
\end{proof}

 \begin{theorem} \label{thm:3 agents not sufficient}
    There is no algorithm that solves \pbmPerpExpl\ with~$3$ agents on all path graphs with at least~$9$ nodes, even assuming that the agents know the number of nodes.
\end{theorem}

\begin{proof}
    Let~$\calA$ be an algorithm that solves \pbmPerpExpl\ with~$3$ agents on all path graphs with at least~$9$ nodes. Fix a path~$P$ with $n\geq 9$ nodes. Let the home node~$h$ be one of the extremities of~$P$, and let~$v_i$, $1\leq i\leq 8$ be the node at distance~$i$ from~$h$. Let $3$ agents start at node~$h$ and consider the set of instances~$\calI=\{\left\langle P,h,3,v_i\right\rangle:1\leq i\leq 8\}$. Even with knowledge of~$n$, at the very beginning of the benign execution~$\calE$ of~$\calA$ (beginning of first round) in, say, $I_1$, every instance in~$\calI$ is contained in~$\cons(\calE,I_1,1)$. Moreover, $v_1,v_2,v_3$ are consecutive in~$P$. This contradicts Lemma~\ref{lemma: 3 agents fail}. 
\end{proof}

\section{Proof of Theorem \ref{thm:impossibility with 5 PerpExplorationBBH-Home Path}}\label{Appendix: Proof of mpossibility with 5 PerpExplorationBBH-Home Path}

 We further refine the result of the previous section to show that \pbmPerpExplHome\ is actually impossible even with 5 agents in sufficiently large paths, even assuming knowledge of the size of the graph (Theorem~\ref{thm:5 agents impossibility}).

In this section, we use the notion of \emph{suspicious nodes} to simplify the presentation of the proofs. The suspicious nodes are potential positions of the Byzantine black hole on the path. Formally, given an instance~$I=\left\langle P,k,h,\frakb\right\rangle$, where $P=(V,E,\lambda)$ is a port-labeled path graph, and an execution~$\calE$ on~$I$, the set~$S(t)$ of suspicious nodes at time~$t$ is defined as:
\[
S(t) = \{s:\left\langle P',k,h,s\right\rangle\in\cons(\calE,I,t) \text{ for some path~$P'$}\}
\]
Since the graph is assumed anonymous, in this definition a node is identified by its signed distance from the home node~$h$ in~$P$, by assuming that the positive (resp.\ negative) direction on the path is the direction of outgoing port number~$1$ (resp.\ $2$) from~$h$.

Note that, as long as the agents have not learned the size of the path and if the Byzantine black hole has not been activated yet, $S(t)$ is an infinite set. Indeed, $\cons(\calE,I,t)$ contains all arbitrarily long paths which are consistent with the finite history of the agents at time~$t$.

In this section, though, we will assume that the agents know the size of the path. This simplifies the definition of~$S(t)$ to:
\[
S(t) = \{s:\left\langle P,k,h,s\right\rangle\in\cons(\calE,I,t)\}
\]
A node is said to be \emph{safe} if it is in~$V\setminus S(t)$. At the beginning of an execution only the home node can be assumed safe, therefore $S(1)= V\setminus\{h\}$. In subsequent rounds, each agent keeps its own track of the suspicious node set, which can only be reduced. Indeed, if an agent detects the destruction of one or more agents, it can eliminate from its suspicious node set those nodes that could not house the Byzantine black hole. Note, however, that it may happen during an execution that one or more agents have an outdated idea of the suspicious node set, specifically if they have not been able yet to receive information from the agent(s) that first detected the destruction of an agent. In this case, the agents with outdated information must still continue to behave consistently with their presumed suspicious node set, until they can communicate with the agents having the most current information. We use the notation~$S_i(t)$ to denote the suspicious node set of agent~$i$ at time~$t$.

In all the following lemmas and theorems we assume that $n\ge 9$.

\begin{lemma}\label{lemma: 2 agents exploration cases LB}
        

For any exploration algorithm $\mathcal{A}$ with two agents $A_1$ and $A_2$ on a path graph $P=(V,E,\lambda)=L[v_0,v_{n-1}]$ containing a BBH, where $L[v_i,v_j]$ indicates the induced subgraph of $P$ induced by the nodes $\{v_x:0\le i\le x \le j\le |V|-1\}$. There exists a time $T_f\ge T_d$ such that  adversary can create one of the two following situations at $T_f$. 

\begin{enumerate}
    \item Both agents are destroyed.
    \item Exactly one agent, say $A_i$ stays alive and $|S_i(T_f)| \ge \lfloor\frac{n}{2}\rfloor$.
\end{enumerate}
\end{lemma}
\begin{proof}
    Let the path graph $P=(V,E,\lambda)=L[v_0,v_{n-1}]$ where $|V|=n$. Without loss of generality, we denote the home $h$ to be $v_0$. It may be noted that $S_i(T)=S_i(0)=L[v_1,v_{n-1}]$ where $i \in \{1, 2\}$ for all $0\le T\le T_d$. Let us consider the ordered sequence of line segments $L_{\lfloor\frac{n}{2}\rfloor} \subset L_{\lfloor\frac{n}{2}\rfloor-1}\subset\dots L_2\subset L_1\subset P$, where $L_i=L[v_i,v_{n-1}]$. Now we claim that if there is a round $T_i$ such that at $T_i$ both $A_1$ and $A_2$ are in $P_i$ then adversary can create $i$ many instances denoted by $I_j= \langle P,2,v_0,v_j\rangle$ where both agents are destroyed where $0 \le j \le i$ . 
    
    Let in the benign execution at some round $T_i$ both $A_1$ and $A_2$ are in $P_i$. So there is a round such that at $T_j$ both $A_1$ and $A_2$ are in $L_j$ where $j\le i$. Now let without loss of generality $A_1$ first enters $L_j$ before $A_2$. Let $T_j'$ be the last time when $A_1$ was on $v_j$ before $T_j$. So, by activating the BBH at $T_j'$ and $T_j$ adversary can destroy both $A_1$ and $A_2$ creating scenario 1 for all $j\le i$. Here $T_d=T_j'$ and $T_f=T_j$.

    So, if $i=\lfloor\frac{n}{2}\rfloor$ then, adversary can create $\lfloor\frac{n}{2}\rfloor$ many instances destroying each of the agents in all those instances given there is a time $T_{\lfloor\frac{n}{2}\rfloor}$ when both agents are in $L_{\lfloor\frac{n}{2}\rfloor}$. 
    
    Now, let for all $T$ exactly one agent, say $A_1$, remains outside of $L_{\lfloor\frac{n}{2}\rfloor}$. Due to the benign execution property, $A_2$ has to explore the whole path graph. So, there exists a time $T_*$ when $A_2$ is on $v_{n-1}$. Let $T_*'$ be the last time $A_2$ was on $v_{\lfloor\frac{n}{2}\rfloor}$ before $T_*$. Then adversary can create $\lfloor\frac{n}{2}\rfloor$ many instances denoted as $I_j^*=\langle P,2,v_0,v_j\rangle$ where $\lfloor\frac{n}{2}\rfloor \le j\le n-1$. In each of these instances $A_2$ can be destroyed at $\frakb$ at time $T_d$ by the adversary after $T_*'$ and $A_1$ can not distinguish between any of them. So, for the alive agent $A_1$, $S_1(T_f) \ge \lfloor\frac{n}{2}\rfloor$ where $T_f=T_d$.
\end{proof}

\begin{lemma}\label{lemma: consequence of 3 agents Path Graph}
    For any exploration algorithm $\mathcal{A}$ on a path graph $P=(V,E,\lambda)=L[v_0,v_{n-1}]$ with 3 agents $A_1$, $A_2$ and $A_3$, in presence of a BBH, there exists a time $T_f\ge T_d$ such that at $T_d$, the adversary can create any of the following three scenarios.
    \begin{enumerate}
        \item All three agents are destroyed.
        \item At least one agent stays alive at the component $C_1$ and no agents are in $C_2$ of the graph $P\setminus\{\frakb\}$ where $C_1$ contains the home $h$, $P-\frakb=C_1\cup C_2$ and for all alive agents $A_j$, $|S_j(T_f)| \ge \lfloor\frac{n}{4}\rfloor$.
        \item Exactly one agent $A_i$ stays alive at component $C_2$ of the graph $G -\{\frakb\}$, where $C_2$ does not contain the home and $|S_i(T_f)| \ge 2$.
    \end{enumerate}
    \label{lemma: 3 agents cases}
\end{lemma}
\begin{proof}
    Without loss of generality, let $h=v_0$, and we denote let $L_i=L[v_i,v_{n-1}]$ such that $V=\{v_0,\dots,v_{n-1}\}$. We first show that, if two agents, say $A_1$ and $A_2$ does not enter $L_{\lfloor\frac{n}{2}\rfloor}$ but $A_3$ enters it, then the adversary can create  scenario 2 at some time $T_f$, where two agents $A_1$ and $A_2$ stays alive at $C_1$ and $|S_i(T_f)|\ge \lfloor\frac{n}{2}\rfloor>\lfloor\frac{n}{4}\rfloor$, for $i\in \{1, 2\}$ and for some $T_f\ge T_d$. In this case, due to a benign execution property, there must exists a time $T_1$ when $A_3$ is at $v_{n-1}$, and let $T_1'$ be the last time before $T_1$, when $A_3$ was on $v_{\lfloor\frac{n}{2}\rfloor}$. Now, the adversary can create $\lfloor\frac{n}{2}\rfloor$ distinct instances denoted as $I_{\alpha}$ where $I_{\alpha}=\langle P, 3,v_0,v_{n-1-\alpha} \rangle$ and $0\le \alpha < \lfloor\frac{n}{2}\rfloor$ such that it can destroy $A_3$ in each of these instances at round $T_d$ between $T_1$ and $T_1'$. We consider $T_f=T_d$ here. This construction ensures that $A_2$ and $A_3$ can not distinguish between these instances. This implies $|S_i(T_f)| \ge \lfloor\frac{n}{2}\rfloor > \lfloor\frac{n}{4}\rfloor $ for all $i \in \{1,2\}$ where $A_1$ and $A_2$ are in $C_1$. 
    
    So, between $T_1'$ and $T_1$ if none among $A_1$ and $A_2$ enters $L_{\lfloor\frac{n}{2}\rfloor}$ then adversary can create a situation worse than scenario 2.

    So, let us assume between $T_1'$ and $T_1$ among $A_1$ and $A_2$ only $A_2$ moved to $L_{\lfloor\frac{n}{2}\rfloor}$ by being on $v_{\lfloor\frac{n}{2}\rfloor}$, for the first time at $T_2$ and $A_1$ never moves to $L_{\lfloor\frac{n}{2}\rfloor}$. 
Now, between $T_2$ and $T_1$ if $A_2$ never enters $L_{\lfloor\frac{3n}{4}\rfloor}$ then, again adversary can create $\lfloor\frac{n}{4}\rfloor$ instances $I_{\beta}= \langle P,3, v_0, v_{\lfloor\frac{3n}{4}+\beta\rfloor}\rangle$ where $0\le \beta \le \lfloor\frac{n}{4}\rfloor$, and destroy $A_3$ at $T_d$ for each of these instances between $T_1''$ and $T_1$, $T_1''$ being the last time $A_3$ was on $v_{\lfloor\frac{3n}{4}\rfloor}$ before $T_1$. Note that $A_2$ and $A_3$ can not distinguish between these instances due to the construction. So, in this way if we consider $T_f=T_d$,  then again adversary can create scenario 2 for this case.

Now let there exists a time $T_{x_0}$ between $T_1'$ and $T_1$, when both $A_2$ and $A_3$ are in $L_{\lfloor\frac{3n}{4}\rfloor}$. So, there exists $T_{x_i}$ in between $T_1'$ and $T_1$ such that both $A_2$ and $A_3$ are in $L_{\lfloor\frac{3n}{4}\rfloor-i}$ and $T_{x_{i+1}} < T_{x_i}$ for all $i$ where $0\le i \le \lfloor\frac{n}{4}\rfloor$. By similar argument as the second paragraph of Lemma~\ref{lemma: 2 agents exploration cases LB}, adversary can destroy both of $A_2$ and $A_3$ for each of the $\lfloor \frac{n}{4} \rfloor$ instances $I_{i}=\langle P,3,v_0, v_{\lfloor\frac{3n}{4}\rfloor-i}\rangle$ in between time $T_1'$ and $T_1$. Let $T_f$ be the time when last alive agent among $A_2$ and $A_3$ is killed  and $T_1'\le T_d\le T_f\le T_1$. Due to this construction, $A_1 \in C_1$, can not distinguish between these instances and so adversary can again create scenario 2 at $T_f$. 

So, to be specific, till now we have concluded that if there is at least one agent that never enters $L_{\lfloor\frac{n}{2} \rfloor}$ in between $T_1'$ and $T_1$ then adversary can always create scenario 2. 

So, now let us consider the case when there is a time $T_{*1}$ between $T_1'$ and $T_1$ when all three agents are in $L_{\lfloor\frac{n}{2}\rfloor}$. Let $T_*$  be the first time when all agents are in $L_{\lfloor \frac{n}{2}\rfloor}$. Let $A_{x_i}$ enters $L_{\lfloor \frac{n}{2}\rfloor}$ last time before $T_*$  at $T_{x_i}$. Here, for any $i \in \{1,2,3\}$, $x_i \in \{1,2,3\}$ and $x_i \ne x_j$ if $i\ne j$. Without loss of generality let, $T_{x_1} \le T_{x_2}\le T_{x_3}\le T_*$. Now if $T_{x_1}=T_{x_2}$ then for the instance $I=\langle P,3,v_0,v_{\lfloor\frac{n}{2}\rfloor}\rangle$, adversary can destroy $A_{x_1}$ and $A_{x_2}$ at $T_{x_1}=T_{x_2}$ and can destroy $A_{x_3}$ at time $T_{x_3}$ so creating the scenario 1 at time $T_f=T_{x_3}$. Now let us consider another case, $T_{x_1}<T_{x_2}\le T_*$. Let $T_{x_2}<T_*$ and let at $T_{x_2}$, $A_{x_3}$ is in $L_{\lfloor \frac{n}{2}\rfloor}$. Then this is a contradiction to the assumption that $T_*$ is the first round when all three of the agents are in side $L_{\lfloor \frac{n}{2}\rfloor}$. Thus either $T_{x_2}=T_*$, or, $T_{x_2}< T_*$ and $A_{x_3}$ is outside of $L_{\lfloor \frac{n}{2}\rfloor}$ at time $T_{x_2}$. When $T_{x_2}=T_*$ then $T_{x_2}=T_{x_3}=T_*$. In this case for  both the instances $I_{*_1}=\langle P,3,v_0,v_{\lfloor\frac{n}{2}\rfloor-1}\rangle$ and  $I_{*_2}=\langle P,3,v_0,v_{\lfloor\frac{n}{2}\rfloor}\rangle$, adversary can destroy both $A_{x_2}$ and $A_{x_3}$ at time $T_{x_2}$. Now if at $T_{x_2}$, $A_{x_1}$ is also at $v_{\lfloor\frac{n}{2}\rfloor}$ or moving to it then adversary destroys all of them leading to scenario 1 at time $T_f=T_{x_2}$. On the other hand if $A_{x_1}$ is not at $v_{\lfloor\frac{n}{2}\rfloor}$ or not moving at it at round $T_{x_2}$  then $A_{x_1}$ can not distinguish between both of the above mentioned instances. So, this case leads to scenario 3 at time $T_f=T_{x_2}$. Now let at $T_{x_2} (< T_{x_3}\le T_*)$, $A_{x_3}$ is not in $L_{\lfloor\frac{n}{2}\rfloor}$. Let us consider the time span starting at $T_{x_2}$ ending at $T^*_{x_3}$ where $T^*_{x_3}$ is the time when $A_{x_3}$ enters $L_{\lfloor\frac{n}{2}\rfloor}$ for the first time after $T_{x_2}$. Now let us consider the two instances $I_{*_1}$ and $I_{*_2}$, as discussed earlier. If adversary activates the BBH  for all rounds in the above mentioned time span, i.e., between $T_{x_2}$ and $T^*_{x_3}$, then it can destroy both the agents $A_{x_2}$ and $A_{x_3}$. Now if in this time span $A_{x_1}$ visits $v_{\lfloor\frac{n}{2}\rfloor}$ then it will be destroyed creating scenario 1 at $T_f=T^*_{x_3}$. Otherwise if it does not visit $v_{\lfloor\frac{n}{2}\rfloor}$ then only $A_{x_1}$ stays alive at $T_{x_3}$ without distinguishing between the instances. This leads to scenario 3 at $T_f=T^*_{x_3}$.

    Note that if there is a time when all three agents are in $L_{\lfloor\frac{n}{2}\rfloor}$ implies there is a time $T_{*_i}$ when all three agents are in  $L_{\lfloor\frac{n}{2}\rfloor-i}$, where $0\le i\le \lfloor\frac{n}{4}\rfloor$. So with similar reasoning as above for each we can create scenario 1 or scenario 3 where the BBH can be at any one of the vertex, $v_{\lfloor\frac{n}{2}\rfloor-i}$. 
 \end{proof}

\begin{lemma}\label{lemma: at least two groups are required}
  
    For any algorithm $\mathcal{A}$ solving \textsc{PerpExploration-BBH-Home} on a path graph $P=(V,E,\lambda)=L[v_0,v_{n-1}]$, at least one agent needs to be present at home until the destruction time $T_d$.
\end{lemma}

\begin{proof}

We prove this lemma with contradiction. Let, without loss of generality $v_0$ be the home. Let there be an algorithm $\mathcal{A}$ that solves \textsc{PerpExploration-BBH-Home} for the path graph $L[v_0,v_{n-1}]$ such that there exists a time $T (<T_d)$ when there are no agent at the home $v_0$. Let $i>0$ be the least integer such that $v_i$ contains any agent at time $T$. Then for the instance $I=\langle P,k, v_0,v_i\rangle$ if adversary activates the BBH for all rounds starting from round $T$, then all agents which are alive gets stranded on the component $C_2$ of graph $P-\{\frakb\}$ where $\frakb=v_i$ and $C_2$ does not contain the home $v_0$. So, it contradicts the assumption that $\mathcal{A}$ solves \textsc{PerpExploration-BBH-Home} on $P$. 
\end{proof}

\begin{theorem}\label{theorem: 4 agents cannot solve BBH-home}
    A set of 4 agents cannot solve \textsc{PerpExploration-BBH-Home} on a path graph of size $n'$, where $n'>36$.
\end{theorem}

\begin{proof}
Let us chose the path graph, $P=(V,E,\lambda)$ to be of length $n'=4n>36$. 
Without loss of generality let $v_0$ be the \emph{home}.
    We prove this theorem by proving the following cases one by one.
    
    \noindent\textbf{Case-I:} Let until $T_d$ the number of agents staying at \emph{home} is 3.
    Let $A_x$ be the only agent exploring the path graph $P$ until $T_d$. In this case, due to benign execution $A_x$ has to visit $v_{n-1}$ at some time, say $T_1$. Let $T_1'$ be the time $A_x$ visited $v_0$ last time before $T_1$. So adversary can create $n-1$ instances $I_j=\langle P,4,v_0,v_j\rangle$ where $1\le j\le n-1$ and can destroy $A_x$ in each of these instances in between $T_1'$ and $T_1$ at $T_d$. Thus the remaining 3 agents at home can not distinguish between these instances. Thus the problem can be thought of as solving \textsc{PerpExploration-BBH-Home} with 3 agents for a path graph of length $\lfloor\frac{n'}{4}\rfloor=n>9$ (where $n\ge 9$ as assumed), which is impossible due to Theorem~\ref{thm:3 agents not sufficient}.

     \noindent\textbf{Case-II:} Let until $T_d$ the number of agents staying at \emph{home} is 2. 
     
     In this case, by Lemma~\ref{lemma: 2 agents exploration cases LB}, adversary can create two possible scenarios, i.e., there exists a time $T_f\ge T_d$ such that, either  both the agents are destroyed on or before $T_f$ or, at $T_f$ there is exactly one agent, say $A_x$ among the three agents, remains alive in the component $C_1$ of graph $P-\{\frakb\}$ such that $C_1$ contains the home and $|S_x(T_f)|\ge \lfloor\frac{n'}{2}\rfloor=2n>18$. It can be easily proved that until $T_f$ no other agent leaves 
 \emph{home} if until $T_d$, 2 agents stays at \emph{home}. Otherwise, adversary can force more agents to leave \emph{home} before $T_d$ which contradicts our assumption. So for both these cases, this problem can be reduced to solving \textsc{PerpExploration-BBH-Home} on a path graph of size greater than 9 with at most 3 agents . So by Theorem~\ref{thm:3 agents not sufficient}, we can conclude that there is no algorithm that solves \textsc{PerpExploration-BBH-Home} on a path graph with 4 agents if until $T_d$, 2 agents stay at \emph{home}. 

      \noindent\textbf{Case-III:} Let until $T_d$ the number of agents staying at \emph{home} is 1. 
      
      In this case for the rest of the 3 agents, adversary can create one of the three cases as described in Lemma~\ref{lemma: 3 agents cases} at some time $T_f\ge T_d$. It can be easily shown that agents those are staying at \emph{home} until $T_d$ never leaves from \emph{home} until $T_f$ if number of agents staying at \emph{home} until $T_d$ is 1. It can be shown by providing an adverserial strategy to force the agent at \emph{home} to leave \emph{home} before $T_d$.
      
      Now, If all the three agents are destroyed then the problem reduces to solving the problem \textsc{PerpExploration-BBH-Home} on a path graph with at least $\lfloor\frac{n'}{4}\rfloor=n>9$ nodes with one agent. This is impossible by Theorem~\ref{thm:3 agents not sufficient}. For the case where by time $T_f$ at least one among the 3 agents exploring is destroyed and the remaining agents $A_{x_i}$  are alive at $C_1$ ($C_1$ being the component of $P-\{\frakb\}$ that contains the home) with $|S_{x_i}|\ge \lfloor\frac{n'}{4}\rfloor=n>9$. So the problem can be reduced to solving the \textsc{PerpExploration-BBH-Home} with 2 or 3 agents on a path graph of size $n$. Now it is impossible due to Theorem~\ref{thm:3 agents not sufficient}. Now let us consider the case where there is exactly one alive agent, say $A_x$, among the three agents exploring initially before $T_f$ and it is at $C_2$ having $|S_x(T)| \ge 2$ at $T_f$. So, after $T_f$ it is the sole duty of the agent at \emph{home}, say  $A_h$, to explore $C_1$ perpetually. But due to argument in the last paragraph of Lemma~\ref{lemma: 3 agents cases}, we have $|S_h(T)|\ge \lfloor\frac{n'}{4}\rfloor=n>9$. So, it has to visit $\frakb$ at some time $T_*>T_f$. Now, if the adversary activates, $\frakb$ for all rounds after $T_f$, then no agents from $C_2$ can cross $\frakb$ to come to $C_1$ and also $A_{h}$ gets destroyed at $T_*$. This implies the agents can not solve \textsc{PerpExploration-BBH-Home}.

      From all these above cases it is evident that there can not be an algorithm that solves \textsc{PerpExploration-BBH-Home} with 4 agents on any path graph.
\end{proof}

\begin{theorem} \label{thm:5 agents impossibility}
    A set of 5 agents cannot solve \textsc{PerpExploration-BBH-Home} on a path graph of size $n'$, where $n'>144$.
\end{theorem}
\begin{proof}
    Let without loss of generality,the path graph $P=(V,E,\lambda)$ we assume is $L[v_0,v_{16n-1}]$ such that $n>9$. Let $v_0$ be the \emph{home}. Let $k$ be the number of agents that stays at $v_0$ until $T_d$. Based on the values of $k \in \{1,2,3,4\}$, we have the following cases. We prove this theorem by proving these cases. 

    \textbf{Case-I ($k=4$):} Let until $T_d$ the number of agents staying at home is 4. In this case only one agent, say $A_x$ explores the path graph until $T_d$. Now by similar argument as in Case-I of Theorem~\ref{theorem: 4 agents cannot solve BBH-home}, at $T_d$, the problem can be reduced to solving \pbmPerpExplHome~with four agents on a path graph of length at least $4n>36$ which is impossible due to Theorem \ref{theorem: 4 agents cannot solve BBH-home}. 

\textbf{Case-II ($k=3$):} Let until $T_d$, the number of agents staying at $v_0$ is 3. Now, let $A_{x_1}$ and $A_{x_2}$ be the only two agents exploring the path graph $L[v_0,v_{16n-1}]$ until $T_d$. Now as stated in Lemma~\ref{lemma: 2 agents exploration cases LB}, adversary can create two scenarios which are as follows. There exists a round $T_f\ge T_d$ such that at $T_f$ either both the agents are destroyed or exactly one, say $A_{x_i}$ where $i \in \{1,2\}$ remains alive with  $|S_{x_i}(T_f)|\ge 8n$. Also till $T_f$ no agents that are staying at \emph{home} till $T_d$ leaves \emph{home}. As otherwise adversary can make them leave \emph{home} even before $T_d$, which is a contradiction to our assumption that only two agents explore the path until $T_d$.  So, for both the cases the problem now reduces to solving \pbmPerpExplHome~using at most 4 agents on a path of length at least $4n>36$ which is impossible due to Theorem~\ref{theorem: 4 agents cannot solve BBH-home}.

\textbf{Case-III ($k=2$):} Let until $T_d$, the number of agents staying at $v_0$ is 2. So, there are three agents, say, $A_{x_1},A_{x_2}$ and $A_{x_3}$ exploring the path graph until $T_d$. By Lemma~\ref{lemma: 3 agents cases}, the adversary can create three scenarios at some time $T_f\ge T_d$. Also. for each of these scenarios the agents staying at \emph{home} till $T_d$ never leaves \emph{home} even until $T_f$. As otherwise by delaying $T_d$, adversary can enforce agents staying at \emph{home} till $T_d$ to leave \emph{home} before $T_d$ which is a contradiction to our assumption that until $T_d$ exactly two agents stay at \emph{home}.

In the first scenario all three agents are destroyed in at least $4n>36$ many distinct instances denoted by $I_j=\langle L[v_0,v_{16n-1}],5,v_0, v_j\rangle$, where $1\le j\le 4n$ (by argument from the last paragraph of Lemma~\ref{lemma: consequence of 3 agents Path Graph}). So at $T_f$, the problem reduces to solving \pbmPerpExplHome~on a path of length at least $4n>36$  with two agents which is impossible due to Theorem~\ref{thm:3 agents not sufficient}. Now for the second scenario, where adversary destroys at least one agent among $A_{x_i}$ where $i \in \{1,2,3\}$, without loss of generality, let it be the agent $A_{x_3}$. Then, $A_{x_1}$ and $A_{x_2}$ stays at $C_1$, where $C_1$ is the component of $L[v_0,v_{16n-1}]\setminus \{\frakb \}$ containing $v_0$. Also due to Lemma~\ref{lemma: consequence of 3 agents Path Graph}, $|S_{x_i}(T_f)| \ge \frac{16n}{4}=4n>36$, for each $i\in \{1,2\}$. So at $T_f$, the problem now reduces to solving \pbmPerpExplHome~with 3 or 4 agents on a path graph of size at least $4n>36$ which is impossible due to Theorem~\ref{theorem: 4 agents cannot solve BBH-home}. Now we consider the scenario three at $T_f$, where the adversary can create a situation where, exactly one of $A_{x_1}, A_{x_2}$ and $A_{x_3}$ remains alive at the component $C_2$, where $C_2$ is the component of $L[v_0,v_{16n-1}]\setminus\{\frakb\}$ such that $v_0 \notin C_2$. Let without loss of generality, $A_{x_1}$ be the alive agent at $C_2$ then as per the condition in Lemma~\ref{lemma: consequence of 3 agents Path Graph}, we have $|S_{x_1}(T_f)| \ge 2$ where $T_f\ge T_d$. In fact, $S_{x_1}(T_f)$ is a contiguous segment of the path graph. Also from the argument in last paragraph of Lemma~\ref{lemma: consequence of 3 agents Path Graph}, adversary can create $4n$ instances denoted by $I_j=\langle L[v_0,v_{16n-1}], 5, v_0, v_j\rangle$, where $1\le j\le 4n$, such that for each of these $4n>36$ instances, adversary can induce scenario 3 at some time $T_f$. So, after scenario 3 is induced at $T_f$, the two agents, say $A_{x_4}$ and $A_{x_5}$ at \emph{home}, has the sole duty to explore $C_1$. Note that, at $T_f$, agents $A_{x_4}$ and $A_{x_5}$, has $|S_{x_i}(T_f)| \ge 4n>36$ where $i\in \{4,5\}$. Now by Theorem~\ref{thm:3 agents not sufficient}, it is impossible for two agents to solve \pbmPerpExplHome\ on a path graph of length $4n>36$. So to able to perpetually explore, at least one of the agents, i.e., $A_{x_4}$ or $A_{x_5}$ must meet with $A_{x_1}$ after $T_f$. Let for $A_{x_1}$, two possible positions of BBH are, $v_{p}$ and $v_{p+1}$. Without loss of generality, let $A_{x_1}$ meets $A_{x_4}$. Note that before meeting, if $A_{x_1}$ reaches $v_{p+1}$, then for the instance $\langle L[v_0,v_{16n-1}],5,v_0,v_{p+1}\rangle$ adversary can destroy $A_{x_1}$ before it can meet any of $A_{x_4}$ and $A_{x_5}$, thus for the remaining alive agents $A_{x_4}$ and $A_{x_5}$, the problem still remains to solve \pbmPerpExplHome~on a path graph of size at least $4n>36$ which is impossible due to Theorem~\ref{thm:3 agents not sufficient}. Also we claim that $A_{x_4}$ cannot meet $A_{x_1}$ for the first time at $v_x$ where $x >p+1$. This is because in this case $A_{x_4}$ have to cross $v_p$, and adversary can chose from any of the two distinct instances $\langle L[v_0,v_{16n-1}],5,v_0,v_{p}\rangle$ and $\langle L[v_0,v_{16n-1}],5,v_0,v_{p+1}\rangle$ to destroy $A_{x_4}$ in both instances and the other agent in $C_1$ i.e., $A_{x_5}$ can not distinguish between these two instances, if it is located at $v_x$ ($x<p$) when $A_{x_4}$ moves to $v_{p+1}$ from $v_p$. Thus failing to solve \pbmPerpExplHome~as well. Note that when $A_{x_4}$ moves from $v_{p}$ to $v_{p+1}$, if at this round $A_{x_5}$ is on $v_x$ where $x \ge p$ then for the instance  $\langle L[v_0,v_{16n-1}],5,v_0,v_{p}\rangle$ adversary activates the BBH for all the upcoming rounds and all alive agents gets stranded in $C_2$, again failing to solve \pbmPerpExplHome. So if $A_{x_4}$ and $A_{x_1}$ meets, it must be at the node $v_{p+1}$. Let they meet first time on $v_{p+1}$ at some round $T^*$. In this case at $T^*-1$ $A_{x_4}$ must be on $v_p$ and $A_{x_1}$ must be on $v_{p+2}$. Now if at $T^*-1$, $A_{x_5}$ is with $A_{x_4}$ at $v_p$ then, adversary can chose the instance $\langle L[v_0,v_{16n-1}],5,v_0,v_{p}\rangle$ and makes the BBH act as BH for the rest of the execution from $T^*-1$ to destroy $A_{x_4}$ and $A_{x_5}$ and make the only living agent $A_{x_1}$, get stranded at $C_2$, thus making \pbmPerpExplHome~impossible. On the other hand, if at $T^*-1$, $A_{x_5}$ is not with $A_{x_4}$ at $v_p$ in $C_1$ then adversary can chose any of the following two instances, $I_1= \langle L[v_0,v_{16n-1}],5,v_0,v_{p}\rangle$ and $I_2= \langle L[v_0,v_{16n-1}],5,v_0,v_{p+1}\rangle$, and makes the BBH act as BH for the rest of the execution. In this case $A_{x_5}$ can not distinguish between these two instances and can not meet $A_{x_1}$ for any further help. So in this case also, \pbmPerpExplHome\ remains impossible to solve. Next we tackle the case where $k=1$.

\textbf{Case-IV} ($k=1$): In this case until $T_d$, 4 agents namely, $A_{x_1}, A_{x_2}, A_{x_3}$ and $A_{x_4}$, explores the path graph $L[v_0,v_{16n-1}]$. Now let only one agent, without loss of generality, say $A_{x_1}$, enters $L_{12n}$, then by similar argument as in first paragraph of Lemma~\ref{lemma: 2 agents exploration cases LB}, adversary can create $4n>36$ instances denoted by $I_j=\langle L[v_0,v_{16n-1}], 5, v_0,v_j \rangle$ ($12n\le j\le 16n-1$) for each of which it can destroy $A_{x_1}$. Thus the problem now reduces to solving \pbmPerpExplHome~by 4 agents on a path graph of length at least $4n>36$, which is impossible due to Theorem~\ref{theorem: 4 agents cannot solve BBH-home}. 

Now let only two agents, among $A_{x_1}, A_{x_2}, A_{x_3}$ and $A_{x_4}$ enters $L_{12n}$. Let these agents who enters $L_{12n}$ be, $A_{x_1}$ and $A_{x_2}$. Without loss of generality, let there exists a round $T_1$ when $A_{x_1}$ is at $v_{16n-1}$. Let $T_1'$ be the time when $A_{x_1}$ was at $v_{12n}$ for the last time before $T_1$. We can say that there exists a time $T_1^*$ between $T_1'$ and $T_1$  when both $A_{x_1}$ and $A_{x_2}$ are in $L_{12n} \subset L_{8n}$ (due to similar argument as in second and third paragraph of Lemma~\ref{lemma: consequence of 3 agents Path Graph}). If no other agents ever enters $L_{8n}$ when $A_{x_1}$ and $A_{x_2}$ are in $L_{8n+j}$ where $0\le j\le 4n$, then we can also conclude there are rounds $T_j^*$ such that at $T_j^*$ both these agents $A_{x_1}$ and $A_{x_2}$ are in $L_{8n+j}$ where $0 \le j\le 4n$. Now due to similar argument used in Lemma~\ref{lemma: 2 agents exploration cases LB}, adversary can create $4n$ distinct instances denoted as $I_{\alpha}= \langle L[v_0,v_{16n-1}], 5,v_0, v_{8+\alpha} \rangle$ ($0 \le \alpha \le 4n$) such that it can destroy both the agents for each of these $4n>36$ instances. So now the problem reduces to solving \pbmPerpExplHome~with 3 agents on a path graph of length at least $4n>36$ which is impossible due to Theorem~\ref{thm:3 agents not sufficient}. 

Now as argued in Lemma~\ref{lemma: 3 agents cases}, there is a time when at least three agents are in $L_{8n}$. So, let there exists a round $T_2$ such that three agents, without loss of generality, say $A_{x_1},A_{x_2}$ and $A_{x_3}$ are in $L_{8n} \subset L_{4n}$. Let $A_{x_1},A_{x_2}$ and $A_{x_3}$ are in  $L_{4n+j}$ ( for all $0\le j\le 4n$), but $A_{x_4}$ never enters $L_{4n}$. Now as argued in the last paragraph of Lemma~\ref{lemma: consequence of 3 agents Path Graph}, adversary can create $4n$ many instances  denoted as $I_{\beta}=\langle L[v_0,v_{16n-1}], 5, v_0, v_{8n-\beta}\rangle$  ($0\le \beta \le 4n$) for which it can create scenario 1 or 3 of Lemma~\ref{lemma: consequence of 3 agents Path Graph}, considering each vertex of segment $[v_0,v_{4n-1}]$ as \textit{segmented home}, which implies the agent which did not leave actual home and $A_{x_4}$ which did not leave the segment $[v_0,v_{4n-1}]$ are at segmented home and rest of three agents left segmented home, to explore the rest of the path.  Now if scenario 1 is achieved,  the problem reduces to solving \pbmPerpExplHome\ on a path of length at least $4n>36$ with two agents which is impossible due to Theorem~\ref{thm:3 agents not sufficient}. If scenario 3 is achieved then as argued in the case-III of this proof, similarly we can say that solving \pbmPerpExplHome\ is impossible for this case too. 

So, this concludes that there is a round $T_3$ (the existence of such a round follows by extending the arguments in Lemma~\ref{lemma: consequence of 3 agents Path Graph}) when  each of the 4 agents, $A_{x_1},A_{x_2},A_{x_3}$ and $A_{x_4}$ needs to be in $L_{4n} \subset L_1$, as in earlier case we assumed that $A_{x_4}$ did not leave the segment $[v_0,v_{4n-1}]$. Let $T_{m}$ be the first round, at which all these four agents are at $L_{4n}$. In this situation, the adversary can create $I_i=\langle L[v_0,v_{16n-1}],5,v_0,v_i\rangle$, where $1\le i \le 4n-1$ instances, and in each instance it activates the BBH from round $T_m$ onwards. Now, this shows all agents except $A_{x_5}$ (which is at \emph{home} till $T_m$) gets stuck at $C_2$, and it is impossible for $A_{x_5}$ to detect the BBH in the segment $[v_1,v_{4n-1}]$.

So we conclude that for each of the cases (i.e., $k\in \{1, 2, 3, 4\}$) it is impossible to solve \pbmPerpExplHome. So, it is impossible to solve \pbmPerpExplHome\ with 5 agents.
    
\end{proof}

\section{A perpetual exploration algorithm for paths with a BBH} \label{Appendix: Path Alg}

We provide an algorithm that solves \pbmPerpExplHome\ with 6 agents, even when the size of the path is unknown to the agents. We also show how to adapt this algorithm to solve \pbmPerpExpl\ with 4 agents. Note that, by Theorems~\ref{thm:3 agents not sufficient} and~\ref{thm:5 agents impossibility}, these are the optimal numbers of agents and they cannot be reduced, even assuming knowledge of the size of the path.

We call this algorithm \textsc{Path\_PerpExplore-BBH-Home}. Let $\left\langle P,6,h,\frakb\right\rangle$ be an instance of~\pbmPerpExplHome, where $P=(V,E,\lambda)$ is a port-labeled path. As per Definition~\ref{def:correctalg}, all agents are initially co-located at~$h$ (the home node). To simplify the presentation, we assume that $h$ is an extremity of the path, and we explain how to modify the algorithm to handle other cases in Remark~\ref{remark: h is not an extreme end}. Our algorithm works with 6 agents. Initially, among them the four least ID agents will start exploring $P$, while the other two agents will wait at $h$, for the return of the other agents. We first describe the movement of the four least ID agents say, $a_0$, $a_1$, $a_2$ and $a_3$, on $G$. Based on their movement, they identify their role as follows: $a_0$ as $F$, $a_1$ as $I_2$, $a_2$ as $I_1$ and $a_3$ as $L$. The exploration is performed by these four agents in two steps, in the first step, they form a particular pattern on $P$. Then in the second step, they move collaboratively in such a way that the pattern is translated from the previous node to the next node in five rounds. Since the agents do not have the knowledge of $n$, where $|V|=n$, they do the exploration of $P$, in $\log n$ phases, and then this repeats. In the $i$-th phase, the four agents start exploring $P$ by continuously translating the pattern to the next node, starting from $h$, and moves up to a distance of $2^i$ from $h$. Next, it starts exploring backwards in a similar manner, until they reach $h$. It may be observed that, any phase after $j-$th phase (where $2^j=n$), the agents behave in a similar manner, as they behaved in $j-$th phase, i.e., they move up to $2^j$ distance from $h$, and then start exploring backwards in a similar manner, until each agent reach $h$. Note that since all the agents have $\mathcal{O}(\log n)$ bits memory, each of them knows which phase is currently going on. Let $T_i$ be the maximum time, required for these 4 agents (i.e., $L,~I_1,~I_2$ and $F$) to return back to $h$ in $i-$th phase. Let us denote the waiting agents at $h$, i.e., $a_4,~a_5$ as $F_1,~F_2$. Starting from the $i$-th phase, they wait for $T_i$ rounds for the other agents to return. Now, if the set of agents $L,~I_1,~I_2$ and $F$ fail to return back to \emph{home} within $T_i$ rounds in phase, $i$, the agents $F_1$ and $F_2$ starts moving \textit{cautiously}. In the cautious move, first $F_1$ visits the next node and in the next round, returns back to the previous node to meet with $F_2$, that was waiting there for $F_1$. If $F_1$ fails to return then $F_2$ knows the exact location of $\frakb$, which is the next node $F_1$ visited. In this case, $F_2$ remains the component of $h$, hence can explore it perpetually. Otherwise, if $F_1$ returns back to $F_2$, then in the next round both of them move together to the next node.

Now, if the first set of agents (i.e., $L,~I_1,~I_2$ and $F$) fails to return to $h$ within $T_i$ rounds, then that means the adversary has activated $\frakb$. In this case we claim that, at least one agent among $L,~I_1,~I_2$ and $F$ stays alive at a node in $C_2$ knowing the location of $\frakb$, where $C_2$ is the component of $P- \{\frakb\}$, such that it does not contain $h$. Let $\alpha$ be such an alive agent, where $\alpha\in\{L,I_1,I_2,F\}$. Then, $\alpha$ places itself on the adjacent node of $\frakb$ in $C_2$. It may be noted that, $\alpha$ knows which phase is currently going on and so it knows the exact round at which $F_1$ and $F_2$ starts cautious move. Also, it knows the exact round at which $F_1$ first visits $\frakb$, say at round $r$. $\alpha$ waits till round $r-1$, and at round $r$ it moves to $\frakb$. Now at round $r$, if adversary activates $\frakb$, it destroys both $F_1$ and $\alpha$ then, $F_1$ fails to return back to $F_2$ in the next round. This way, $F_2$ knows the exact location of $\frakb$, while it remains in $C_1$. So it can explore $C_1$ by itself perpetually. The right figure in Fig. \ref{fig:perpetuallyexplore-timediagram}, represents the case where $L$ detects $\frakb$ at round $r_0+2$, and waits till round $r_1+3$. In the meantime, at round $r_1$ (where $r_1=r'_0+T_i$, $r'_0<r_0$ and $r'_0$ is the first round of phase $i$), $F_1$ and $F_2$ starts moving cautiously. Notably, along this movement, at round $r_1+4$, $F_1$ visits $\frakb$, and at the same time $L$ as well visits $\frakb$ from $v_{j+3}$. The adversary activates $\frakb$, and both gets destroyed. So, at round $r_1+5$, $F_2$ finds failure of $F_1$'s return and understands the next node to be $\frakb$, while it is present in $C_1$. Accordingly, it perpetually explores $C_1$.

On the other hand, if at round $r$, adversary doesn't activate $\frakb$, then $F_1$ meets with $\alpha$ and knows that they are located on the inactivated $\frakb$. In this case they move back to $C_1$ and starts exploring the component $C_1$, avoiding $\frakb$. The left figure in Fig. \ref{fig:perpetuallyexplore-timediagram} explores this case, where $L$ detects the position of $\frakb$ at round $r_0+2$, and stays at $v_{j+3}$ until $r_1+3$, then at round $r_1+4$, when $F_1$ is also scheduled to visit $\frakb$, $L$ also decides to visit $\frakb$. But, in this situation, the adversary does not activate $\frakb$ at round $r_1+4$, so both $F_1$ and $L$ meets, gets the knowledge from $L$ that they are on $\frakb$. In the next round, they move to $v_{j+1}$ which is a node in $C_1$, where they meet $F_2$ and shares this information. After which, they perpetually explore $C_1$.

We now describe how the set of four agents, i.e., $L$, $I_1$, $I_2$ and $F$ create and translate the pattern to the next node. 

\noindent\textbf{Creating pattern:} $L, I_1,I_2$ and $F$ takes part in this step from the very first round  of any phase, starting from $h$. In the first round, $L,~I_1$ and $I_2$ moves to the next node. Then in the next round only $I_2$ returns back to \emph{home} to meet with $F$. Note that, in this configuration the agents $L$, $I_1$, $I_2$ and $F$ are at two adjacent nodes while $F$ and $I_2$ are together and $L$ and $I_1$ are together on the same node. We call this particular configuration the \textit{pattern}, and it is pictorially explained in Fig \ref{fig: make pattern formation}.  

\begin{figure}[h!]
    \centering
    \begin{minipage}{0.45\textwidth}
       \centering
        \includegraphics[width=\textwidth]{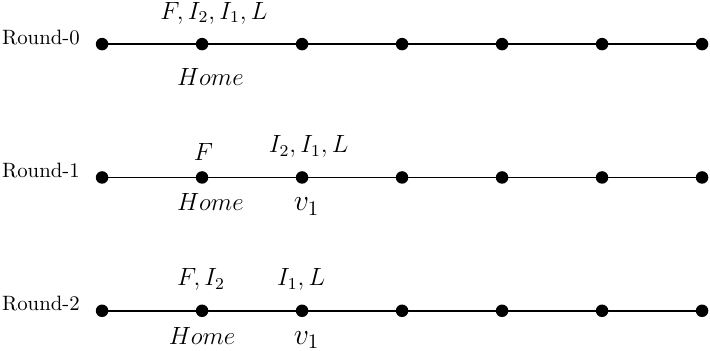}
        \caption{Depicts step by step movement, while creating the pattern from $h$}
        \label{fig: make pattern formation}
    \end{minipage}%
    \hfill
    \begin{minipage}{0.45\textwidth}
         \centering
        \includegraphics[width=\textwidth]{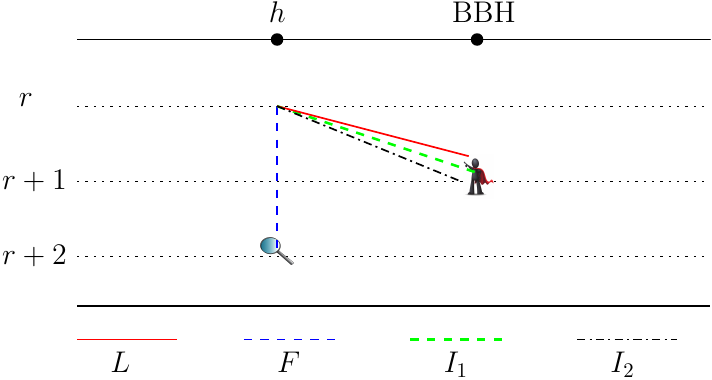}
        \caption{The time diagram, which depicts the case when $F$ detects the position of $\frakb$ while it is at $h$, while creating pattern.}
        \label{fig:make pattern detection}
    \end{minipage}
    \label{fig:create-and-translate}
\end{figure}

Note that, in the second round of creating pattern if $I_2$ does not return to $F$, $F$ knows that it must have been destroyed in the next node by the BBH. So $F$ knows the exact location of the BBH and explores the component $C_1$ (where $P-\{\frakb\}=C_1\cup C_2$, where $C_1,C_2\subset P$ and $h\in C_1$) perpetually, refer to Fig. \ref{fig:make pattern detection}. 

\noindent\textbf{Translating pattern:} After the pattern is formed in first two rounds of a phase, the agents translate the pattern to the next node until the agent $L$ reaches either at one end of the path graph $P$, or, reaches a node at a distance $2^i$ from $h$ in the $i$-th phase. Let, $v_0,~v_1,~v_2$ be three consecutive nodes on $P$, where, suppose $L$ and $I_1$ is on $v_1$ and $F$ and $I_2$ is on $v_0$. This translation of the pattern makes sure that after 5 consecutive rounds $L$ and $I_1$ is on $v_2$ and $F$ and $I_2$ is on $v_1$, thus translating the pattern by one node. We call these 5 consecutive rounds where the pattern translates starting from a set of 2 adjacent nodes, say $v_0,~v_1$, to the next two adjacent nodes, say $v_1,~v_2$, a \textit{sub-phase} in the current phase. The description of the 5 consecutive rounds i.e., a sub-phase without the intervention of the BBH at $\frakb$ (such that $\frakb\in\{v_0,v_1,v_2\}$) are as follows.

\noindent\underline{\textbf{Round 1:}} $L$ moves to $v_2$ from $v_1$.

\noindent\underline{\textbf{Round 2:}} $I_2$ moves to $v_1$ from $v_0$ and $L$ moves to $v_1$ back from $v_2$.

\noindent\underline{\textbf{Round 3:}} $I_2$ moves back to $v_0$ from $v_1$ to meet with $F$. Also, $L$ moves back to $v_2$ from $v_1$.

\noindent\underline{\textbf{Round 4:}} $F$ and $I_2$ moves to $v_1$ from $v_0$ together.

\noindent\underline{\textbf{Round 5:}} $I_1$ moves to $v_2$ from $v_1$ to meet with $L$.

So after completion of round 5 the pattern is translated from nodes $v_0$, $v_1$ to $v_1$, $v_2$. The pictorial description of the create and translate pattern are explained in Fig. \ref{fig:translate-pattern}.

\begin{figure}[h!]
    \centering
    \begin{minipage}{0.48\textwidth}
        \centering
        \includegraphics[width=0.8\textwidth]{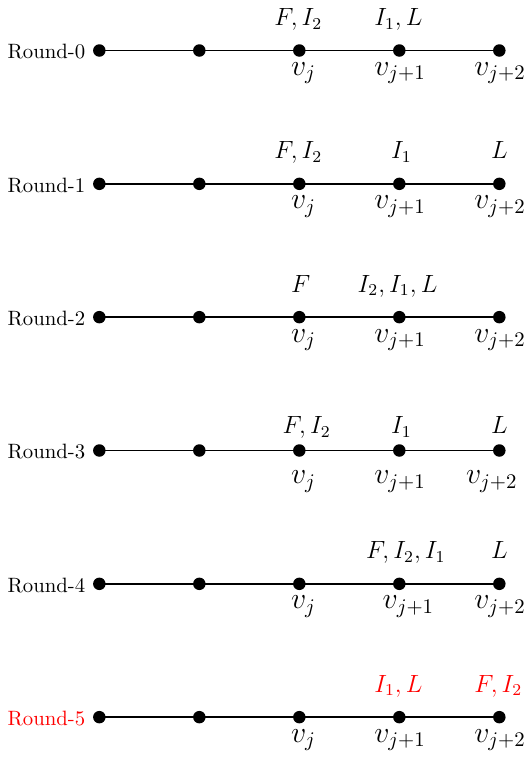}
        \caption{Depicts the step-wise interchange of roles, when the agents reach the end of the path graph, while performing translate pattern}
        \label{fig:translate-pattern-flip}
    \end{minipage}%
    \hfill
    \begin{minipage}{0.48\textwidth}
        \centering
        \includegraphics[width=1.0\textwidth]{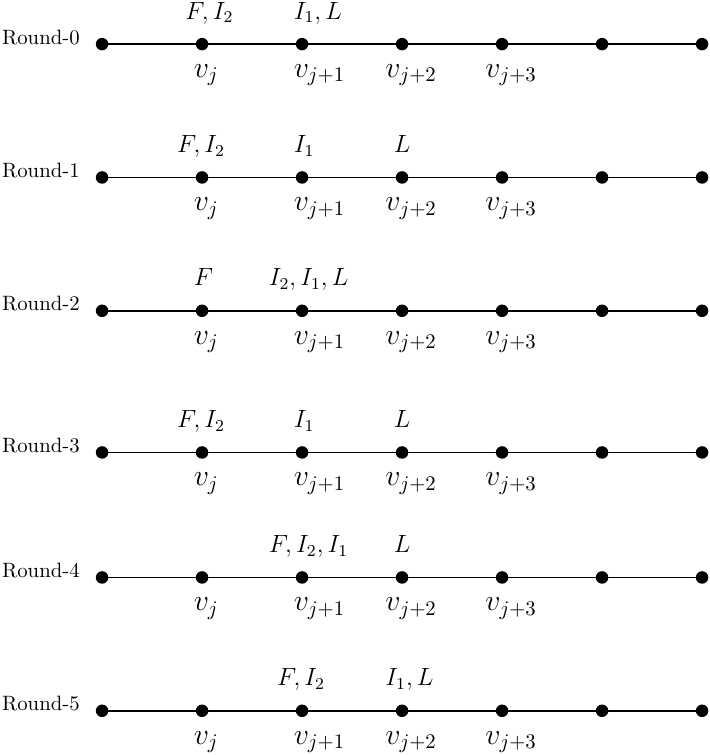}
        \caption{Depicts translating pattern steps}
        \label{fig:translate-pattern}
    \end{minipage}
\end{figure}

Now, suppose $v_2$ is the node upto which the pattern was supposed to translate at the current phase (or, it can be the end of the path graph also). So, when $L$ visits $v_2$ for the first time, in round 1 of  some sub-phase it knows that it has reached the end of the path graph for the current phase. Then in the same sub-phase at round 2 it conveys this information to $I_2$ and $I_1$. In round 3 of the same sub-phase, $F$ gets that information from $I_2$. So at the end of the current sub-phase all agents has the information that they have explored either one end of the path graph or the node upto which they were supposed to explore in the current phase. In this case, they interchange the roles as follows: the agent which was previously had role $L$ changes role to $F$, the agent having role $F$ changes it to $L$, $I_1$ changes role to $I_2$ and $I_2$ changes role to $I_1$ (refer to Fig. \ref{fig:translate-pattern-flip}). Then from the next sub-phase onwards they start translating the pattern towards $h$. It may be noted that, once $L$ (previously $F$) reaches $h$ in round 1 of a sub-phase, it conveys this information to the remaining agents in similar manner as described above. So, at round 5 of this current sub-phase, $F$ (previously $L$) and $I_2$ (previously $I_1$) also reaches $h$, and meets with $L$ (previously $F$) and $I_1$ (previously $I_2$). We name the exact procedure as \textsc{Translate\_Pattern}.

\noindent\textbf{Intervention by the BBH:} Next, we describe the situations which can occur, if the BBH destroys at least one among these 4 agents within sub-phase $i$ in phase $j$, say. At the starting of sub-phase $i$, suppose the agents $L$, $I_1$ is on $v_1$ and $F$, $I_2$ is on $v_0$, at the end sub-phase $i$, the goal of $L$, $I_1$ is to reach $v_2$ and $F$, $I_2$ is to reach $v_1$. Without loss of generality, here we are assuming that the pattern is translating away from $h$ in the current sub-phase. 

\noindent\underline{\textbf{Case-I:}} In this case, we look in to the case when $v_2$ is not the end point for the current sub-phase.

\noindent\underline{\textbf{If $v_2$ is $\frakb$:}} We describe the cases, which can arise, depending on at which round within this sub-phase, $\frakb$ gets activated and destroys at least one among these 4 agents.

Let at round 1, $\frakb$ is activated for the first time in the sub-phase $i$. Then, in round 2 of the sub-phase $i$, $L$ does not return to $v_1$ and meets with $I_1$ and $I_2$. Thus $I_1$ and $I_2$ knows that $v_2$ is $\frakb$ and they can explore $C_1$ by themselves as in round 2 all of them are at $C_1$ (refer to (i) of Fig. \ref{fig:v2 is BBH}, where $v_{j+2}$ is $\frakb$ and it symbolises $v_2$ in the description, also here in the figure, the sub-phase starts from round $r$). 

\begin{figure}[h!]
    \centering
    \includegraphics[width=0.8\linewidth]{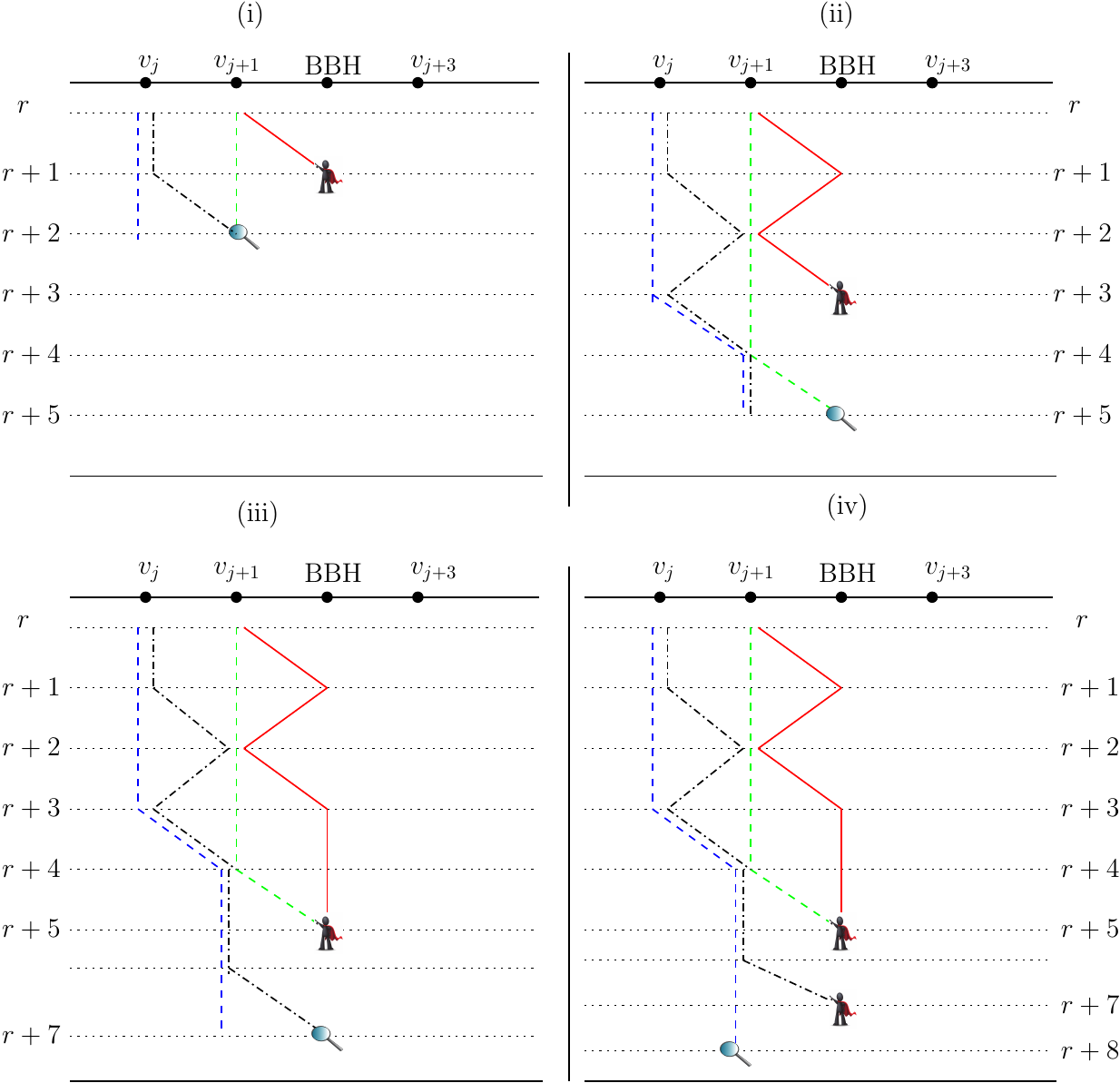}
    \caption{Time Diagram depicting each case, that can arise if $\frakb=v_{j+2}$ intervenes, while 4 agents are translating the pattern.}
    \label{fig:v2 is BBH}
\end{figure}

If, it is activated at round 2 for the first time in in the sub-phase $i$, then no agents are destroyed, as at round 2 none of these agents are present at $v_2$, and they continue the next rounds of the sub-phase as usual. 

If at round 3, $\frakb$ is activated for the first time in the sub-phase $i$, then at round 5 of sub-phase $i$ when $I_1$ reaches $v_2$ and it stays alive then it knows that it is on $\frakb$ as $L$ is not there (refer to (ii) of Fig. \ref{fig:v2 is BBH}, where $L$ gets destroyed at round $r+3\approx 3$ of sub-phase $i$, and $I_1$ detects it at round $r+5\approx 5$ of sub-phase $i$, when it visits $\frakb=v_{j+2}$). So, in this case it moves back to $v_1 \in C_1$ and starts exploring $C_1$ by itself. On the other hand if at round 5 of sub-phase $i$, $I_1$ is also destroyed at $v_2$ then at round 2 of sub-phase $(i+1)$, when $I_2$ reaches $v_2$, if it stays alive, it knows that it is on $\frakb$ as both $I_1$ and $L$ are missing (refer to (iii) of Fig. \ref{fig:v2 is BBH}, where at round $r+5\approx 5$ of sub-phase $i$, $I_1$ also gets destroyed, and $I_2$ detects it when it visits $\frakb$ at round $r+7\approx 2$ of sub-phase $(i+1)$). So, similarly it moves back to $v_1$ and starts exploring $C_1$ by itself. Now, if $I_1$ is also destroyed at round 2 of sub-phase $(i+1)$ then at round 3 of sub-phase $(i+1)$, $F$ at $v_1 (\in C_1)$ finds that $I_2$ is missing. From this, $F$ interprets that the BBH is at $v_2$, and thus it starts exploring $C_1$ by itself. 


Let $\frakb$ is activated at round 4, for the first time in the sub-phase $i$ then it destroys $L$, so at round 5 of the same sub-phase, when $I_1$ moves to $v_2$ and it stays alive, it finds that $L$ is missing. Hence, it knows that it is on $\frakb$ (refer to (ii) of Fig. \ref{fig:v2 is BBH}, except that instead of round $r+3\approx 3$ of sub-phase $i$, $\frakb$ is activated at round $r+4\approx 4$ of sub-phase $i$) and moves back to $v_1 \in C_1$ and starts exploring $C_1$ by itself. 

If $\frakb$ is activated at $v_2$ in the sub-phase $i$ for the first time in round 5, then it destroys both $L$ and $I_1$ at $v_2$. Next, when $I_2$ moves to $v_2$ at round 2 of sub-phase $(i+1)$, suppose $\frakb$ is again activated (if not already activated), then $I_2$ is also destroyed at round 2 of sub-phase $(i+1)$, so it fails to return back to $v_1$ at the same round. So, at round 3 of sub-phase $(i+1)$, $F$ knows that $v_2$ is the exact location of BBH by finding out that $I_2$ is missing at $v_1$ (refer to (iv) of Fig. \ref{fig:v2 is BBH}, where after $L$, $I_1$ gets destroyed at round $r+5\approx 5$ of sub-phase $i$, $I_2$ also gets destroyed at round $r+7\approx 2$ of sub-phase $(i+1)$, then $F$ understands this at round $r+8\approx 3$ of sub-phase $(i+1)$). So, $F$ then starts exploring $C_1$ by itself.

 \noindent\underline{\textbf{If $v_1$ is $\frakb$:}} Again, we describe the cases, which can arise depending on at which round within $i$-th sub-phase, $\frakb$ is activated, and destroys at least one among these 4 agents.

\begin{figure}
    \centering
    \includegraphics[width=0.9\linewidth]{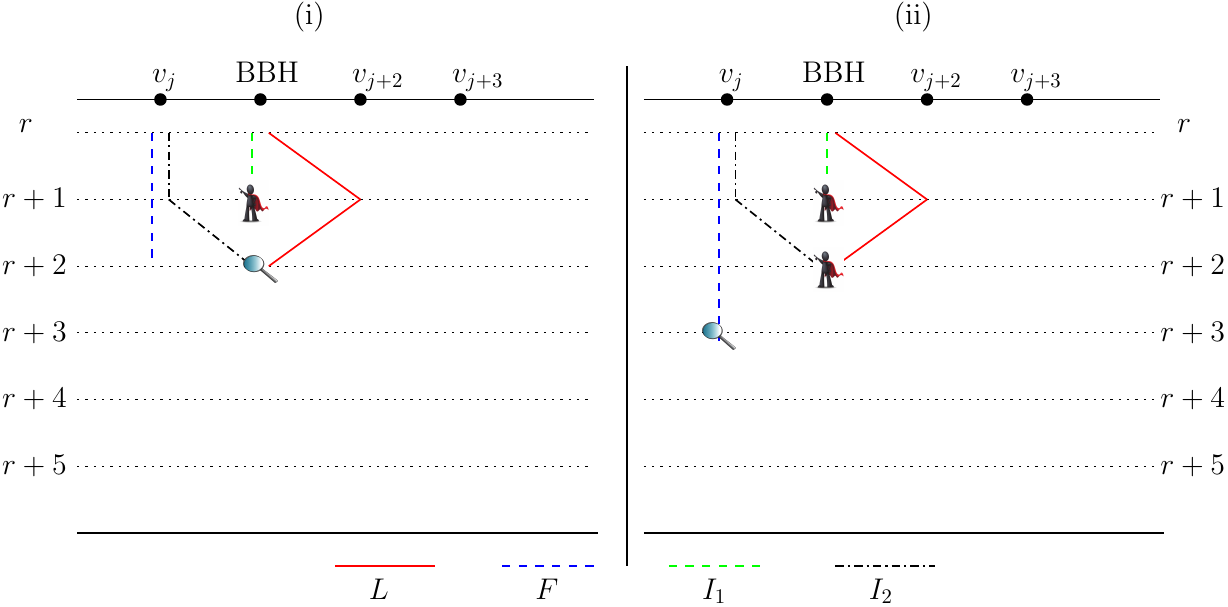}
    \caption{Depicts the time diagram, where $v_1=\frakb$, in (i) it is activated at $r+1$ and in (ii) it is activated at rounds $r+1$ and $r+2$}
    \label{fig:v1 is BBH1}
\end{figure}

Let at round $1$, $\frakb$ is activated for the first time in the sub-phase $i$. Then, in round 2 of the sub-phase $i$, when $I_2$ and $L$ visits $v_1$, and if they stay alive, they find that $I_1$ is missing. This information helps them understand that they are on $\frakb$ (refer to (i) of Fig. \ref{fig:v1 is BBH1} where at round $r+2\approx 2$ of sub-phase $i$, $L$ and $I_1$ detect the position of $\frakb$), so they move back to $v_0\in C_1$ and starts exploring $C_1$, perpetually. On the other hand, if $\frakb$ is activated at round 2, then both $L$ and $I_2$ also gets destroyed at $v_1$ at round 2. In this situation, at round 3, whenever $F$ (present at $v_0\in C_1$) finds that $I_2$ has not arrived from $v_1$, it understands that $v_1$ is $\frakb$ (refer to (ii) of Fig. \ref{fig:v1 is BBH1}, where $L$ detects the position of $\frakb$ at round $r+3\approx 3$ in sub-phase $i$, after $I_2$, $L$ gets destroyed at round $r+2$ and $I_1$ gets destroyed at $r+1$), and explores $C_1$ perpetually.

If at round 2, $\frakb$ is activated for the first time in the sub-phase $i$. Then, in round 2 itself, $I_1$, $I_2$ and $L$ gets destroyed at $v_1$. In this case, $F$ present at $v_0\in C_1$, finds that $I_2$ fails to return from $v_1$ at round $3$, hence it understands, $v_1$ is $\frakb$, and performs perpetual exploration of $C_1$. Refer to (ii) of Fig. \ref{fig:v1 is BBH1}, a similar case is discussed, where all three agents $L$, $I_1$ and $I_2$ gets destroyed due to $v_{j+1}$ being $\frakb$ and that is detected by $F$ at round $r+3\approx 3$ in sub-phase $i$.

\begin{figure}
    \centering
    \includegraphics[width=0.9\linewidth]{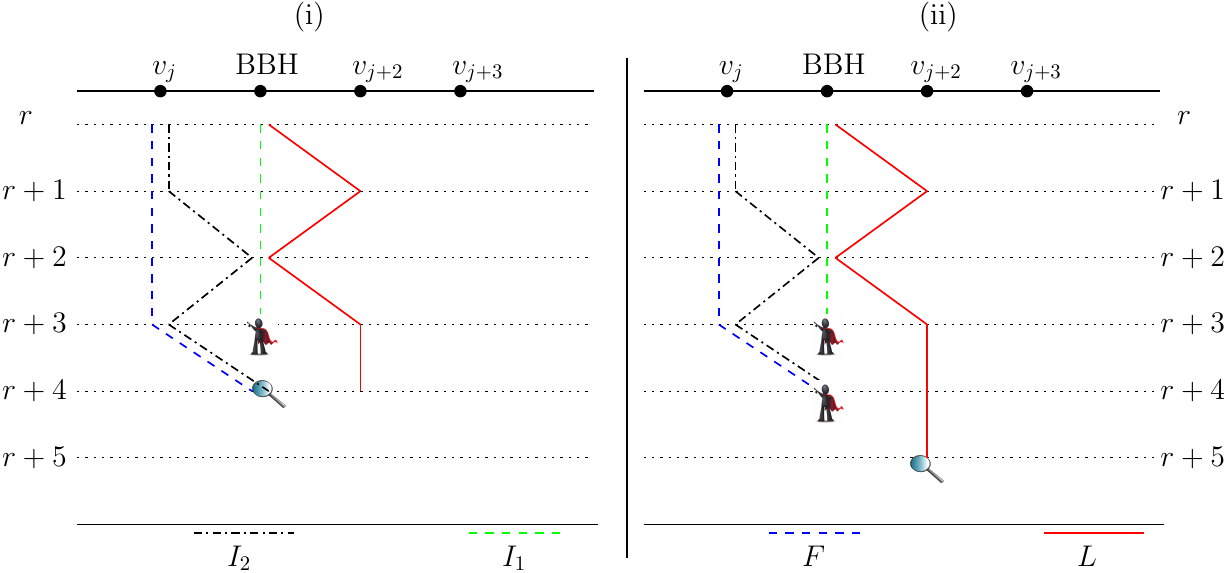}
    \caption{Depicts the time diagram, where $v_{j+1}=\frakb$, in (i) it is activated at $r+3$ and in (ii) it is activated at rounds $r+3$ and $r+4$}
    \label{fig:v1 is BBH2}
\end{figure}

If at round 3, $\frakb$ is activated for the first time in the sub-phase $i$. Then, at round 4, whenever $F$ and $I_2$ visits $v_1$, and if they stay alive, then they find that $I_1$ is missing. Hence, they understand that they are on $\frakb$, and so they return to $v_0\in C_1$, and explore $C_1$ perpetually. In (i) of Fig. \ref{fig:v1 is BBH2}, a similar case is discussed.

On the other hand, if $\frakb$ is active at round $4$ as well, then $F$ and $I_2$ also gets destroyed at $v_1$. In this situation, at round $5$, $L$ present at $v_2\in C_2$, finds that $I_1$ fails to return from $v_1$. Hence, $L$ understands that $v_1$ is $\frakb$ (refer to (ii) of Fig. \ref{fig:v1 is BBH2}). In this situation, since $L\in C_2$ hence it cannot reach $h$ in the current phase. But, it knows the distance between $h$ and $v_1$, and also it knows the exact round since the start of the current phase. Moreover, it also understands when the other two agents, i.e., $F_1$, $F_2$, currently waiting at $h$, starts their cautious movement, and exactly one among them reaches $v_1$ after it identifies $\frakb$ uniquely. Let $t_1$ be that round. In this scenario, $L$ waits until round $t_1-1$ and moves to $v_1$ at round $t_1$ along with exactly one agent, say $F_1$, from the group that was moving cautiously. Next, if $\frakb$ is activated at round $t_1$, then at round $t_1+1$, $F_2$ waiting for $F_1$ at $v_0\in C_1$, finds that $F_1$ fails to return. In this situation, $F_2$ understands that $v_1$ is $\frakb$, and explores $C_1$ perpetually (refer to (ii) of Fig. \ref{fig:perpetuallyexplore-timediagram}). On the other hand, if at $t_1$, $\frakb$ is not activated, then at round $t_1$ it meets with $L$, gathers the knowledge that $v_1$ is $\frakb$, and returns back to $v_0\in C_1$, and explores $C_1$ perpetually (refer to (i) of Fig. \ref{fig:perpetuallyexplore-timediagram}).

\begin{figure}[h]
    \centering
\includegraphics[width=0.7\linewidth]{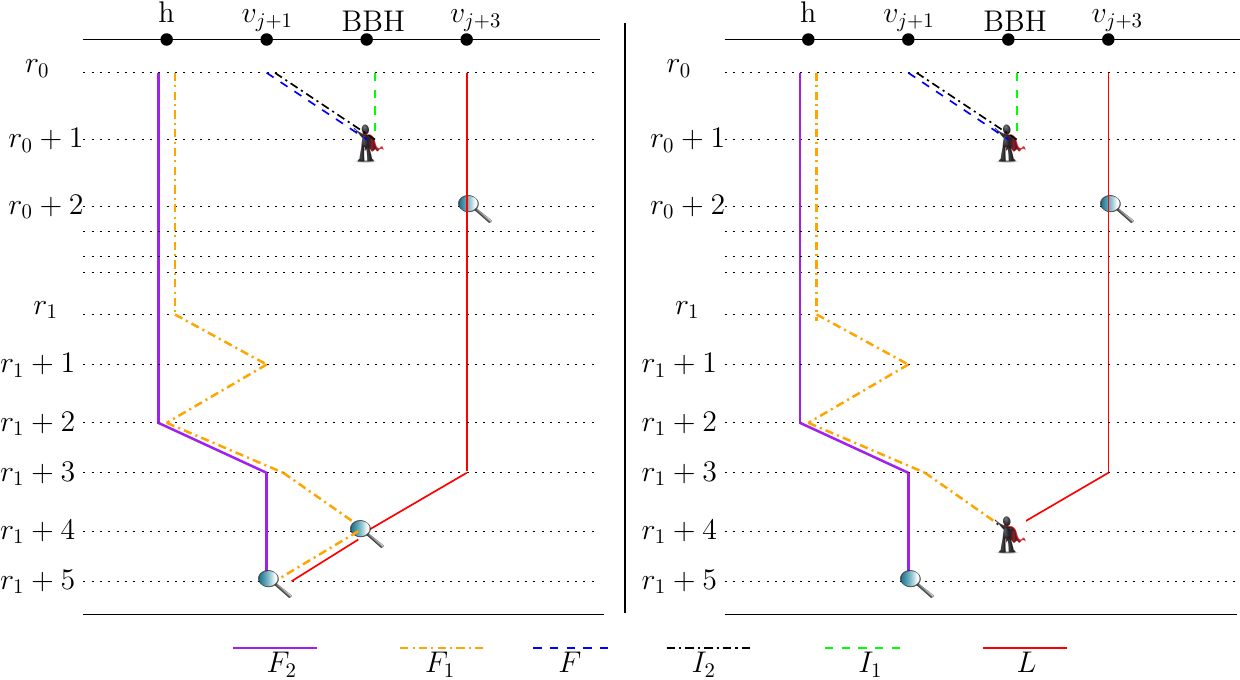}
    \caption{Represents the time diagram, in which at least one among $F_1$ and $F_2$ detects $\frakb$, and perpetually explores the path graph}
    \label{fig:perpetuallyexplore-timediagram}
\end{figure}

In round 4, if $\frakb$ is activated for the first time at $v_1$. Then, $F$, $I_2$ and $I_1$ each are destroyed at $v_1$. In this case, $L$ present at $v_2\in C_2$ understands $v_1$ to be $\frakb$, when at round 5, it finds that $I_1$ fails to return from $v_1$ (refer to Fig. \ref{fig:v1 is BBH3}). By similar argument, as described earlier, at least $F_2$ explores perpetually.

Now if the BBH is activated at $v_1$ in the sub-phase $i$ for the first time in round 5, then it destroys $F$ and $I_2$. In this situation, $I_1$ and $L$ understands this, when they are at $v_2\in C_2$ in round 2 of sub-phase $(i+1)$, after failure of $I_2$'s return from $v_1$ at the same round.

\noindent\underline{\textbf{If $v_0$ is $\frakb$:}} We accordingly describe the cases, which can arise depending on at which round within this sub-phase, $\frakb$ gets activated and destroys at least one among these 4 agents.

\begin{figure}
    \centering
    \includegraphics[width=0.9\linewidth]{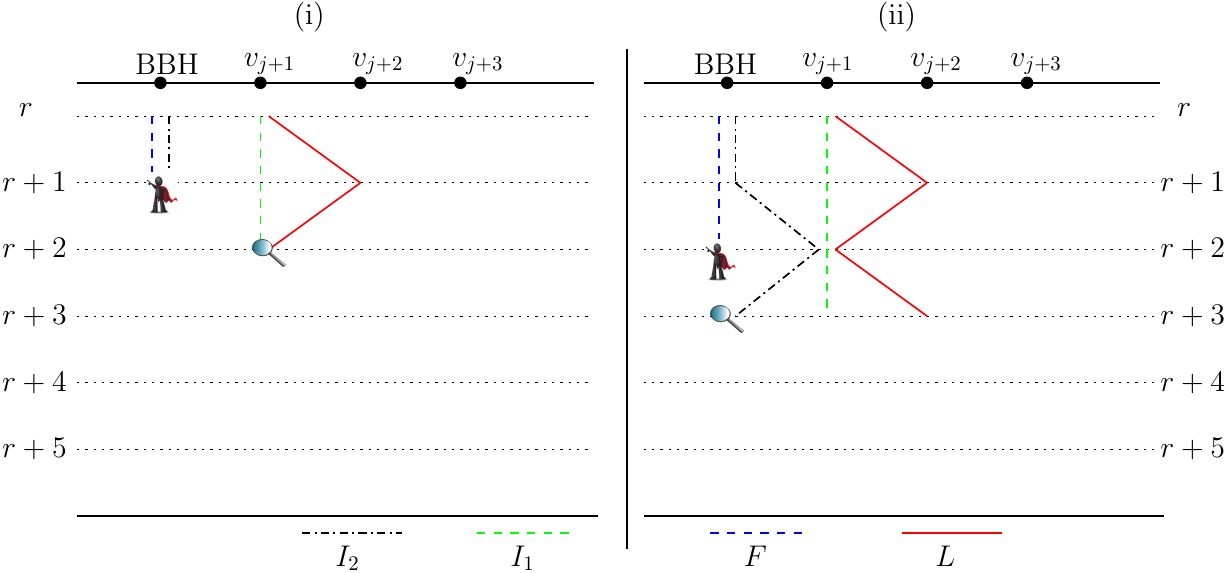}
    \caption{Depicts the time diagram, where $v_{j}=\frakb$, in (i) it is activated at $r+1$ and in (ii) it is activated at round $r+2$}
    \label{fig:v0 is BBH1}
\end{figure}

Let at round $1$, $\frakb$ is activated for the first time in the sub-phase $i$. Then it destroys both $F$ and $I_2$, and this can be understood by $L$ and $I_1$ at round 2, when they find that $I_2$ fails to meet them from $v_0$ at round 2 (refer to (i) of Fig. \ref{fig:v0 is BBH1}). By similar argument, as described for the case when $v_1$ is $\frakb$, here as well at least $F_2$ perpetually explores $C_1$.

If $\frakb$ is activated for the first time at round 2, then $F$ is destroyed at $v_0$. Next, in round 3 of this sub-phase, whenever $I_2$ visits $v_0$ and it is not destroyed, it finds that $F$ is missing. Hence, it understands that $v_0$ is $\frakb$ (refer to (ii) of Fig. \ref{fig:v0 is BBH1}), and moves to $v_1\in C_2$. On the contrary, if $I_2$ also gets destroyed at round 3, then at round 4 of this sub-phase, whenever $I_1$ finds that $F$ and $I_2$ fails to return from $v_0$. $I_1$ understands that $v_0$ is $\frakb$ (refer to Fig. \ref{fig:v0 is BBH2}), and it stays at $v_1\in C_2$, until $F_1$ visits $v_0$. Similar to earlier argument, here as well at least $F_2$ perpetually explores $C_1$.

In round 4 and 5 of this sub-phase, if the BBH is activated for the first time, then since no agents are present on $v_0$, hence none of the agent understands that $v_0$ is the BBH.

\begin{figure}[h!]
    \centering
    \begin{minipage}{0.5\textwidth}
        \centering
        \includegraphics[width=0.86\textwidth]{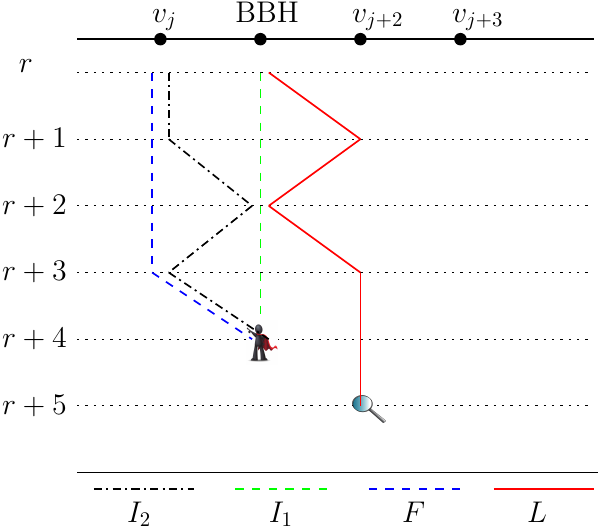}
        \caption{Depicts the time diagram, where $v_{j+1}=\frakb$ and it is activated at round $r+4$}
        \label{fig:v1 is BBH3}
    \end{minipage}%
    \hfill
    \begin{minipage}{0.46\textwidth}
        \centering
        \includegraphics[width=1.0\textwidth]{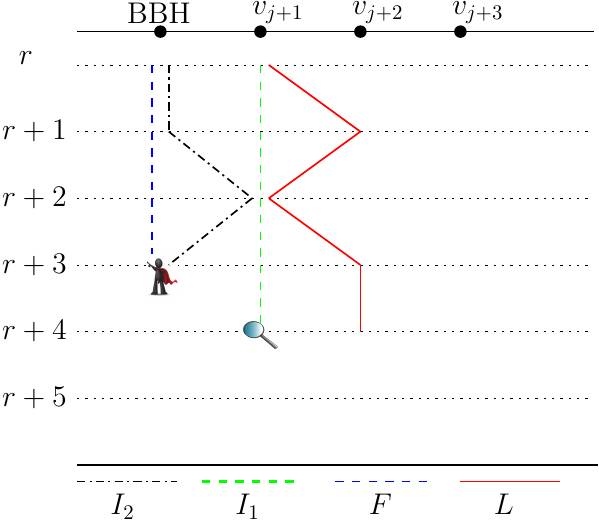}
        \caption{Depicts the time diagram, where $v_{j}=\frakb$ and it is activated at round $r+3$}
        \label{fig:v0 is BBH2}
    \end{minipage}
\end{figure}

\noindent \underline{\textbf{Case-II:}} In this case we look into the case where $v_2$ is the end node (i.e., either the end of the path graph, or the last node to be explored in the current phase) for the current sub-phase. We will only discuss some particular cases, the other cases are similar as the ones described in Case-I. 

Let us consider $i$ to be a sub-phase, at which $L$ first understands that $v_2$ is the end node, so within round 5 of this sub-phase, each agent understands this fact, and at round 5, they change their roles (refer to Fig. \ref{fig:translate-pattern-flip}). Suppose $v_2$ is $\frakb$, and let the adversary activates it at round 5 of sub-phase $i$, while new $I_2$ and $F$ are present at $v_2$. This destroys both of them, and in round 2 of sub-phase $(i+1)$, the new $L$ and $I_1$ finds that new $I_2$ fails to visit $v_1$ from $v_2$, hence they identify $v_1$ to be $\frakb$, and as they are present in $C_1$, so they perpetually explore $C_1$. Refer to (i) of Fig. \ref{fig:time diagram flip1}, where $v_{j+3}$ is the last node, and it is $\frakb$. The adversary activates $\frakb$ at round $r+1$, which symbolises the last round of the earlier sub-phase. In this round, all alive agents change their roles, i.e., earlier $F$, $I_2$ change to $L$ and $I_1$ at $v_{j+1}$. Next round onwards, they start executes the next sub-phase, with their new roles. So, accordingly, at round $r+3$, both $L$ and $I_1$ finds that new $I_2$ fails to arrive from $v_{j+2}$, and accordingly detect $v_{j+2}$ to be $\frakb$, and continue exploring $C_1$ perpetually.

\begin{figure}
    \centering
    \includegraphics[width=0.9\linewidth]{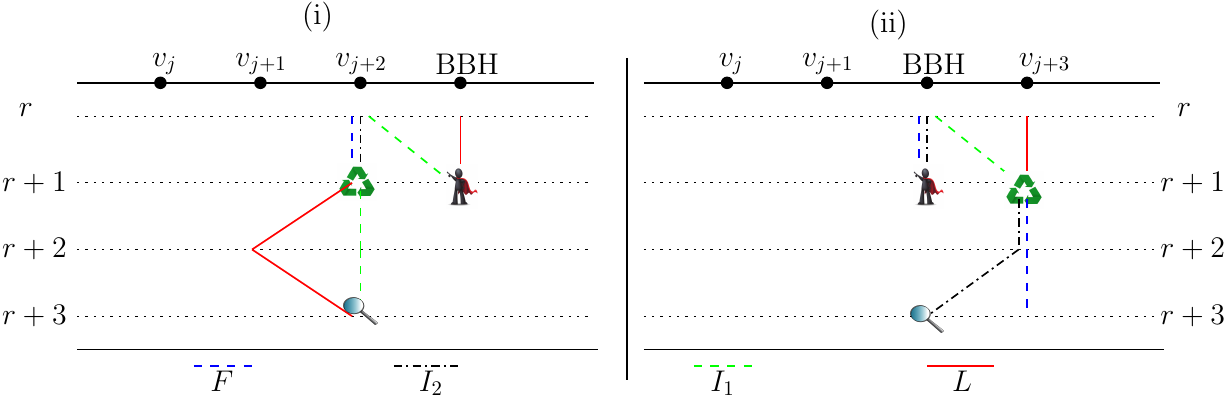}
    \caption{Depicts the time diagram, where in (i) $v_{j+3}=\frakb$ and it is activated at round $r+1$, in (ii) $v_{j+1}=\frakb$ and it is activated at round $r+1$ as well}
    \label{fig:time diagram flip1}
\end{figure}

On the other hand, if suppose $v_1$ is selected to be $\frakb$, and it is activated at round 5 of sub-phase $i$, when they were about to change their roles, after understanding $v_2$ to be the end node. Hence, the new $L$ and $I_1$ gets destroyed, whereas in round 2 of sub-phase $(i+1)$, new $I_2$ visits $v_1$. If $\frakb$ is not activated at this round, $I_2$ finds $L$ and $I_1$ to be missing, hence concludes that current node is $\frakb$ (refer to (ii) of Fig. \ref{fig:time diagram flip1}), and returns to $v_2$. Next, by similar argument as described in Case-I, when $v_1$ is $\frakb$, eventually $F_1$ and $F_2$ starts moving from $h$, cautiously, and at the end at least $F_2$ perpetually explores $C_1$. On the contrary if $\frakb$ is activated at round 2 of sub-phase $(i+1)$, then $I_2$ gets destroyed at $v_2$, and this is detected by new $F$ at $v_2$ at round 3 of sub-phase $(i+1)$, when it finds $I_2$ fails to arrive from $v_1$ (refer to Fig. \ref{fig:time diagram flip2}). It stays at $v_2$, until $F_1$ visits $v_1$, and by similar argument as described in the earlier case, at least $F_2$ explores $C_1$ perpetually. 
\begin{figure}
    \centering
    \includegraphics[width=0.45\linewidth]{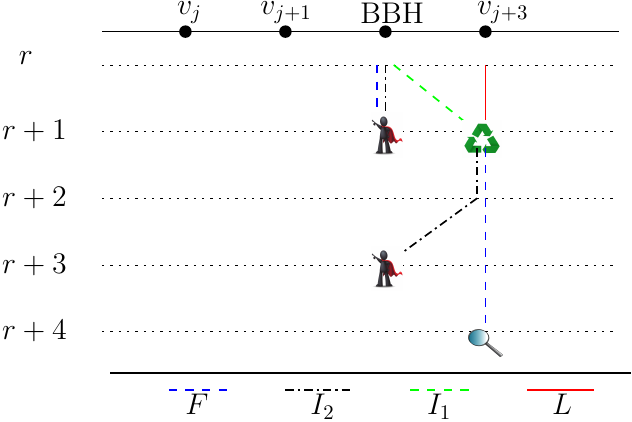}
    \caption{Depicts the time diagram, where $v_{j+2}=\frakb$ and it is activated at rounds $r+1$ and $r+3$.}
    \label{fig:time diagram flip2}
\end{figure}
The other situations of this case, follow exactly similar to the ones in Case-I. Only difference is that, in all situations in Case-I, where the agent detects $\frakb$, the moment it is situated on $\frakb$, then it is asked to visit the previous node (as the previous node is in $C_1$) and explore $C_1$ perpetually, but here in Case-II, in all those situations, the agent must visit the next node (as now the previous node belongs to $C_2$ and the next node belongs to $C_1$, since they changed their direction).

\begin{remark}\label{remark: h is not an extreme end}
    If $h$ is not at an extreme end of the path graph $P$, then the first group of 4 agents, chooses the lowest port direction, first, starting from $h$, at the start of each phase. They make pattern and then translating the pattern along that direction until at most $2^i$ distance (if the current phase is $i-$th phase), then they start returning back to $h$ by translating the same pattern, and thereafter from $h$, makes pattern and translates pattern to the other direction, again until at most $2^i$ distance. Thereafter, using similar movement, returns back to $h$. Only after all agents return to $h$, the current phase ends.
\end{remark}

\begin{theorem}
    Algorithm \textsc{Path\_PerpExplore-BBH-Home} solves \pbmPerpExplHome\ with 6 agents in path graphs, without knowledge of the size of the graph.
\end{theorem}

It may be noted from the high-level idea, in any benign execution, the first set of 4 agents perpetually explores $P$. Otherwise, if the BBH intervenes in the movement of the first 4 agents, either when they are making pattern or translating pattern, then two situations can occur: (1) at least one agent is alive, it knows the precise position of $\frakb$ and it is situated in $C_1$, (2) at least one agent is alive, it knows the precise position of $\frakb$ and it is situated in $C_2$.

In situation (1), since the agent is in $C_1$, and it knows the position of $\frakb$, it can perpetually explore $C_1$, so \textsc{PerpExploration-BBH-Home} is achieved, since $h\in C_1$. In situation (2), it has been shown in Case-I and Case-II earlier, that at least $F_2$ eventually also gets to know the position of $\frakb$, while it is still positioned in $C_1$, and thereafter it perpetually explores $C_1$. In this situation as well, \textsc{PerpExploration-BBH-Home} is achieved.

\begin{corollary}
    The algorithm consisting of \textsc{Make\_Pattern} and \textsc{Translate\_Pattern} solves \pbmPerpExpl\ with 4 agents in path graphs, without knowledge of the size of the graph.
\end{corollary}

The above corollary follows from the fact that, in any benign execution, the 4 agents executing these algorithms explores $P$ perpetually, but if the BBH intervenes, then these algorithms ensure that at least one agent remains alive either in $C_1$ or $C_2$, knowing the exact position of $\frakb$. Hence, thereafter it can perpetually explore its current component, in turn solving \textsc{PerpExploration-BBH}.

\section{A perpetual exploration algorithm for trees with a BBH} \label{Appendix: Tree Alg}

Based on the algorithm we presented in Appendix~\ref{Appendix: Path Alg}, we show how to obtain algorithm \textsc{Tree\_PerpExplore-BBH-Home}, that solves \pbmPerpExplHome\ with 6 agents in trees, without the agents having knowledge of the size of the tree.

Let $\left\langle G,6,h,\frakb\right\rangle$ be a \pbmPerpExplHome\ instance, where $G=(V,E,\lambda)$ is a port-labeled tree. As per Definition~\ref{def:correctalg}, all agents are initially co-located at~$h$ (the \emph{home} node). We consider without loss of generality, the node $h$, to be the root of $G$. Initially four least ID agents, start exploring $G$, while the other two agents wait at $h$. The four agents, namely, $a_0$, $a_1$, $a_2$ and $a_3$, are termed as $L$, $I_1$, $I_2$ and $F$, based on their movements. The exploration they perform is exactly similar to the one explained in Appendix~\ref{Appendix: Path Alg} (i.e., \textsc{Make\_Pattern} and then \textsc{Translate\_Pattern}). But unlike a path graph, where except the parent port (i.e., the port along which the agent has reached the current node from previous node) only one port remains, to be explored from each node, here there can be at most $\Delta-1$ ports to choose, where $\Delta$ is maximum degree in $G$. To tackle this, the agents perform a strategy similar to $k-\texttt{Increasing-DFS}$ \cite{fraigniaud2005graph}. Similar to earlier algorithm on a path graph, the exploration of first four agents is divided in to phases, in the $i$-th phase, the agents explore at most $2^i$ nodes and then returns back to $h$, where in each phase the pattern is translated. As stated in \cite{fraigniaud2005graph}, to explore a graph with diameter at most $2^i$ and maximum degree $\Delta$, performing $k-\texttt{Increasing-DFS}$ where $k\ge \alpha 2^i\log \Delta$ such that $\alpha$ is a constant, for each phase the agents require $\mathcal{O}(2^i\log \Delta)$ memory, i.e., in total $\mathcal{O}(n\log \Delta)$ bits of memory. So, within $T_i\le 5\cdot 2^{i+1}+5$ rounds, if the first four agents fail to reach $h$, then the remaining two agents, $a_4$ and $a_5$, termed as $F_1$ and $F_2$ start moving cautiously, while executing $k-\texttt{Increasing-DFS}$, where $k\ge \alpha 2^i\log \Delta$. As per the earlier argument, if the first four agents fail to return, that implies at least one agent is alive, which knows the exact position of $\frakb$, and it is located in $C_j$ ($1<j\leq i$), where $G-\frakb=C_1\cup\dots \cup C_i$, such that $h\in C_1$. In addition to that, since the two agents $F_1$ and $F_2$ follow exactly same path, as followed by the first four agents, so the only alive agent, say $L$, not only knows the exact round at which $F_1$, $F_2$ started their movement from $h$, but it also knows the path they must follow. So, inevitably, it knows the exact round at which $F_1$ (minimum ID among $F_1$, $F_2$) is scheduled to visit $\frakb$. At the same time, $L$ also jumps to $\frakb$, and as per the earlier argument in case of a path graph, at least $F_2$ gets to know the position of $\frakb$, while it is still in $C_1$, hence it can perpetually explore $C_1$.

The next theorem concludes the successful solution of \pbmPerpExplHome\ with 6 agents in trees, and the proof follows from the preceding discussion and the arguments of Appendix~\ref{Appendix: Path Alg} for path graphs.

\begin{theorem}
    Algorithm \textsc{Tree\_PerpExplore-BBH-Home} solves \pbmPerpExplHome\ with 6 agents in tree graphs, without knowledge of the size~$n$ of the graph. Each agent needs  at most~$\calO(n\log\Delta)$ bits of memory, where $\Delta$ is the maximum degree of the tree.
\end{theorem}

Again, at least one agent among the first four agents remains alive and eventually knows the position of $\frakb$, if the BBH intervenes their movement. Moreover, the alive agent (or agents) with knowledge of $\frakb$ is either in $C_1$ or $C_j$, for some $j\neq 1$. Hence, it can perpetually explore $C_1$ or $C_j$. This concludes our next corollary.

\begin{corollary}
    A modification of Algorithm~\textsc{Tree\_PerpExplore-BBH-Home} solves \pbmPerpExpl\ with 4 agents in tree graphs, without knowledge of the size~$n$ of the graph.Each agent needs  at most~$\calO(n\log\Delta)$ bits of memory, where $\Delta$ is the maximum degree of the tree.
\end{corollary}

\section{Proof of Theorem \ref{thm: equivalent statement}}\label{appendix: lower bound proof General BBH}

In this section, we construct a class of graphs $\mathcal{G}$, and we show that, given any algorithm $\calA$ that claims to solve \pbmPerpExpl\ with $2(\Delta-1)$ agents, where~$\Delta$ is the maximum degree of the graph, an adversarial strategy exists such that it chooses a $G\in\mathcal{G}$ on which $\mathcal{A}$ fails.
 
The construction of~$\mathcal{G}$ is based on incremental addition of blocks as follows.

\noindent\textbf{Block-1 (Constructing Path):} We define a path $\mathcal{P}$, which consists of two types of vertices. The first type of vertices are denoted by $v_i$, for all $i\in\{1,2,\dots,\Delta\}$. The second type of vertices lie between $v_i$ and $v_{i+1}$, for all $i\in\{1,2,\dots,\Delta-1\}$, and they are denoted by $u^i_j$, for all $j\in\{1,2,\dots,l_i\}$, where $l_i+1=dist_{\mathcal{P}}(v_i,v_{i+1})$, such that $dist_{G}(a,b)$ indicates the shortest distance between two vertices $a$ and $b$ in $G$. Fig. \ref{fig:LowerBound-GeneralGraph-PhaseI} illustrates the path graph $\mathcal{P}$, with vertices of type $v_i$ and $u^i_j$.

\begin{figure}[H]
    \centering
    \includegraphics[width=0.7\linewidth]{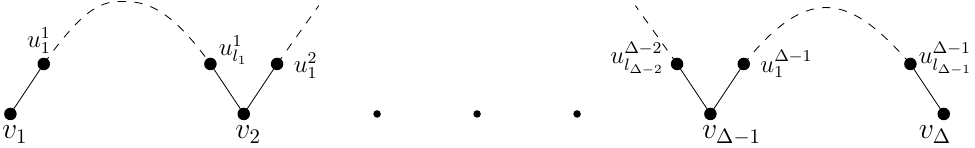}
    \caption{This figure depicts the path graph $\mathcal{P}$, constructed in Block-1 consisting of two types of vertices, where $l_{i-1}$ depicts the $dist(v_{i-1},v_{i})-1$.}
    \label{fig:LowerBound-GeneralGraph-PhaseI}
\end{figure}

\noindent\textbf{Block-2 (Attaching BBH):} We partition the set $\mathcal{V}=\{v_1,v_2,\dots,v_{\Delta-1}\}$ in to two disjoint sets $\mathcal{V}_1$ and $\mathcal{V}_2$. For each $v_i\in\mathcal{V}_1$, the node $\frakb$ (i.e., BBH) is attached to $v_i$, this extension structure is called \texttt{Ext}$_1$. Finally from $\frakb$ a node $z$ is attached, which is further connected to $\Delta-1$ other degree 1 vertices. For each $v_i\in \mathcal{V}_2$, a node $w_i$ is connected, which is further connected to $\Delta-2$ degree 1 vertices. Finally, each $w_i$ is connected to $\frakb$, this type of extension structure is denoted by \texttt{Ext}$_2$. If $|\mathcal{V}_1|=0$ we have a special case\footnote{As a special instance, abusing the definition of $\mathcal{V}$ we can include $v_{\Delta}$ to the set $\mathcal{V}_2$ if $|\mathcal{V}_1|=0$. Otherwise, if $|\mathcal{V}_1|>0$, then $v_{\Delta}\notin \mathcal{V}_1~\text{or}~v_{\Delta}\notin \mathcal{V}_2$.}. An example in Fig. \ref{fig:LowerBound-GeneralGraph-Figure3} is illustrated the aforementioned extension, where at least $v_1,v_{\Delta-1}\in\mathcal{V}_1$ and at least $v_2\in\mathcal{V}_2$.

\begin{figure}[H]
    \centering
    \includegraphics[width=0.7\linewidth]{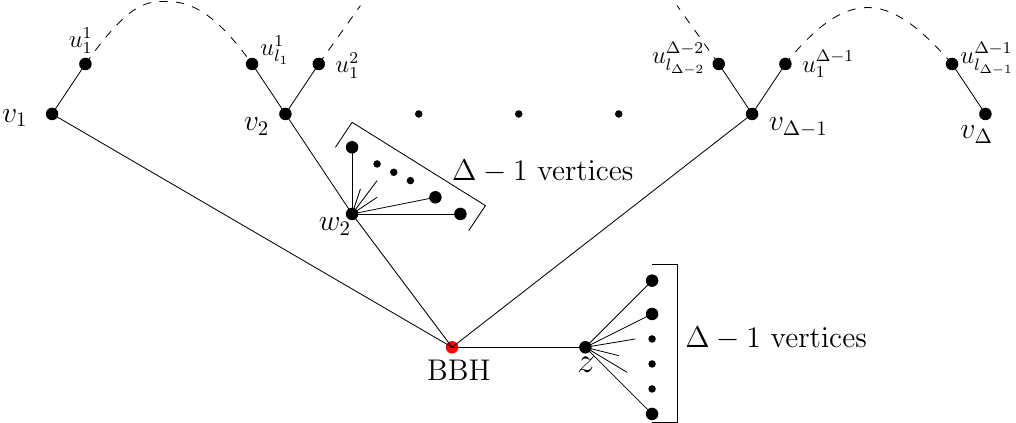}
    \caption{Indicates the addition of \texttt{Ext}$_1$ and \texttt{Ext}$_2$ along $\mathcal{P}$, where at least $v_1,v_{\Delta-1}\in\mathcal{V}_1$ and at least $v_2\in \mathcal{V}_2$.}
    \label{fig:LowerBound-GeneralGraph-Figure3}
\end{figure}

\noindent\textbf{Block-3 (Attaching Trees):} With respect to a vertex $v$, we define a tree $T_v$ as follows: $T_v$ is rooted at $v$ with height 2, $v$ has only one child, say $v'$ and where $v'$ has $\Delta-1$ leaves attached to it. In particular, if from $v$, the edge $(v,v')$ originating from $v$ has port $i$, then we call the tree as $T^i_v$. Next, these trees are attached from each vertex along $\mathcal{P}$ in the following manner:

\textbf{1.} From $v_1\in \mathcal{P}$, $\Delta-2$ such trees are attached, only except the port leading to $u^1_1$ and the port leading to \texttt{Ext}$_1$ or \texttt{Ext}$_2$, depending on $v_1\in\mathcal{V}_1~\text{or}~ \mathcal{V}_2$.

\textbf{2.} From each $v_i\in\mathcal{P}$, where $i\in\{2,\dots,\Delta-1\}$, $\Delta-3$ such trees are attached, except along the ports leading to $u^{i-1}_{l_{i-1}}$, $u^{i+1}_1$ and the one leading to \texttt{Ext}$_1$ or \texttt{Ext}$_2$, depending on $v_i\in\mathcal{V}_1~\text{or}~ \mathcal{V}_2$.

\textbf{3.} From each $u^i_j\in\mathcal{P}$ where $i\in\{1,2,\dots,\Delta-1\}$ and $j\in\{1,2,\dots,l_i\}$, $\Delta-2$ such trees are attached, only except the ports leading to the previous and next node along $\mathcal{P}$.

\textbf{4.} If $v_{\Delta}\in \mathcal{V}_2$, then $\Delta-2$ such trees are attached from $v_{\Delta}$, except along the ports leading to \texttt{Ext}$_2$ and $u^{\Delta-1}_{l_{\Delta-1}}$. If $v_{\Delta}\notin\mathcal{V}_2$, then $\Delta-1$ such trees are attached from $v_{\Delta}$, except along the port leading to $u^{\Delta-1}_{l_{\Delta-1}}$.

\begin{figure}[H]
    \centering
    \includegraphics[width=0.6\linewidth]{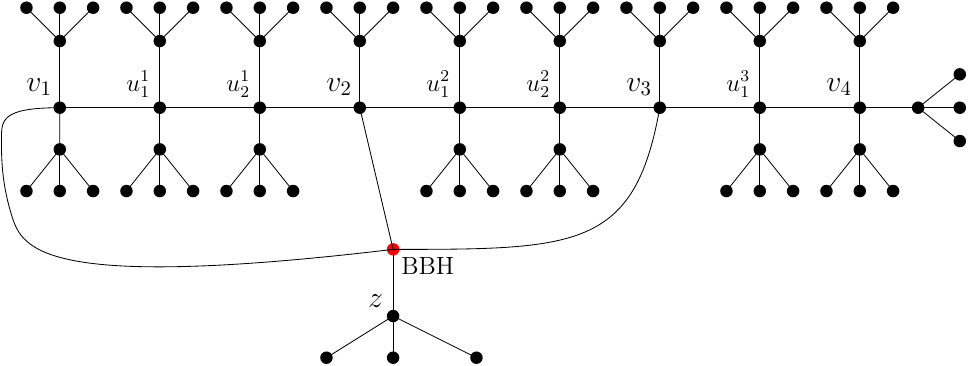}
    \caption{Illustrates an example graph of a member of graph class $\mathcal{G}_1$, where $\Delta=4$, where the distance between $v_1$, $v_2$ is 2, $v_2$, $v_3$ is 2 and $v_3$, $v_4$ is 1.}
    \label{fig:LowerBound-GeneralGraph-G1}
\end{figure}

\noindent\textbf{Final Graph Class}: The final graph class $\mathcal{G}$ is amalgamation of Block-1, Block-2 and Block-3. Now, based on the partition of $\mathcal{V}$, three separate graph subclass can be defined from $\mathcal{G}$, namely $\mathcal{G}_1$, $\mathcal{G}_2$ and $\mathcal{G}_3$. The graph class $\mathcal{G}_1$, preserves the characteristics that, each $v_i$ belong to $\mathcal{V}_1$, where $i\in\{1,2,\dots, \Delta-1\}$, refer Fig. \ref{fig:LowerBound-GeneralGraph-G1}. The graph class $\mathcal{G}_2$, preserves the characteristics that, each $v_i$ belong to $\mathcal{V}_2$, where $i\in\{1,2,\dots,\Delta-1\}$, refer to Fig. \ref{fig:LowerBound-GeneralGraph-G2G3}$(i)$. Lastly, the remaining graphs in $\mathcal{G}$ belong to $\mathcal{G}_3$, refer to Fig. \ref{fig:LowerBound-GeneralGraph-G2G3}$(ii)$.

\begin{figure}
    \centering
    \includegraphics[width=0.6\linewidth]{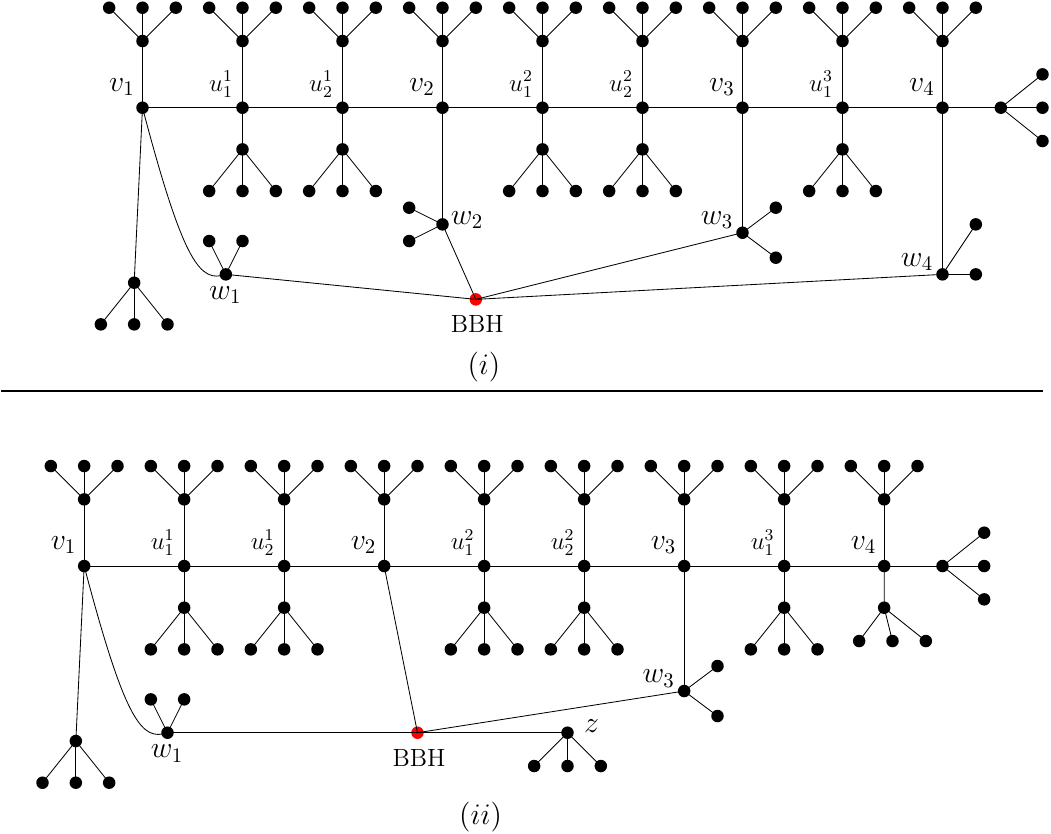}
    \caption{(i) Illustrates an example graph from graph class $\mathcal{G}_2$, (ii) an example graph from graph class $\mathcal{G}_3$, where $v_1,v_3\in \mathcal{V}_2$ and $v_2\in\mathcal{V}_1$. In both graphs the $\Delta=4$, such that distance between $v_1$, $v_2$ is 2, $v_2$, $v_3$ is 2 and $v_3$, $v_4$ is 1.}
    \label{fig:LowerBound-GeneralGraph-G2G3}
\end{figure}


 Given, an algorithm $\mathcal{A}$ which claims to solve \pbmPerpExpl~with $2(\Delta-1)$ agents, the adversarial counter strategy takes the graph class $\mathcal{G}$ and algorithm $\mathcal{A}$ as input, and returns a specific port-labeled graph $G\in\mathcal{G}$. We will prove that $\mathcal{A}$ fails on $G$.

Before formally explaining the counter strategies used by the adversary for all choices (i.e., the instructions given to the agents during the execution) the algorithm $\mathcal{A}$ can have, we define some of the functions and terminologies to be used by the adversary to choose a graph from $G\in \mathcal{G}$. Let $V_{exp}$ denote the nodes in $G$, already explored by the collection of agents executing $\mathcal{A}$, and $V^c_{exp}$ denote the remaining nodes to be explored, where $V=V_{exp}\cup V^c_{exp}$. The function $\lambda^1_{u,\mathcal{A}}: E_u\rightarrow \{1,\dots,\delta_u\}$ denotes the port label ordering to be returned from each node of $u\in V\setminus\{v_1,v_2,\dots,v_{\Delta}\}$, based on the algorithm $\mathcal{A}$. The function $\lambda^2_{v,\mathcal{A}}:E_v\rightarrow\{1,\dots,\delta_v\}$ returns the port label ordering from each $v\in\{v_1,v_2,\dots,v_{\Delta}\}$, based on the algorithm $\mathcal{A}$. So, $\lambda^1=(\lambda^1_{u,\mathcal{A}})_{u\in V\setminus\{v_1,\dots,v_\Delta\}}$ and $\lambda^2=(\lambda^2_{v,\mathcal{A}})_{v\in\{v_1,\dots,v_\Delta\}}$ are denoted to be the set of port label, and finally $\lambda=\lambda^1\cup\lambda^2$ is defined to be the collection of port labels on $G$. Next, we define the distance function $D_{\mathcal{A}}:\{v_1,\dots,v_{\Delta}\}\times\{v_1,\dots,v_{\Delta}\}\rightarrow \mathbb{N}$, which assigns the graph to be chosen from $\mathcal{G}$, to have a distance of $dist_{\mathcal{P}}(v_i,v_j)$ from $v_i$ to $v_j$ along $\mathcal{P}$ (defined in Block-1) on the chosen graph, for $i\neq j$ and $v_i,v_j\in \{v_1,v_2,\dots,v_{\Delta}\}$.

The adversarial counter strategy works as follows: $v_1$ is set as \emph{home} for any graph chosen from $\mathcal{G}$, i.e., all the agents are initially co-located at \emph{home}. Next, we discuss all possible choices $\mathcal{A}$ can use from $v_1$, then from $v_2$ and finally from $v_i$ (where $i\in\{2,\dots,\Delta-1\}$) one after the other. Accordingly, we state the counter strategies of the adversary to choose a graph from $\mathcal{G}$, after each such choices. We start with \emph{home} (i.e., $v_1$).

\paragraph*{Choices of algorithm $\mathcal{A}$ and counter strategies of adversary at \emph{home}}


The agents are initially co-located at $v_1$, since it is designated as \emph{home}. Based on the choices that an algorithm can have from $v_1$, we explain the counter measures taken by the adversary to choose a certain graph from $\mathcal{G}$ as per $\mathcal{A}$, accordingly.

\textit{Choice-1:} If at the first step, $\mathcal{A}$ assigns at least 2 agents to perform the first movement, and suppose it is along a port $j$ from $v_1$, where $j\in\{1,\dots, \delta_{v_1}\}$. 

\textit{Counter:} The adversary chooses a graph in which $v_1\in \mathcal{V}_1$. Moreover, the port-labeled function at $v_1$, $\lambda^2_{v_1,\mathcal{A}}$, returns an ordering where the port from $v_1$ along the edge $(v_1,\frakb)$ is $j$. 

\textit{Choice-2:} As per the execution of $\mathcal{A}$, let $r_1$ (for some $r_1>0$) be the first round, at which any one agent from $v_1$ travels a node, which is at a 2 hop distance.

\textit{Counter:} If there exists a round $r'_1<r_1$ at which more than one agent visits a neighbor node of $v_1$ with respect to some port $j$, then in that scenario, adversary selects $v_1\in \mathcal{V}_1$ and returns $\lambda^2_{v_1,\mathcal{A}}$ such that $v_1$ is connected to $\frakb$ via port $j$. On the other hand, if there does not exist such a round $r'_1$, then choose $v_1\in \mathcal{V}_1~\text{or}~v_1\in\mathcal{V}_2$, if $v_1\in \mathcal{V}_1$ then return $\lambda^2_{v_1,\mathcal{A}}$ and $\lambda^1_{\frakb,\mathcal{A}}$ such that at $r_1-1$ and $r_1$ rounds, the single agent must be at $\frakb$ and $z$, respectively. If $v_1\in\mathcal{V}_2$ then return $\lambda^2_{v_1,\mathcal{A}}$ and $\lambda^1_{w_1,\mathcal{A}}$ (where $w_1$ is part of \texttt{Ext}$_2$, refer explanation in Block-2) such that at rounds $r_1-1$ and $r_1$, the single agent must be at $w_1$ and $\frakb$, respectively.

\begin{lemma}\label{lemma:v1SuspiciousNode2}
As per Counter of Choice-2 if there exists no $r'_1<r_1$, then adversary can activate $\frakb$, such that even after destruction of one agent within round $r_1$, remaining alive agents cannot know exact node of $\frakb$ from $v_1$.
\end{lemma}

\begin{proof}
    Consider two instances, first $v_1\in \mathcal{V}_1$ and second $v_1\in\mathcal{V}_2$. In both the instances, the adversary activates $\frakb$ at rounds $r_1-1$ and $r_1$, respectively. Now, since no new information is gained by the agent till round $r_1-1$ considering the fact that $\delta_{\frakb}=\delta_z$ if $v_1\in\mathcal{V}_1$ and $\delta_\frakb=\delta_{w_1}=\Delta$, where $w_1$ is part of \texttt{Ext}$_2$ if $v_1\in\mathcal{V}_2$. So this means the agent will invariably move to $\frakb$ from $w_1$, if $v_1\in\mathcal{V}_2$ and to $z$ from $\frakb$, if $v_1\in\mathcal{V}_2$ at round $r_1$. This shows that, even after an agent is destroyed, the remaining agents cannot know, which among the two nodes along this 2 length path from $v_1$ is indeed $\frakb$.\end{proof}

 As per above lemma it is shown that, even after destruction of one agent from $v_1$ by $\frakb$, the exact location of $\frakb$ from $v_1$ cannot be determined. But, since $\mathcal{A}$ aims to solve \pbmPerpExpl, so $\mathcal{A}$ needs to destroy at least one more agent from $v_1$ in order to detect the exact position of $\frakb$ from $v_1$. Let that round at which the second agent gets destroyed by $\frakb$ from $v_1$ be $r''_1$. So, $[r_1,r''_1]$ signifies the total number of rounds between the destruction of the first agent and the second agent along $v_1$. Finally, after all the possible choices that can arise from $v_1$ as per execution of $\mathcal{A}$, and its respective counter measures of the adversary, the possible choices of graph class reduces to $\mathcal{G}^1$, where $\mathcal{G}^1\subset \mathcal{G}$. 

\begin{theorem}\label{theorem:v1-2agentsdestroyed}
    At least 2 agents are destroyed by $\frakb$ from $v_1$.
\end{theorem}

The proof of the above theorem follows from the Counter of Choice-1 and the conclusion of Lemma \ref{lemma:v1SuspiciousNode2}, i.e., in other words it is shown that for any choice $\mathcal{A}$ takes from $v_1$, at least 2 agents are destroyed by $\frakb$, within 2 hop of $v_1$. Next, from $v_1$ the agents eventually reach $v_2$, while executing $\mathcal{A}$. Based on the choices that $\mathcal{A}$ can have from $v_2$, we discuss the counter strategies of the adversary in the next section.

\paragraph*{Choices of algorithm $\mathcal{A}$ and counter strategies of adversary at $v_2$}

 We explain the choices that an algorithm $\mathcal{A}$ can have while exploring new nodes from $v_2$, accordingly, we present the adversarial counters. The following lemma discusses the fact that, after reaching $v_2$ for the first time, any agent trying to explore $V^c_{exp}$ from $v_2$, needs to traverse at least 2 hops from $v_2$.

\begin{lemma}\label{lemma:atleast2hopsfromv_2}
    In order to explore $V^c_{exp}$ from $v_2$, $\mathcal{A}$ must instruct at least one agent to travel at least 2 hops from $v_2$.
\end{lemma}

\begin{proof}
    Any graph in which a node of the form $v_2$ exists, as per the construction of $\mathcal{G}$, there exists vertices at least 2 hop distance apart, which are only reachable through $v_2$. This implies that, if no agent from $v_2$ visits a node which is at 2 hop distance, then there will exist some vertices which will never be explored by any agent, in turn contradicting our claim that $\mathcal{A}$ solves \pbmPerpExpl. 
\end{proof}

First, we discuss all possible knowledge that the set of agents can acquire, before reaching $v_2$ from $v_1$. Let at least one agent reaches $v_2$ at round $r^{\circ}_2$ for the first time from $v_1$, and $r_2$ ($>r^{\circ}_2$) is the first round when at least one agent is destroyed from $v_2$ within at most 2 hops of $v_2$. The agents can gain the map of the set $V_{exp}$, explored yet. Define $r''_1-r_1=t_1$, $r_2-r^{\circ}_2={Wait}_2$ and $t_1-{Wait}_2=Time_1$. Next we define the concept of \textit{Conflict-Free} for any node $v_i\in\mathcal{V}$.

\begin{definition}[Conflict-Free]
    A node $v_i\in G$ (for some $G\in\mathcal{G}$) is said to be Conflict-Free, if as per the execution of $\mathcal{A}$, any agent first visits $v_i$ at round $r^{\circ}_i$, and round $r_i$ ($r_i>r^{\circ}_i$) be the first round after $r^{\circ}_i$, at which at least one agent gets destroyed by $\frakb$ within at most 2 hop distance of $v_i$, while moving from $v_i$, then the adversary must ensure the following condition:
\begin{itemize}
     \item Within the interval $[r^{\circ}_i,r_i]$, no agent from $v_j$, for all $j\in\{1,2,\dots,i-1\}$, tries to visit along \texttt{Ext}$_1$ if $v_j\in\mathcal{V}_1$ or \texttt{Ext}$_2$ if $v_j\in\mathcal{V}_2$.     
     \end{itemize}
\end{definition}

 To ensure, $v_2$ to be Conflict-Free, the adversary sets the distance between $v_1$ and $v_2$ along $\mathcal{P}$ as follows: if $Time_1>0$, and there exists an $l_1\in\mathbb{N}$, such that $c_1\cdot \Delta^{l_1+1}<Time_1$, where $c_1$ (is a constant) is the maximum round for which each alive agent remains stationary while moving from $v_1$ towards $v_2$, then return $D_{\mathcal{A}}(v_1,v_2)=l_1+1$, and modify the port-labeling of $v_1$ using $\lambda^2_{v_1,\mathcal{A}}$, and insert the port label of $u^1_i$ using $\lambda^1_{u^1_i,\mathcal{A}}$ (for all $i\in\{1,2,\dots,l_1\}$) such that the first agent's $l_1$ distance movement from $v_1$ after round $r_1$ is to $v_2$. Otherwise, return $D_{\mathcal{A}}(v_1,v_2)=r''_1+1$. Call $D_{\mathcal{A}}(v_1,v_2)=l_1+1$. Again, modify the port-labeling of $v_1$ using $\lambda^2_{v_1,\mathcal{A}}$, and insert the port label of $u^1_i$ using $\lambda^1_{u^1_i,\mathcal{A}}$ (for all $i\in\{1,2,\dots,l_1\}$) such that the first agent's $l_1$ distance movement from $v_1$ after round $r''_1$ is to $v_2$. So, this means the graph to be chosen from $\mathcal{G}$ must have $u^1_1,\dots,u^1_{l_1}$ many nodes between $v_1$ and $v_2$. In the following lemma we prove that $v_2$ is Conflict-Free.

\begin{lemma}\label{lemma:v2ConflictFree}
$v_2$ is Conflict-Free.     
\end{lemma}

\begin{proof}
    Without loss of generality, a single agent performed the first 2 hop distance movement from $v_1$, and got destroyed at round $r_1$. As per conclusion of Lemma \ref{lemma:v1SuspiciousNode2}, at least one other also gets destroyed through $v_1$ at round $r''_1$, in order to determine the exact position of $\frakb$. If two agents first performed a movement from $v_1$, in that case $r_1=r''_1$. Now, $[r_1+1,r''_1-1]$ is the interval within which no agent from $v_1$ tried to visit two possible positions of $\frakb$ from $v_1$, and moreover, the adversary acts in such a way that, at round $r''_1+1$ onwards, each agent trying to visit from $v_1$, knows the exact position of $\frakb$. As per the construction of $\mathcal{G}$, to reach $v_2$ from $v_1$, except using the nodes connected to $v_1$ via \texttt{Ext}$_1$ or \texttt{Ext}$_2$, any agent takes at most $c_1\cdot\Delta^{l_1+1}$ rounds, where $c_1$ signifies the maximum number of rounds for which each alive agent remains stationary during their movement from $v_1$ towards $v_2$. Next, $Time_1$ calculates the number of rounds remaining after subtracting $t_1=r''_1-r_1$ (i.e., the number of rounds between the first and second agent destruction from $v_1$) with ${Wait}_2=r_2-r^{\circ}_2$ (i.e., the number of rounds between an agent first arrives at $v_2$ and an agent gets destroyed by $\frakb$ from $v_2$). There are two conditions for choosing the graph:
    
    \textit{Condition-1:} If there exists an $l_1\in\mathbb{N}$ such that it satisfies the condition $c_1\cdot\Delta^{l_1+1}<Time_1$, then the adversary returns $D_{\mathcal{A}}(v_1,v_2)=l_1+1$, and the adversary chooses a graph from $\mathcal{G}$, such that $dist_{\mathcal{P}}(v_1,v_2)=l_1+1$. Also, it performs the port labeling at the nodes $v_1,u^1_1,\dots,u^1_{l_1}$, in a way that, whenever an agent, executing $\mathcal{A}$ visits a neighbor at a distance $l_1$ for the first time from $v_1$, then it reaches the node $v_2$. Set this round as $r^{\circ}_2$. In the worst case, $r^{\circ}_2=r_1+c_1\cdot\Delta^{l_1+1}$. It may be noted that, since $Time_1=t_1-{Wait}_2>c_1\cdot\Delta^{l+1}$, this implies $t_1>{Wait}_2+c_1\cdot\Delta^{l+1}$, which also implies $r''_1>r_1+r_2-r^{\circ}_2+c_1\cdot\Delta^{l+1}$. This signifies that during the interval $[r^{\circ}_2,r_2]$ no agent from $v_1$ tries to locate $\frakb$, hence $v_2$ is Conflict-Free.

    \textit{Condition-2:} Otherwise, if there does not exist such $l_1$ to satisfy the above condition, then the adversary returns $D_{\mathcal{A}}(v_1,v_2)=r''_1+1$, we call it $l_1+1$. In this situation, the adversary activates $\frakb$ in such a way that, any agent from $v_1$ knows the exact position of $\frakb$ from round $r''_1+1$ onwards, and we assume that, further no agent from $v_1$ tries to visit $\frakb$ along \texttt{Ext}$_1$ or \texttt{Ext}$_2$ (if it does, then a simple modification of $D_{\mathcal{A}}(v_1,v_2)$ will ensure further that, again $v_2$ is Conflict-Free). This shows that, since $r^{\circ}_2\ge r''_1+1$, so it implies that $v_2$ is Conflict-Free.  \end{proof}

Next, since, $\frakb\in V_{exp}$ (as per Choice-2 from the node $v_1$, if $v_1\in\mathcal{V}_1$), so as part of the map of $V_{exp}$, the agents can know the port labeling of the edge $(v_1,\frakb)$. Based on these, $\mathcal{A}$ can have two class of choices: the agents use the knowledge of the port label of $(v_1,\frakb)$ and second they do not. We call these first class of choices as \textit{Choice-A} class and second as \textit{Choice-B} class. We discuss all possible choices in \textit{Choice-A} class first. 

\textit{Choice-A1:} Algorithm $\mathcal{A}$ may instruct at least 2 agents to explore a neighbor of $V^c_{exp}$ at round $r_2$ along port $j$, where $r_2>r^{\circ}_2$ and $j\in\{1,2,\dots,\delta_{v_2}\}$. 

\textit{Counter:}  The adversary sets $v_2\in\mathcal{V}_1$, and accordingly it chooses a graph in $\mathcal{G}$ satisfying this criteria. Moreover, the port-labeled function at $v_2$, i.e., $\lambda^2_{v_2,\mathcal{A}}$, returns an ordering where the port from $v_2$ towards the edge $(v_2,\frakb)$ is $j$.

\textit{Choice-A2:} As per the execution of $\mathcal{A}$, let $r_2$ ($>r^{\circ}_2$) be the first round, at which a single agent travels a node which is at 2 hop distance of $v_2$ in $V^c_{exp}$.

\textit{Counter:} If there exists a round $r'_2$, such that $r^{\circ}_2<r'_2<r_2$, at which more than one agent tries to visit a neighbor of $v_2$ in $V^c_{exp}$ with respect to some port $j$, then in that scenario, adversary selects $v_2\in\mathcal{V}_1$ and returns $\lambda^2_{v_2,\mathcal{A}}$ such that $v_2$ is connected to $\frakb$ via port $j$. On the other hand, if there does not exists such $r'_2$, then choose $v_2\in\mathcal{V}_1~\text{or}~v_2\in\mathcal{V}_2$. If $v_2\in\mathcal{V}_1$, then return $\lambda^2_{v_2,\mathcal{A}}$ and modify $\lambda^1_{\frakb,\mathcal{A}}$ such that at rounds $r_2-1$ and $r_2$, the agent is at $\frakb$ and at any neighbor of $\frakb$ except $v_2$. If $v_2\in\mathcal{V}_2$, then return $\lambda^2_{v_2,\mathcal{A}}$ such that at round $r_2-1$ the agent is at $w_2$, and return $\lambda^1_{w_2,\mathcal{A}}$ such that the port label of $w_2$ to $\frakb$ is exactly same as the port label of $\frakb$ to $v_1$. 

\begin{lemma}\label{lemma:v2-two-destroyed-together}
    If $\mathcal{A}$ instructs at least two agents to move simultaneously to a neighbor of $v_2$ in $V^c_{exp}$, then the adversary can destroy all of these agents, even if $\mathcal{A}$ uses the knowledge of the map of $V_{exp}$ during its execution from $v_2$.
\end{lemma}

The proof of the above lemma is simple, as after the agent's reach $v_2$ from $v_1$ at round $r^{\circ}_2$, if $\mathcal{A}$ decides to send at least two agents from $v_2$ along port $j$, where the $j$-th port does not belong to the current map, in that case, the adversary can choose $v_2\in\mathcal{V}_1$ and set the port connecting $v_2$ to $\frakb$ as $j$. So, this shows that even if the algorithm $\mathcal{A}$ uses the map of $V_{exp}$, the adversary is able to destroy all the agents, instructed to move from $v_2$ after round $r^{\circ}_2$ for the first time.

\begin{lemma}\label{lemma:v2-twosuspiciousnode-A}
    As per Counter of Choice-A2, if there exists no $r'_2$, such that $r^{\circ}_2<r'_2<r_2$, then the adversary can activate $\frakb$, such that after destruction of one agent from $v_2$ within round $r_2$, remaining agents cannot know the exact node from $v_2$ to $\frakb$.
\end{lemma}

\begin{proof}
    Since the Counter is part of Choice-A class, hence $\mathcal{A}$ uses the knowledge of the port from $\frakb$ to $v_1$, while exploring 2 hops from $v_2$ for the first time after round $r^{\circ}_2$. Let that port connecting $\frakb$ to $v_1$ be $\rho$ (where $\rho\in\{1,\dots,\Delta\}$). First, it must be noted that, the distance from $v_1$ to $v_2$ through $\frakb$ along an induced subgraph of $G$, defined by $H=(V\setminus\{u^1_1,\dots,u^1_{l_1}\},E)$ where $G\in\mathcal{G}$ is at least 2. So, in order to reach $v_1$ from $v_2$, any agent, say $a_i$, needs to visit a node which is at least 2 hop distance of $v_2$ on $H$. So, to make use of the knowledge of the port from $\frakb$ to $v_1$, the algorithm $\mathcal{A}$ can use only the following strategy. $\mathcal{A}$ uses the gained knowledge, places some agent at $v_1$, and instructs $a_i$ to use the port $\rho$ at $r_2-1$ round, since at round $r_2$ any agent from $v_2$ visits a node which is at 2 hop distance, for the first time after $r^{\circ}_2$.

    In this strategy, as per our Counter of Choice-A2, the adversary can create two instances, first $v_2\in\mathcal{V}_1$ (refer to Instance-1 in Fig. \ref{fig:2InstanceChoice-A2}), and activates $\frakb$ at round $r_2-1$. In the second instance, $v_2\in\mathcal{V}_2$, where the port from $w_2$ to $\frakb$ is $\rho$ (refer Instance-2 in Fig. \ref{fig:2InstanceChoice-A2}), and activates $\frakb$ at round $r_2$. Since, both $\frakb$ and $w_2$ are of degree $\Delta$, the agent $a_i$ till round $r_2-1$ do not gain any different knowledge, hence in the first instance it gets destroyed at $r_2-1$ and in second instance gets destroyed at $r_2$. The remaining alive agents (including the one waiting at $v_1$), does not gain any different knowledge, as in both instance $a_i$ fails to reach $v_1$ from round $r_2$ onwards. So, the remaining agents cannot determine, the location of $\frakb$ from $v_2$.\end{proof}

\begin{figure}
    \centering
    \includegraphics[width=0.77\linewidth]{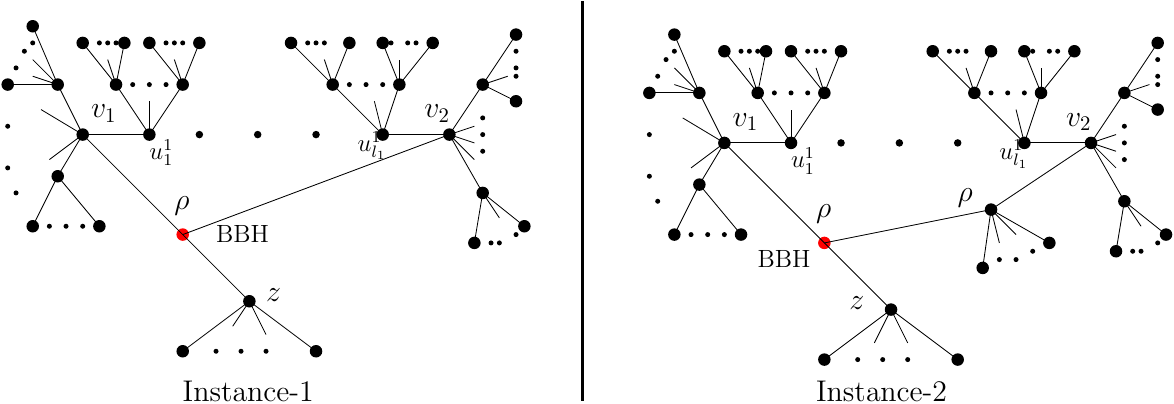}
    \caption{Illustrates the two instances, which the adversary can create in \textit{Counter} of \textit{Choice-A2}.}
    \label{fig:2InstanceChoice-A2}
\end{figure}

    Next, we discuss the choices in the class \textit{Choice-B}, i.e., in which the agent does not use the knowledge of the port label from $\frakb$ to $v_1$.

    \textit{Choice-B1:} $\mathcal{A}$ may instruct at least 2 agents to explore a neighbor of $V^c_{exp}$ at round $r_2$ along port $j$, where $r_2>r^{\circ}_2$ and $j\in\{1,2,\dots,\delta_{v_2}-1\}$.

    \textit{Counter:} Similar counter as the one described in Choice-A1 for the node $v_2$.

    \textit{Choice-B2:} As per the execution of $\mathcal{A}$, let $r_2$ be the first round, at which a single agent travels to a node which is at 2 hop distance of $v_2$ in $V^c_{exp}$.

    \textit{Counter:} If there exists a round $r'_2$, such that $r^{\circ}_2<r'_2<r_2$, at which more than one agent visits a neighbor of $v_2$ with respect to some port $j$, then in that scenario, adversary selects $v_2\in \mathcal{V}_1$ and returns $\lambda^2_{v_2,\mathcal{A}}$ such that $v_2$ is connected to $\frakb$ via port $j$. If there does not exist such $r'_2$, then choose $v_2\in\mathcal{V}_1~\text{or}~v_2\in\mathcal{V}_2$. If $v_2\in\mathcal{V}_1$, then return $\lambda^2_{v_2,\mathcal{A}}$ and modify $\lambda^1_{\frakb,\mathcal{A}}$ such that at rounds $r_2-1$ and $r_2$, the agent is at $\frakb$ and at any neighbor of $\frakb$, except $v_2$. If $v_2\in\mathcal{V}_2$, then return $\lambda^2_{v_2,\mathcal{A}}$ and $\lambda^1_{w_2,\mathcal{A}}$ such that at rounds $r_2-1$ and $r_2$, the agent is at $w_2$ and $\frakb$.

\begin{lemma}\label{lemma:v2-twosuspiciousnode-B}
    As per Counter of Choice-B2, if there exists no $r'_2$, then the adversary can activate $\frakb$, such that after the destruction of one agent from $v_2$ within round $r_2$, remaining agents cannot know the exact node from $v_2$ to $\frakb$.
\end{lemma}

The idea of the proof is similar to Lemma \ref{lemma:v1SuspiciousNode2}. Moreover, we can similarly conclude that, at least one more agent must be destroyed by $\frakb$ from $v_2$ irrespective of the knowledge gained by the agents while traversing from $v_1$ (refer to Lemmas \ref{lemma:v2-two-destroyed-together}, \ref{lemma:v2-twosuspiciousnode-A} and \ref{lemma:v2-twosuspiciousnode-B}), and that round is denoted as $r''_2$.

\begin{corollary}\label{corollary:Exploredsetfromv2}
    Given any graph from $\mathcal{G}$, at round $r^{\circ}_2-1$, the set $V_{exp}$, can only contain the nodes $v_1$, $u^1_j$ (for all $j\in\{1,2,\dots,l_1\}$) along $\mathcal{P}$ of $G\in\mathcal{G}$.
\end{corollary}

The above corollary is a direct consequence of Lemma \ref{lemma:v1SuspiciousNode2} and Lemma \ref{lemma:v2ConflictFree}, as any path to a node in $V^c_{exp}$ from $v_1$ must either pass through $\frakb$ or through $v_2$. Since, $v_2$ is conflict free, that implies till round $r^{\circ}_2$, either one agent is destroyed at $\frakb$ from $v_1$, or the adversary can activate $\frakb$ in such a way that at least 2 agent's have been destroyed and no more agent can cross $\frakb$ without being destroyed. Moreover, $r^{\circ}_2$ indicates the first round any agent moves from $v_1$ to $v_2$. This proves the corollary, claiming $V_{exp}$ can contain only the nodes $v_1$, $u^1_j$ (for all $j\in\{1,2,\dots,l_1\}$) along $\mathcal{P}$ of $G\in\mathcal{G}$.

\begin{theorem}\label{theorem:v2-2agentsdestroyed}
    At least 2 agents are destroyed by $\frakb$ from $v_2$.
\end{theorem}

\begin{proof}
The adversary chooses $G\in\mathcal{G}$, such that $v_2$ is Conflict-Free. This implies, the adversary does not require to activate $\frakb$ between $[r^{\circ}_2,r_2]$ due to any agent's movement from $v_1$, towards \texttt{Ext}$_1$ or \texttt{Ext}$_2$ (depending whether $v_1\in\mathcal{V}_1~\text{or}~v_1\in\mathcal{V}_2$). Next, as per Lemmas \ref{lemma:v2-two-destroyed-together}, \ref{lemma:v2-twosuspiciousnode-A} and \ref{lemma:v2-twosuspiciousnode-B}, shows that irrespective of the knowledge gained by the agents, there exists a round $r''_2$ ($\ge r_2$) at which the second agent gets destroyed from $v_2$, while traversing along \texttt{Ext}$_1$ or \texttt{Ext}$_2$ from $v_2$ (depending on $v_2\in\mathcal{V}_1~\text{or}~v_2\in\mathcal{V}_2$). Hence, this proves the theorem.\end{proof}

Finally, after all possible choices that can arise from $v_2$ as per execution of $\mathcal{A}$, and its respective counter measures of the adversary, the possible choices of the graphs reduces to $\mathcal{G}^2$, where $\mathcal{G}^2\subset \mathcal{G}^1 \subset \mathcal{G}$. Next, in general we discuss the choices $\mathcal{A}$ can have after reaching $v_i$ (where $i\in\{3,\dots,\Delta-1\}$), and accordingly discuss the adversarial counters in the following section.

\paragraph*{Choices of algorithm $\mathcal{A}$ and counter strategies of adversary at $v_i$, for $2< i \le \Delta-1$}

In this case as well, we explain the choices that an algorithm $\mathcal{A}$ can have while exploring new nodes from $v_i$, accordingly we present the adversarial counters.

\begin{lemma}
    In order to explore $V^c_{exp}$ from $v_i$, $\mathcal{A}$ must instruct at least one agent to travel at least 2 hops from $v_i$.
\end{lemma}

The idea of the proof of the above lemma is similar to Lemma \ref{lemma:atleast2hopsfromv_2}. Next, we discuss the possible knowledge the agents might have gained, when any agent first time visits $v_i$.

Let the round at which at least one agent reaches $v_i$ for the first time be $r^{\circ}_i$, and $r_i$ ($>r^{\circ}_i$) indicates the round at which at least one agent gets destroyed from $v_i$ within 2 hops of $v_i$. The agent can gain the map of the set $V_{exp}$, explored yet. Define $T_{i-1}=\max^{i-1}_{j=1}(r''_j-r_j)$, $r_i-r^{\circ}_i={Wait}_{i}$ and $T_{i-1}-{Wait}_{i}=Time_{i-1}$.

To ensure $v_i$ to be Conflict-Free, the adversary sets the distance between $v_{i-1}$ and $v_i$ along $\mathcal{P}$ as follows: if $Time_{i-1}>0$, and there exists some $l_{i-1}\in\mathbb{N}$ such that $C_{i-1}\cdot \Delta^{l_{i-1}+1}<Time_{i-1}$, where $C_{i-1}=\max\{c_1,c_2,\dots,c_{i-2}\}$, then return $D_{\mathcal{A}}(v_{i-1},v_i)=l_{i-1}+1$ and modify the port-labeling of $v_{i-1}$ using $\lambda^2_{v_{i-1},\mathcal{A}}$, and port label $u^{i-1}_j$ using $\lambda^1_{u^{i-1}_j,\mathcal{A}}$ (for all $j\in\{1,2,\dots,l_{i-1}\})$ such that first agent's $l_j$ distance movement from $v_j$ on $V^c_{exp}$, after round $r_{j-1}$ is to $v_j$. Otherwise, return $D_{\mathcal{A}}(v_{i-1},v_i)=r''_{i-1}+1$, call it $l_{i-1}+1$. Also modify the port-labeling of $v_{i-1}$ using $\lambda^2_{v_{i-1},\mathcal{A}}$, and port label of $\lambda^1_{u^{i-1}_j,\mathcal{A}}$ (for all $j\in\{1,2,\dots,l_{i-1}\})$ such that the first agent's $l_{i-1}$ distance movement after round $r''_{i-1}$ is to $v_i$. In the following lemma, we have shown that $v_i$ is indeed Conflict-Free.

\begin{lemma}\label{lemma:viConflictFree}
    $v_i$ is Conflict-Free.
\end{lemma}

Again the proof of this is similar to the one discussed in Lemma \ref{lemma:v2ConflictFree}.

Next, for all $v_j\in\mathcal{V}_1$ where $j\in\{1,2,\dots,i-1\}$, the agent can know the exact port labelings of the edge $(v_j,\frakb)$. Based on these, we have again two class of choices: first, the agent's use the knowledge of the port labels $(v_i,\frakb)$, for all such $j$, and second, they do not. We call these first class of choices as \textit{Choice-A} and second class of \textit{Choice-B}. We discuss all the possible choices in \textit{Choice-A} class first.

\textit{Choice-A1:} $\mathcal{A}$ may instruct at least 2 agents to explore a neighbor of $v_i$ from $V^c_{exp}$ at round $r_i$ along port $j$, where $r_i>r^{\circ}_i$ and $j\in\{1,2,\dots,\delta_{v_i}\}$.

\textit{Counter:} The counter idea is same as the one explained in Choice-A1 for $v_2$.

\textit{Choice-A2:} Let $r_i$ be the first round, at which a single agent travels a node which is at 2 hop distance of $v_i$ belonging to $V^c_{exp}$.

\textit{Counter:} The counter idea is similar as the one explained in Choice-A2 for $v_2$.

\begin{lemma}\label{lemma:vi-two-destroyed-together}
    If $\mathcal{A}$ instructs at least two agents to move simultaneously to a neighbor of $v_2$ in $V^c_{exp}$, then the adversary can destroy all of these agents, even if $\mathcal{A}$ uses the knowledge of the map of $V_{exp}$ during its execution from $v_2$.
\end{lemma}

    The proof of the above lemma is similar to Lemma \ref{lemma:v2-two-destroyed-together}.

\begin{lemma}\label{lemma:vi-twosuspiciousnode-A}
    As per Counter of Choice-A2, if there exists no $r'_i$, where $r^{\circ}_i<r'_i<r_i$, then adversary can activate $\frakb$, such that after destruction of one agent from $v_i$ within round $r_i$, remaining agents cannot know the exact node from $v_i$ to $\frakb$.
\end{lemma}

The proof of the above lemma is similar to Lemma \ref{lemma:vi-twosuspiciousnode-A}.

Next, we discuss the choices in the class \textit{Choice-B}, i.e., when the agents do not use the knowledge of the port label from $\frakb$ to $v_j$, for any $j\in\{1,2,\dots,i-1\}$.

\textit{Choice-B1:} $\mathcal{A}$ may instruct at least 2 agents to explore a neighbor of $v_i$ from $V^c_{exp}$ at round $r_i$ along port $j$, where $r_i>r^{\circ}_i$ and $j\in \{1,2,\dots,\delta_{v_i}\}$.

\textit{Counter:} The counter idea is similar to the one described for Choice-B1 of $v_2$.

\textit{Choice-B2:} Let $r_i$ be the first round at which a single agent travels a node which is at 2 hop distance of $v_i$ in $V^c_{exp}$.

\textit{Counter:} The counter idea is similar to the one described for Choice-B2 of $v_2$.

\begin{lemma}\label{lemma:vi-twosuspiciousnode-B}
    As per Counter of Choice-B2, if there exists no $r'_i<r_i$, then adversary can activate $\frakb$, such that after destruction of one agent from $v_i$ within round $r_i$, remaining agents cannot know the exact node from $v_2$ which is $\frakb$.
\end{lemma}

The proof of the lemma is similar to Lemma \ref{lemma:v2-twosuspiciousnode-B}.

\begin{corollary}\label{corollary:Exploredsetfromvi}
    Given any graph $\mathcal{G}$ at round $r^{\circ}_i-1$, the set $V_{exp}$ can only contain the nodes $v_j$, $u^j_k$ (where $j\in\{3,\dots,\Delta-1\}$ and $k\in\{1,2,\dots,l_j\}$) along $\mathcal{P}$ of $G\in\mathcal{G}$. 
\end{corollary}

\begin{theorem}\label{theorem:vi-2agentsdestroyed}
    At least 2 agents are destroyed by $\frakb$ from $v_i$.
\end{theorem}

\begin{proof}
    The graph chosen by the adversary from $\mathcal{G}$, satisfies the condition that, $v_i$ is Conflict-Free, i.e., within the interval $[r^{\circ}_i,r_i]$, the adversary need not activate $\frakb$, due to any movement from $v_{j}$ along \texttt{Ext}$_1$ or \texttt{Ext}$_2$ (for all $v_i\in\mathcal{V}_1~\text{or}~v_i\in\mathcal{V}_2$, where $j\in\{1,2,\dots,i-1\}$). Further, Lemmas \ref{lemma:vi-two-destroyed-together}, \ref{lemma:vi-twosuspiciousnode-A} and \ref{lemma:vi-twosuspiciousnode-B}, ensure that, irrespective of the knowledge gained by the agents before round $r^{\circ}_2$, there exists a round $r''_i$ ($\ge r_i$) at which the second agent gets destroyed from $v_i$, while traversing along \texttt{Ext}$_1$ or \texttt{Ext}$_2$ from $v_i$ (depending on $v_i\in\mathcal{V}_1~\text{or}~v_i\in\mathcal{V}_2$). Hence, this proves the theorem.\end{proof}

So, the possible graph class choices eventually reduces to $\mathcal{G}^{\Delta-1}$, where $\mathcal{G}^{\Delta-1}\subset \dots \subset \mathcal{G}^1\subset \mathcal{G}$. 

\begin{remark}\label{remark:v_2inV2}
    If $v_{\Delta}\in\mathcal{V}_2$, then choose $D_{\mathcal{A}}(v_{\Delta-1},v_{\Delta})$, in a similar idea as chosen for $v_i$ (for all $i\in\{3,4,\dots,\Delta-1\}$). Also the choices posed by $\mathcal{A}$ from $v_{\Delta}$ are exactly similar as the ones from $v_{\Delta-1}$, and their adversarial counter measures are also similar. So, using them, we can conclude that not only $v_{\Delta}$ can be Conflict-Free, but also similar to Theorem \ref{theorem:vi-2agentsdestroyed}, we can say that at least 2 agents are destroyed by $\frakb$ from $v_{\Delta}$.

    Otherwise, set $D_{\mathcal{A}}(v_{\Delta-1},v_{\Delta})=2$. So, finally we can conclude that depending on whether $v_{\Delta}\in\mathcal{V}_2$ or $v_{\Delta}\notin\mathcal{V}_2$, the possible graph choices further reduces to $\mathcal{G}^{\Delta}$, where $\mathcal{G}^{\Delta}\subset \mathcal{G}^{\Delta-1}$.
\end{remark}

Finally, to conclude the proof of Theorem~\ref{thm: equivalent statement},
the adversary chooses the graph $G=(V,E,\lambda)$ from $\mathcal{G}^{\Delta}$, where $\mathcal{G}^{\Delta-1}\subset\mathcal{G}^{\Delta-2}\subset\dots \subset \mathcal{G}$, and sets $v_1=home$. Next starting from $v_1$, Theorems \ref{theorem:v1-2agentsdestroyed}, \ref{theorem:v2-2agentsdestroyed} and \ref{theorem:vi-2agentsdestroyed} ensure that 2 agents are destroyed from each $v_i\in V$, where $i\in\{1,2,\dots,\Delta-1\}$. This contradicts the fact that $\mathcal{A}$ solves \pbmPerpExpl~with $2(\Delta-1)$ agents.



\begin{table}[h]
\centering
\small
\begin{tabular}{|c|l|}
\hline
\textbf{Notation} & \textbf{Description} \\
\hline
$\mathcal{G}$ & Indicates the specific graph from which adversary chooses the graph, as per $\mathcal{A}$.\\
\hline

$\mathcal{G}^i$ & Subclass of $\mathcal{G}$, where $\mathcal{G}^i\subset \mathcal{G}^{i-1}\subset\dots\subset \mathcal{G}$.\\
\hline
$\mathcal{P}$ & Indicates the path graph on a graph in $\mathcal{G}$ consisting of vertices of type $v_i$ and $u^i_j$.\\
\hline
$\mathcal{V}$ & indicates the set of vertices $v_1,\dots,v_{\Delta-1}$.\\
\hline
$dist_{\mathcal{P}}(v_{i-1},v_i)$ & indicates the distance between $v_{i-1}$ and $v_i$ along $\mathcal{P}$, also denoted by $l_{i-1}+1$.\\
\hline
\texttt{Ext}$_1$,\texttt{Ext}$_2$ & indicates two variation of subgraphs\\
\hline
$\mathcal{G}_1,\mathcal{G}_2,\mathcal{G}_3$ & the different sub-class under class $\mathcal{G}$.\\
\hline
$w_i$ & it is a vertex connected to $v_i\in\mathcal{V}_2$, where it is of degree $\Delta$ \\
\hline
$T^i_v$ & a tree special tree (explained in Block-3) originating from $v$ along port $i$.\\
\hline
$\lambda^1_{v,\mathcal{A}}$ & function for adversary to port label each vertex, except $\{v_1,v_2,\dots,v_{\Delta}\}$.\\
\hline
$\lambda^2_{v_i,\mathcal{A}}$ & function for adversary to port label each vertex $v_i$, where $i\in\{1,2,\dots,\Delta\}$.\\
\hline
$V_{exp}$, $V^c_{exp}$ & indicates the set of explored and unexplored vertices, such that $V=V_{exp}\cup V^c_{exp}$.\\
\hline
$r^{\circ}_i$ & indicates the round at which any agent first visits $v_i$.\\
\hline
$r_i,r''_i$ & indicates the rounds at which the first and second agent destroyed from $v_i$.\\
\hline
$t_i$ & indicates the interval of rounds between destruction of first agent\\
& and second agent from $v_i$.\\
\hline
$T_j$ & maximum of $t_i$, for all $i\in\{1,2,\dots,j-1\}$.\\
\hline
${Wait}_i$ & indicates the number rounds between any agent first visits $v_i$ \\
& and the first agent (or agents) getting destroyed.\\
\hline
$Time_{i-1}$ & indicates $T_{i-1}-{Wait}_i$\\
\hline
\end{tabular}
\caption{Notations Table for the notations used in Appendix \ref{appendix: lower bound proof General BBH}.}
\end{table}

\section{A perpetual exploration algorithm for general graphs with a BBH}\label{Appendix: General Graph Algorithm}

In this section we discuss the algorithm, termed as \textsc{Graph\_PerpExplore-BBH-Home} that solves \textsc{PerpExploration-BBH-Home} on a general graph, $G=(V,E,\lambda)$. We will show that our algorithm uses at most $3\Delta+3$ agents (where $\Delta$ is the maximum degree in $G$), to solve this problem. The structure of our algorithm depends upon four separate algorithms \textsc{Translate\_Pattern} along with \textsc{Make\_Pattern} (discussed in Appendix \ref{Appendix: Path Alg}), \textsc{Explore} (explained in this section) and \textsc{BFS-Tree-Construction} \cite{chalopin2010constructing}. So, before going in to details of our algorithm that solves \textsc{PerpExploration-BBH-Home}, we recall the idea of \textsc{BFS-Tree-Construction}.

An agent starts from a node $h\in V$ (also termed as \textit{home}), where among all nodes in $G$, only $h$ is marked. The agent performs breadth-first search (BFS) traversal, while constructing a BFS tree rooted at $h$. The agent maintains a set of edge-labeled paths, $\mathcal{P}=\{P_v~:~\text{edge labeled shortest path}$ $ \text{from $h$ to $v$,}~\forall v\in V~\text{such that the agent has visited $v$}\}$ while executing the algorithm. During its traversal, whenever the agent visits a node $w$ from a node $u$, then to check whether the node $w$ already belongs to the current BFS tree of $G$ constructed yet, it traverses each stored edge labeled paths in the set $\mathcal{P}$ from $w$ one after the other, to find if one among them takes it 
to the marked node $h$. If yes, then it adds to its map a cross-edge $(u,w)$. Otherwise, it adds to the already constructed BFS tree, the node $w$, accordingly $\mathcal{P}=\mathcal{P}\cup P_w$ is updated. The underlying data structure of \textsc{Root\_Paths} \cite{chalopin2010constructing} is used to perform these processes. This strategy guarantees as per Proposition 9 of \cite{chalopin2010constructing}, that \textsc{BFS-Tree-Construction} algorithm constructs a map of $G$, in presence of a marked node, within $\mathcal{O}(n^3\Delta)$ steps and using $\mathcal{O}(n\Delta\log n)$ memory, where $|V|=n$ and $\Delta$ is the maximum degree in $G$.

 In our algorithm, we use $k$ agents (in Theorem \ref{theorem:Final3DeltaUpperBound-mainversion}, it is shown that $k=3\Delta+3$ agents are sufficient), where they are initially co-located at a node $h\in V$, which is termed as \emph{home}. Initially, at the start our algorithm asks the agents to divide in to three groups, namely, \texttt{Marker}, \texttt{SG} and \texttt{LG}$_0$, where \texttt{SG} (or smaller group) contains the least four ID agents, the highest ID agent among all $k$ agents, denoted as \texttt{Marker} stays at $h$ (hence $h$ acts as a marked node), and the remaining $k-5$ agents are denoted as \texttt{LG}$_0$ (or larger group). During the execution of our algorithm, if at least one member of \texttt{LG}$_0$ detects one port leading to the BBH from one of its neighbor, in that case at least one member of \texttt{LG}$_0$ settles down at that node, acting as an \textit{anchor} blocking that port which leads to the BBH, and then some of the remaining members of \texttt{LG}$_0$ forms \texttt{LG}$_1$. In general, if at least one member of \texttt{LG}$_i$ detects the port leading to the BBH from one of its neighbors, then again at least one member settles down at that node acting as an \textit{anchor} to block that port leading to the BBH, and some of the remaining members of \texttt{LG}$_i$ forms \texttt{LG}$_{i+1}$, such that |\texttt{LG}$_{i+1}|<$|\texttt{LG}$_{i}$|. It may be noted that, a member of \texttt{LG}$_i$ only settles at a node $v$ (say) acting as an \textit{anchor}, only if no other \textit{anchor} is already present at $v$. Also, only if a member of \texttt{LG}$_i$ settles as an \textit{anchor}, then only some of the members of \texttt{LG}$_i$ forms \texttt{LG}$_{i+1}$.
 
 In addition to the groups \texttt{LG}$_0$ and \texttt{SG}, the \texttt{Marker} agent permanently remains at $h$. In a high-level the goal of our \textsc{Graph\_PerpExplore-BBH-Home} algorithm is to create a situation, where eventually at least one agent blocks, each port of $C_1$ that leads to the BBH (where $C_1,C_2,\dots,C_t$ are the connected components of $G-\frakb$, such that $h\in C_1$), we term these blocking agents as \textit{anchors}, whereas the remaining alive agents must perpetually explore at least $C_1$.

Initially from $h$, the members of \texttt{SG} start their movement, and the members of \texttt{LG}$_0$ stays at $h$ until they find that, none of the members of \texttt{SG} reach $h$ after a certain number of rounds. Next, we explain one after the other how both these groups move in $G$.

\noindent\textbf{Movement of \texttt{SG}}: The members (or agents) in \texttt{SG} works in phases, where in each phase the movement of these agents are based on the algorithms \textsc{Make\_Pattern} and \textsc{Translate\_Pattern} (both of these algorithms are described in Appendix \ref{Appendix: Path Alg}). Irrespective of which, the node that they choose to visit during \textit{making pattern} or \textit{translating pattern} is based on the underlying algorithm \textsc{BFS-Tree-Construction}.

More specifically, the $i$-th phase (for some $i>0$) is divided in two sub-phases: $i_1$-th phase and $i_2$-th phase. In the $i_1$-th phase, the members of \texttt{SG} makes at most $2^i$ translations, while executing the underlying algorithm \textsc{BFS-Tree-Construction}. Next, in the $i_2$-th phase, irrespective of their position after the end of $i_1$-th phase, they start translating back to reach $h$. After they reach $h$ during the $i_2$-th phase, they start $(i+1)$-th phase (which has again, $(i_1+1)$ and $(i_2+1)$ sub-phase). Note that, while executing $i_1$-th phase, if the members of \texttt{SG} reach $h$, in that case they continue executing $i_1$-th phase. We already know as per Appendix \ref{Appendix: Path Alg}, each translation using \textsc{Translate\_Pattern} requires 5 rounds and for creating the pattern using \textsc{Make\_Pattern} it requires 2 rounds. This concludes that, it requires at most ${Ti}_j=5\cdot 2^i+2$ rounds to complete $i_j$-th phase, for each $i>0$ and $j\in\{1,2\}$.

If at any point, along their traversal, the adversary activates the BBH, such that it interrupts the movement of \texttt{SG}. In that scenario, at least one member of \texttt{SG} must remain alive, exactly knowing the position of the BBH from its current node (refer to the discussion of \textbf{Intervention by the BBH} in Appendix \ref{Appendix: Path Alg}). The agent (or agents) which knows the exact location of the BBH, stays at the node adjacent to the BBH, such that from its current node, it knows the exact port that leads to the BBH, or in other words they act as \textit{anchors} with respect to one port, leading to the BBH. In particular, let us suppose, the agent holds the adjacent node of BBH, with respect to port $\alpha$ from BBH, then this agent is termed as \texttt{Anchor($\alpha$)}.

\noindent\textbf{Movement of \texttt{LG}$_0$}: These group members stay at $h$ with \texttt{Marker}, until the members of \texttt{SG} are returning back to $h$ in the $i_2$-th phase, for each $i>0$. If all members of \texttt{SG} do not reach $h$, in the $i_2$-th phase, i.e., within ${Ti}_2$ rounds since the start of $i_2$-th phase, then the members of \texttt{LG}$_0$ start their movement.

Starting from $h$, the underlying movement of the members of \texttt{LG}$_0$ is similar to \textsc{BFS-Tree-Construction}, but while moving from one node to another they do not execute neither \textsc{Make\_Pattern} nor \textsc{Translate\_Pattern}, unlike the members of \texttt{SG}. In this case, if all members of \texttt{LG}$_0$ are currently at a node $u\in V$, then three lowest ID members of \texttt{LG}$_0$ become the explorers, they are termed as $E^0_1$, $E^0_2$ and $E^0_3$ in increasing order of their IDs, respectively. If based on the \textsc{BFS-Tree-Construction}, the next neighbor to be visited by the members of \texttt{LG}$_0$ is $v$, where $v\in N(u)$, then the following procedure is performed by the explorers of \texttt{LG}$_0$, before \texttt{LG}$_0$ finally decides to visit $v$. 

Suppose at round $r$ (for some $r>0$), \texttt{LG}$_0$ members reach $u$, then at round $r+1$ both $E^0_2$ and $E^0_3$ members reach $v$. Next at round $r+2$, $E^0_3$ traverses to the first neighbor of $v$ and returns to $v$ at round $r+3$. At round $r+4$, $E^0_2$ travels to $u$ from $v$ and meets $E^0_1$ and then at round $r+5$ it returns back to $v$. This process iterates for each neighbor of $v$, and finally after each neighbor of $v$ is visited by $E^0_3$, at round $r+4\cdot (\delta_v-1)+1$ both $E^0_2$ and $E^0_3$ returns back to $u$. And in the subsequent round each members of \texttt{LG}$_0$ visit $v$. The whole process performed by $E^0_1$, $E^0_2$ and $E^0_3$ from $u$ is termed as \textsc{Explore}$(v)$, where $v$ symbolizes the node at which the members of \texttt{LG}$_0$ choose to visit from a neighbor node $u$. After the completion of \textsc{Explore}$(v)$, each member of \texttt{LG}$_0$ (including the explorers) visit $v$ from $u$.
\begin{figure}
    \centering
    \includegraphics[width=0.9\linewidth]{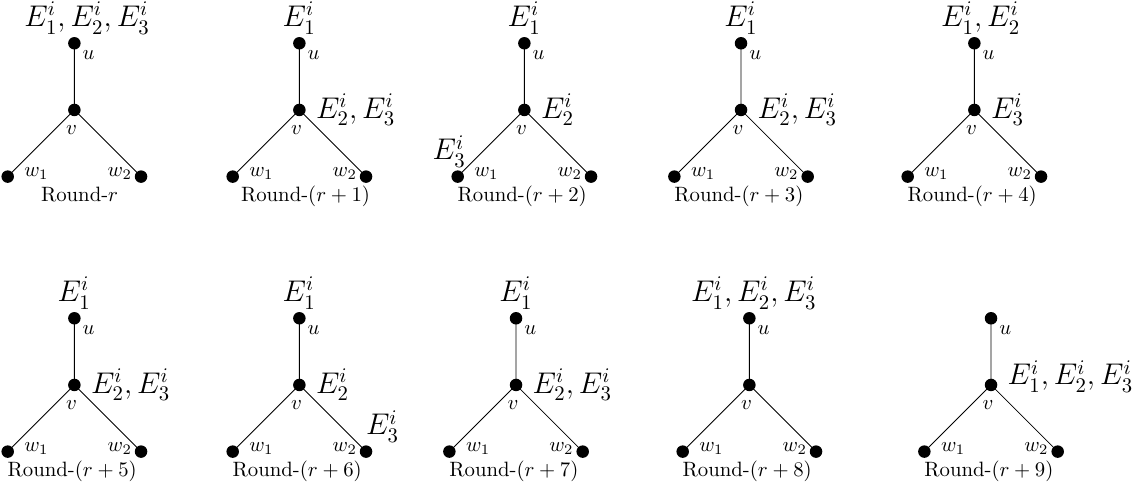}
    \caption{Depicts the round-wise execution of \textsc{Explore}$(v)$ from $u$ by the explorer agents of \texttt{LG}$_i$, for some $i\ge 0$ on the neighbors $w_1$ and $w_2$ of $v$.}
    \label{fig:Explore(v)}
\end{figure}

It may be noted that, if the members of \texttt{LG}$_0$ at the node $u$, according to the \textsc{BFS-Tree-Construction} algorithm, are slated to visit the neighboring node $v$, then before executing \textsc{Explore}$(v)$. The members of \texttt{LG}$_0$ checks if there exists an \textit{anchor} agent blocking that port which leads to $v$. If that is the case, then \texttt{LG}$_0$ avoids visiting $v$ from $u$, and chooses the next neighbor, if such a neighbor exists and no \textit{anchor} agent is blocking that edge. If no such neighbor exists from $u$ to be chosen by \texttt{LG}$_0$ members, then they backtrack to the parent node of $u$, and start iterating the same process.

From the above discussion we can have the following remark.
\begin{remark}
    If at some round $t$, the explorer agents of \texttt{LG}$_i$  (i.e., $E^i_1, E^i_2$ and $E^i_3$), are exploring a two length path, say $P= u \rightarrow v \rightarrow w$, from $u$, then all members of \texttt{LG}$_i$ 
 agrees on $P$ at $t$. This is due to the fact that the agents while executing \textsc{Explore}($v$) from $u$ must follow the path $u \rightarrow
 v$ first. Now from $v$, $E^i_3$ chooses the next port in a particular pre-decided order (excluding the port through which it entered $v$). So, whenever it returns back to $v$ to meet $E^i_2$ after visiting a node $w$, $E^i_2$ knows which port it last visited and which port it will chose next and relay that information back to other agents of \texttt{LG}$_i$ on $u$. So, after $E^i_2$ returns back to $v$ again from $u$ when $E^i_3$ starts visiting the next port, all other agents know about it.
 \label{remark: agreement of path chosen}
\end{remark}

During the execution of \textsc{Explore}$(v)$ from $u$, the agent $E^0_3$ can face one of the following situations:

\begin{itemize}
    \item It can find an \textit{anchor} agent at $v$, acting as \texttt{Anchor}$(\beta)$, for some $\beta\in\{1,\dots,\delta_v\}$. In that case, during its current execution of \textsc{Explore}$(v)$, $E^0_3$ does not visit the neighbor of $v$ with respect to the port $\beta$.
    \item It can find an \textit{anchor} agent at a neighbor $w$ (say) of $v$, acting as \texttt{Anchor}$(\beta')$, where $\beta'\in\{1,\dots, \delta_w\}$. If the port connecting $w$ to $v$ is also $\beta'$, then $E^0_3$ understands $v$ (or its previous node) is the BBH, and accordingly tries to return to $u$, along the path $w\rightarrow v \rightarrow u$, and if it is able to reach to $u$, then it acts as an \textit{anchor} at $u$, with respect to the edge $(u,v)$. On the other hand, if port connecting $w$ to $v$ is not $\beta'$, then $E^0_3$ continues its execution of \textsc{Explore}$(v)$.
\end{itemize}

The agent $E^0_2$ during the execution of \textsc{Explore}$(v)$, can encounter one of the following situations, and accordingly we discuss the consequences that arise due to the situations encountered.

\begin{itemize}
    \item It can find an \textit{anchor} agent at $v$ where the \textit{anchor} agent is not $E^0_3$, in which case it continues to execute \textsc{Explore}$(v)$.
    \item It can find an \textit{anchor} agent at $v$ and finds the \textit{anchor} agent to be $E^0_3$. In this case, $E^0_2$ returns back to $u$, where \texttt{LG}$_1$ is formed, where \texttt{LG}$_1=\texttt{LG}_0\setminus\{E^0_3\}$. Next, the members of \texttt{LG}$_1$ start executing the same algorithm from $u$, with new explorers as $E^1_1$, $E^1_2$ and $E^1_3$.
    \item $E^0_2$ can find that $E^0_3$ fails to return to $v$ from a node $w$ (say), where $w\in N(v)$. In this case, $E^0_2$ understands $w$ to be the BBH, and it visits $u$ in the next round to inform this to remaining members of \texttt{LG}$_0$ in the next round, and then returns back to $v$, and becomes \texttt{Anchor}$(\beta)$, where $\beta\in\{1,\dots,\delta_v\}$ and $\beta$ is the port connecting $v$ to $w$. On the other hand, \texttt{LG}$_0$ after receiving this information from $E^0_2$, transforms to \texttt{LG}$_1$ (where \texttt{LG}$_1=\texttt{LG}_0\setminus\{E^0_2,E^0_3\}$) and starts executing the same algorithm, with $E^1_1$, $E^1_2$ and $E^1_3$ as new explorers.
    
\end{itemize}

Lastly, during the execution of \textsc{Explore}$(v)$ the agent $E^0_1$ can face the following situation. 

\begin{itemize}
    \item $E^0_2$ fails to return from $v$, in this situation $E^0_1$ becomes \texttt{Anchor}$(\beta)$ at $u$, where $\beta$ is the port connecting $u$ to $v$. Moreover, the remaining members of \texttt{LG}$_0$, i.e., \texttt{LG}$_0\setminus\{E^0_1,E^0_2,E^0_3\}$ forms \texttt{LG}$_1$ and they start executing the same algorithm from $u$, with new explorers, namely, $E^1_1$, $E^1_2$ and $E^1_3$, respectively.
\end{itemize}

For each $E^0_1$, $E^0_2$ and $E^0_3$, if they do not face any of the situations discussed above, then they continue to execute \textsc{Explore}$(v)$.

From the above discussions, we conclude that following lemma.

\begin{lemma}
\label{lemma: explore cases}
    Let at round $t$ (for some $t>0$), the explorer agents of \texttt{LG}$_i$ (i.e., $E^i_1,E^i_2$ and $E^i_3$) starts exploring the 2 length path $P=u \rightarrow v\rightarrow w$ from $u$ while executing \textsc{Explore}($v$). Then during the exploration of the path $P$, we have the following results.
    \begin{enumerate}
        \item If any of the explorer agents gets destroyed at $v$  then, there will be at least one explorer agent of \texttt{LG}$_i$ that identifies $v$ as the BBH and the port from $u$ leading to $v$ is the port leading to the BBH. Also it can relay this information to the remaining agents of \texttt{LG}$_i$. 
        \item If any of the explorer agents gets destroyed at $w$  then, there will be at least one explorer agent of \texttt{LG}$_i$ that identifies $w$ as the BBH and the port from $v$ leading to $w$ is the port leading to the BBH. Also it can relay this information to the remaining agents of \texttt{LG}$_i$.
        \item If an explorer agent  meets an \textit{anchor} agent at $w$ pointing the port from $w$ to $v$ as the port leading to the BBH then, there will be at least an explorer agent of \texttt{LG}$_i$ that identifies $v$ as the BBH and the port from $u$ leading to $v$ is the port leading to the BBH. Also it can relay this information to the remaining agents of \texttt{LG}$_i$.
    \end{enumerate}
\end{lemma}

Pictorial description of points 1 and 3 of Lemma \ref{lemma: explore cases} is described in (i) and (ii) of Fig. \ref{fig:Time-Diagram-Explore-on-path-P}, respectively.

It may be noted that, if for some $j>0$, |\texttt{LG}$_j$|<3, then those members perpetually explore $G$ by executing simply \textsc{BFS-Tree-Construction} and not performing \textsc{Explore}(), unlike the movement for the members of \texttt{LG}$_t$, for $0\le t <j$. Whenever during this execution, the member of \texttt{LG}$_j$, encounter \texttt{Anchor}($\beta$) at a node, say $u$, then the member of \texttt{LG}$_j$ avoids choosing the port $\beta$ for its next move.

In addition to that, after the members in \texttt{LG}$_i$ (for some $i\ge 1$), obtain the map of $G$, they perpetually explore each node in $G$ except the BBH, by performing simple BFS traversal, and while performing this traversal, the members of \texttt{LG}$_i$ always avoids the ports, blocked by some \textit{anchor} agent.

Next, we prove the correctness of our algorithm \textsc{Graph\_PerpExplore-BBH-Home}.


\begin{lemma}
    The members of \texttt{SG} perpetually explores $G$ until their movement is intervened by the BBH.
\end{lemma}

\begin{proof}
    The members of \texttt{SG} operates in phases. As stated earlier, each $i$-th phase (for some $i>0$) is sub-divided in to two parts: $i_1$-th phase and $i_2$-th phase. In the $i_1$-th phase, the agent translates up to $2^i$ nodes, and each node that it translates to, is decided by the underlying algorithm \textsc{BFS-Tree-Construction}. In the $i_2$-th phase, no matter whatever the position of the members of \texttt{SG}, they start translating back to the home. By Proposition 9 of \cite{chalopin2010constructing}, it is shown that if an agent follows \textsc{BFS-Tree-Construction} algorithm, then it eventually constructs the map within $\mathcal{O}(n^3\Delta)$ steps with $\mathcal{O}(n\Delta\log n)$ memory, where $|V|=n$ and $\Delta$ is the maximum degree of $G=(V,E)$. This implies that, for a sufficiently large $j$, such that $2^j\ge c n^3\Delta$, for some constant $c$, the members of \texttt{SG} must have obtained the map of $G$ in the $j_1$-th phase, if the BBH does not intervene in the movement of \texttt{SG}, in any phase less than equal $j$-th phase. So, for any subsequent phase after $j$-th phase until the BBH intervenes, the members of \texttt{SG} explores $G$.     
\end{proof}

\begin{lemma}\label{lemma:one member of SG acts as an anchor}
    If BBH intervenes the movement of \texttt{SG}, at least one member of \texttt{SG} acts as an \textit{anchor}. 
\end{lemma}

\begin{proof}
  The movement of the members of \texttt{SG} is based on translation, and in order to translate, they execute the algorithms \textsc{Make\_Pattern} and \textsc{Translate\_Pattern}. If the BBH intervenes during their movement, means the BBH intervenes during translation, or in other words, the BBH intervenes during \textsc{Make\_Pattern} or \textsc{Translate\_Pattern}. We already know, if the BBH intervenes during the execution of either \textsc{Make\_Pattern} or \textsc{Translate\_Pattern}, then at least one agent executing these algorithms, must identify the location of the BBH (as per Appendix \ref{Appendix: Path Alg}). As per the argument in Appendix \ref{Appendix: Path Alg}, if at round $r$ (for some $r>0$), the member determines the position of the BBH, then either it is on BBH at round $r$ in which case it moves to an adjacent node in $C_1$, say $v$, or at round $r$ the member is present at an adjacent node $v'$ ($v'\in C_1~\text{or}~C_j$, $j\neq 1$). In either situation, the member remains at $v$ or $v'$, \textit{anchoring} the edge connecting BBH to $v$ from round $r$ onwards or BBH to $v'$ from round $r'$ onwards.
\end{proof}

Before going to the next section of correctness, we define the notion of \texttt{LG}$_i$ at a node $u$ (for any $u\in V$), if every member that is part of \texttt{LG}$_i$ is at that node $u$.

\begin{lemma}\label{lemma: LG_i cannot be on BBH}
    For any $i\ge0$, \texttt{LG}$_i$ can never be located at the BBH.
\end{lemma}
\begin{proof}
    If possible let there exists a round $t$ (for some $t>0$) and  $i \ge 0$ such that \texttt{LG}$_i$ is located on the BBH $\frakb$ at round $t$. Without loss of generality, let this be the first round when \texttt{LG}$_i$ is on the BBH for any $i$. Now, since initially \texttt{LG}$_0$ was located at $h \in C_1$ ($C_1$ is the component of $G-\frakb$ containing $h$), \texttt{LG}$_i$ must have moved onto $\frakb$ from some vertex $u \in N(\frakb)\cap C_1$. This implies there must exists a round $t'<t$ such that at $t'$, \texttt{LG}$_i$ is at $u$ and $E^i_1, E^i_2$ and $E^i_3$ starts procedure \textsc{Explore}($\frakb$) from $u$. Note that $u$ does not have any agent acting as anchor, otherwise \texttt{LG}$_i$ can never move to $\frakb$ from $u$. Now, since there is at least one neighbor, say $u'$, of $\frakb$ where there is an agent settled as an anchor (as per Lemma~\ref{lemma:one member of SG acts as an anchor}), distance between $u$ and $u'$ must be 2. So when executing \textsc{Explore}($\frakb$), either $E^i_3$ meets with the anchor agent at $u'$ while exploring the path $u\rightarrow \frakb \rightarrow u'$ or, at least one of $E^i_2$ and $E^i_3$ gets destroyed at $\frakb$ while executing \textsc{Explore}($\frakb$). Now, from the results obtained in point 1 and 3 of Lemma~\ref{lemma: explore cases}, there exists an explorer agent that will know that $\frakb$ is the BBH
 and the port, say $\beta$, from $u$ leading to $\frakb$ is the port leading to the BBH. In this case, this agent relays this information to the other members of \texttt{LG}$_i$ located at $u$ and settles at $u$ as \texttt{Anchor}($\beta$). The remaining members of \texttt{LG}$_i$ forms \texttt{LG}$_{i+1}$ and continue the algorithm, avoiding the port $\beta$ from $u$. This contradicts our assumption that \texttt{LG}$_i$ moves to $\frakb$ from $u$. Hence, we can conclude that \texttt{LG}$_i$ can never be located on the BBH for any $i\ge 0$. 
 \end{proof}
\begin{figure}
    \centering
    \includegraphics[width=0.8\linewidth]{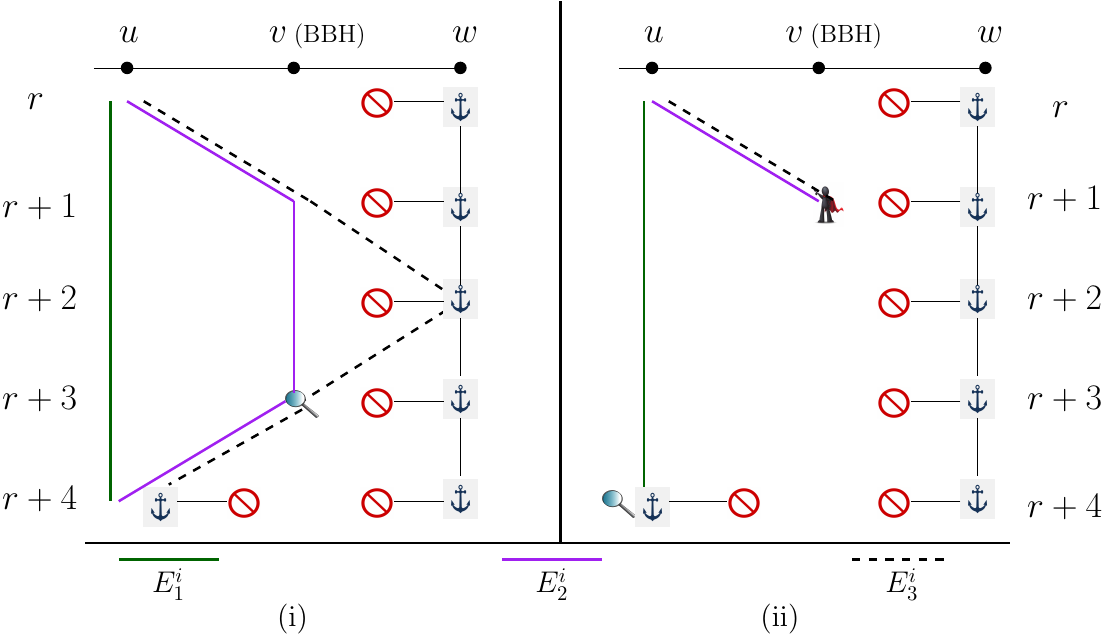}
    \caption{Depicts the time diagram of \textsc{Explore}$(v)$ of \texttt{LG}$_i$ along a specific path $P=u\rightarrow v \rightarrow w$, where $v=\frakb$, in which $w$ contains an \textit{anchor} agent for the edge $(v,w)$. In (i), it depicts even if $\frakb$ is not activated, an explorer agent settles as an \textit{anchor} at $u$ for the edge $(u,v)$. In (ii), it depicts that activation of $\frakb$ destroys both $E^i_2$ and $E^i_3$, then an explorer $E^i_1$ gets settled as \textit{anchor} at $u$ for $(u,v)$.}
    \label{fig:Time-Diagram-Explore-on-path-P}
\end{figure}

 From the above result we can have the following corollary.
\begin{corollary}
    For any $i\ge0$, \texttt{LG}$_i$ can never be located at a node outside $C_1$, where $G-\frakb=C_1\cup\dots \cup C_t$ and $h\in C_1$, $\frakb$ is denoted to be the BBH.
\end{corollary}

\begin{lemma}\label{lemma: U reduces}
    Let $\mathcal{U}\subseteq N(\frakb)\cap C_1$, such that if $u\in\mathcal{U}$, then $u$ does not contain any \textit{anchor}. If $|\mathcal{U}|>0$, then our algorithm \textsc{Graph\_PerpExplore-BBH-Home} ensures that $|\mathcal{U}|$ decreases, eventually.
\end{lemma}

\begin{proof}
    If at time $t$ (for some $t>0$), the members of \texttt{LG}$_i$ (for some $i\ge 0$) reach $u$, then we show that eventually $u$ will not belong to $\mathcal{U}$. As per our algorithm, whenever \texttt{LG}$_i$ reaches a node $u$, then three explorer agents of \texttt{LG}$_i$ perform \textsc{Explore}() on the next neighbor node to be chosen by \texttt{LG}$_i$ to visit from $u$. Suppose, at time $t'>t$, \texttt{LG}$_i$ chooses to visit $\frakb$ from $u$ (since $u\in N(\frakb)$). Hence, before visiting $\frakb$, three explorer agents start \textsc{Explore}$(\frakb)$, and during this execution, as per points 1 and 3 of Lemma \ref{lemma: explore cases}, one explorer agent of \texttt{LG}$_i$ settles down at $u$ acting as an \textit{anchor} for the edge $(u,\frakb)$. So, now $u$ cannot belong to $\mathcal{U}$ (as it violates the definition of the set $\mathcal{U}$), so $\mathcal{U}$ gets modified to $\mathcal{U}\setminus\{u\}$.

    Let at time $t$, we assume $|\mathcal{U}|=c$ for some constant $c$, and we suppose $\mathcal{U}$ does not decrease from $t$ onwards. Now, as per definition of $\mathcal{U}$, it contains all such nodes which are neighbor to $\frakb$ but does not contain any \textit{anchor}. Let $u_1\in V$ be the last node in $\mathcal{U}$, at which an \textit{anchor} got settled at time $t'$ (where $t'<t$) and due to which \texttt{LG}$_{i}$ got formed from \texttt{LG}$_{i-1}$. Now, as per our assumption, as $\mathcal{U}$ does not decrease from $t$ onwards, so \texttt{LG}$_i$ remains \texttt{LG}$_i$. As per our algorithm, it follows that \texttt{LG}$_i$ will visit every node of $G$ (except through the edges, blocked by \textit{anchor}). This means that \texttt{LG}$_i$ will visit $\frakb$ through a neighbor node that does not contain any \textit{anchor}, which cannot happen as per Lemma \ref{lemma: LG_i cannot be on BBH}.
\end{proof}

The corollary follows from the above lemma.

    \begin{corollary}\label{corollary: U becomes empty}
        Our algorithm ensures that eventually, $\mathcal{U}$ becomes $\emptyset$.
    \end{corollary}

We can now prove Theorem~\ref{theorem:Final3DeltaUpperBound-mainversion}, which we recall below:

\ubgeneral*


\begin{proof}
    As per Corollary \ref{corollary: U becomes empty}, we know that eventually $|\mathcal{U}|=\phi$. This means that, every neighbor of $\frakb$, reachable from $h$ contains an \textit{anchor} agent, blocking that edge which leads to $\frakb$. Again, as per the correctness of \textsc{BFS-Tree-Construction} or just BFS traversal (provided the remaining agents know the map of $G$), the remaining agents (i.e., the agents which are alive and not \textit{anchor} or \texttt{Marker}) will visit every reachable node of $G$, except the one's blocked by \textit{anchor}, perpetually. Since, only $\frakb$ is blocked to visit by all the \textit{anchors}, present at all of its neighbors, so this guarantees that $C_1$ will be perpetually visited, hence it solves \textsc{PerpExplore-BBH-Home}.

    Next, we consider just one 2 length path $P=u\rightarrow\frakb\rightarrow w$ (refer to Fig. \ref{fig:Time-Diagram-Explore-on-path-P}, where $u\in \mathcal{U}$ and $w$ contains an \textit{anchor}, blocking the port $(v,\frakb)$. So, whenever \texttt{LG}$_i$ (for some $i\ge 0$) reaches $u$ (it will reach $u$, as per Lemma \ref{lemma: U reduces}), according to our algorithm \textsc{Explore}$(\frakb)$ will start eventually (i.e., before \texttt{LG}$_i$ decides to move to $\frakb$), and it is known as per Lemma \ref{lemma: explore cases}, an explorer will settle as a \textit{anchor} at $u$, surely after they finish visiting the path $P$ (it may happen earlier during the execution of \textsc{Explore}$(\frakb)$ as well). Now, while visiting $P$, in the worst scenario, adversary can activate $\frakb$, the moment $E^i_2$ and $E^i_3$ is at $\frakb$. This means, the remaining explorer $E^i_1$ will become an \textit{anchor} at $u$, for the edge $(u,\frakb)$ and $u$ will not be in $\mathcal{U}$. Now, this situation can occur from each neighbor of $\frakb$, except one, i.e., where already at least one member of \texttt{SG} is settled as an \textit{anchor}. This proves that, $3(\Delta-1)$ agents are required to \textit{anchor}, $\Delta-1$ adjacent ports of $\frakb$ (if $\delta_\frakb=\Delta$), in the worst case. So, except the \texttt{Marker} agent (which remains stationary at $h$), in the worst case $3\Delta+1$ agents are required to \texttt{anchor} each $\Delta$ many adjacent ports of $\frakb$, including members of \texttt{SG}. Now, this does not ensure perpetual exploration, as all alive agents is stationary, either as an \textit{anchor} or \texttt{Marker}. Hence, we require at least $3\Delta+3$ agents to solve \textsc{PerpExplore-BBH-Home}. In addition to this, to execute our algorithm, we use the underlying concept of \textsc{BFS-Tree-Construction}, and this requires each agent to have a memory of $\mathcal{O}(n\Delta\log n)$.  
\end{proof}

The next remark calculates the total time required for \texttt{SG} to translate to each node in $G$.

\begin{remark}
    The movement of \texttt{SG} is performed in phases. So, in $j$-th phase it performs at most $2^j$ translations using the underlying algorithm of \textsc{BFS-Tree-Construction}, and then returns back to $h$ using again at most $2^j$ translations. Each translation further requires 5 rounds. So, $j$-th phase is executed within $\mathcal{O}(2^{3j}\Delta)$ rounds. Hence, to explore $G$ it takes $\sum^{\log n}_{j=1}\mathcal{O}(2^{3j}\Delta)=\mathcal{O}(n^3\Delta)$ rounds.
\end{remark}

The following remark discusses the time required to explore every node of $G$ by \texttt{LG}$_i$ after \texttt{SG} fails to return at the $i$-th phase, for some $i>0$. and \texttt{SG}, respectively.

\begin{remark}\label{remark: Eventual time for placing anchor}
    Let $v$ be a node, and $w$ be a neighbor of $v$ with degree $\Delta$. It must be observed that, to perform \textsc{Explore}$(w)$ by \texttt{LG}$_i$ from $u$ (for some $i\ge 0$), and eventually reach $w$, it takes $4(\Delta-1)+2$ rounds, as to explore a neighbor of $w$ it takes 4 rounds and except the parent, all neighbors must be explored, and finally from $v$ to reach $w$ it takes another 2 rounds. As discussed earlier, in general to perform \textsc{BFS-Tree-Construction}, it takes $\mathcal{O}(n^3\Delta)$ rounds to visit each node of $G$, and construct a map of $G$.
    
    So, for \texttt{LG}$_i$ to visit each node in $G$, and accordingly construct a map also during the course of which \textit{anchor} each port in $C_1$, \texttt{LG}$_i$ takes $(4(\Delta-1)+2)\cdot \mathcal{O}(n^3\Delta)=\mathcal{O}(n^3\Delta^2)$ rounds.     
\end{remark}

\section{Perpetual exploration in presence of a Black Hole}\label{Appendix: PerpExploration-BH}

In this section, we prove that if the BBH acts as a classical black hole (or BH), then any algorithm requires at least $\Delta+1$ agents  to perform perpetual exploration of at least one component, in an arbitrary graph. We can call this problem \textsc{PerpExploration-BH}, as instead of a BBH, the underlying graph contains a BH. Further, we discuss an algorithm that solves \textsc{PerpExploration-BH} with $\Delta+2$ initially co-located agents.

For a fixed $\Delta\geq 4$, we construct an (unlabeled) graph~$G_\Delta=(V_\Delta,E_\Delta)$, of maximum degree~$\Delta$, which contains (referring to Figure~\ref{fig:lowerbound-GeneralGraph}) a single vertex $v$ of degree $\Delta$ (i.e., the maximum degree in $G$), $\Delta$ many vertices of degree 4 (refer to the vertices $u_i$, for $i\in\{0,1,\dots,\Delta-1\}$), and $\Delta$ many vertices of degree 1 (refer to the vertices $w_i$, for $i\in\{0,1,\dots,\Delta-1\}$).

\begin{figure}
    \centering
    \includegraphics[width=0.6\linewidth]{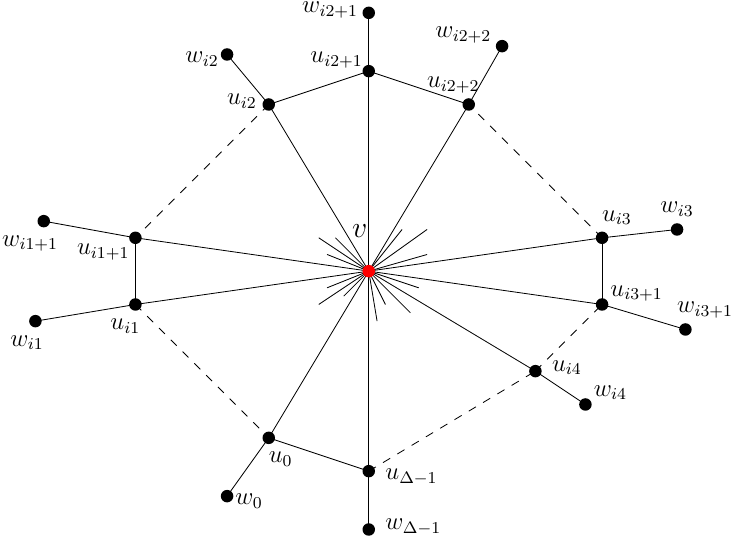}
    \caption{Depicts a graph $G=(V,E,\lambda)$, with the vertex $v$ being the BBH.}
    \label{fig:lowerbound-GeneralGraph}
\end{figure}

\begin{theorem}
    For every~$\Delta\geq 4$, there exists a port-labeling of graph~$G_\Delta$ such that any algorithm using at most~$\Delta$ agents fails to solve \textsc{PerpExploration-BH} in~$G_\Delta$.  
\end{theorem}

\begin{proof}
   Consider the graph $G=(V,E,\lambda)$ explained in Fig. \ref{fig:lowerbound-GeneralGraph}, and let $u_0$ be \emph{home}, where initially the agents $A=\{a_1,\dots,a_{\Delta}\}$ are co-located, consider $v$ to be the BH. We shall prove this theorem by contradiction. Suppose, there exists an algorithm $\mathcal{A}$ that solves \textsc{PerpExploration-BH} with $\Delta$ agents, starting from $u_0$. It may be noted that, the agents do not have the knowledge of the port-labeling of~$G_\Delta$.

   Since, in the worst case, to gather the map of $G$, the algorithm $\mathcal{A}$ must instruct the agents to visit every node in $V$. Now, for each $u_i\in V$, whenever $u_i$ is visited by an agent for the first time, say at round $r_i$, the agent cannot identify the ports from $u_i$ leading to $w_i$ or $v$ or $u_{i+1}$. So, the adversary, based on which, has returned the function $\lambda_{u_i,\mathcal{A}}$ for each $i\in\{0,1,\dots,\Delta-1\}$, such that, let $r'_i$ be the first round after $r_i$ at which any agent (say, $a_j$) from $u_i$ decides to visit an unexplored port, then $a_j$ reaches $v$ at round $r'_i+1$, based on $\lambda_{u_i,\mathcal{A}}$. Since, $v$ is the BH so $a_j$ gets destroyed along the edge $(u_i,v)$. This phenomenon can occur for each $u_i\in V$, for all $i\in\{0,1,\dots,\Delta-1\}$. This shows that $\Delta$ agents get destroyed, and leads to a contradiction for the algorithm $\mathcal{A}$.\end{proof}

   The following remark discusses an algorithm that solves \textsc{PerpExploration-BH} with $\Delta+2$ agents.

\begin{remark}\label{remark: Delta+2 algorithm for BH}
    A slight modification of \textsc{BFS-Tree-Construction} algorithm can generate a $\Delta+2$ algorithm for \textsc{PerpExploration-BH}. A brief idea of which is as follows: an agent stays at \emph{home}, acting as a \texttt{Marker}. Among the remaining agents, one agent is selected to be an \texttt{Explorer}. Each agent maintains a set of edge-labeled paths, $\mathcal{P}$, as explained in \textsc{BFS-Tree-Construction}. During their traversal, the agents perform a \textit{cautious} walk, i.e., the \texttt{Explorer} agent is the one to visit a new node, and after \texttt{Explorer} agent finds it safe, then only remaining agents visits the new node. After the agents visit a new node, say $v$, they check whether $v$ is already visited or not. To check this, \texttt{Explorer} remains at $v$ whereas other agents reaches back to \texttt{Marker}. Next, they start visiting all the paths stored in $\mathcal{P}$, one after the other. If any of these paths lead to $v$ (where the \texttt{Explorer} is present), they update the edge, say $(u,v)$, taken by the \texttt{Explorer} to reach $v$, as a cross edge in the map. Otherwise, if no path leads to the \texttt{Explorer}, then $v$ is added to the map, and $\mathcal{P}=\mathcal{P}\cup P_v$. Accordingly, the remaining agents, except \texttt{Marker}, reach back to the \texttt{Explorer}. This process repeats. 
    
    Now, suppose $v$ is already visited, and there exists a port $\rho$, which leads to the black hole. So, whenever \texttt{Explorer} visits the node with respect to $\rho$, it gets destroyed. Remaining agents at $v$, update in the map that, with respect to $v$, the port $\rho$ leads to the black hole. Accordingly, a new \texttt{Explorer} is elected, and process continues. In future during this process, if ever the agents again reach $v$ (which they understand via the stored edge-labeled paths in $\mathcal{P}$), through any cross edge not already present in the map constructed yet, then the agent's do not take the port $\rho$ from $v$ again. This shows that, at most $\Delta$ agents are destroyed, one agent remains alive, except the \texttt{Marker}, which has the knowledge of all the ports leading to the black hole, and can perpetually explore the current component of $G-\text{BH}$. So, a $\Delta+2$ agent \textsc{PerpExploration-BH} algorithm can be formulated, using this strategy. 
\end{remark}

\end{document}